\documentclass[aps,amssymb,showkeys]{revtex4}

\usepackage{array}
\usepackage{scalerel}
\usepackage{tabularx}
\usepackage{multirow}
\usepackage{amsmath}
\usepackage{amsthm}
\usepackage{graphicx} 
\usepackage{amssymb}
\usepackage[pdftex, bookmarks, colorlinks=true, plainpages = false, citecolor = red, urlcolor = blue, filecolor = blue]{hyperref}
\usepackage[usenames,dvipsnames]{color}
\usepackage{mathrsfs}
\usepackage{enumerate}
\usepackage{diagbox}
\usepackage{subeqnarray}
\usepackage{stmaryrd}

\newcommand{\red}{\textcolor{red}}
\newcommand{\be}{\begin{equation}}
\newcommand{\ee}{\end{equation}}
\newcommand{\bea}{\begin{eqnarray}}
\newcommand{\eea}{\end{eqnarray}}
\newcommand{\bean}{\begin{eqnarray*}}
\newcommand{\eean}{\end{eqnarray*}}

\theoremstyle{plain}
\newtheorem{theorem}{Theorem}

\newtheorem{cor}[theorem]{Corollary}
\newtheorem{lem}[theorem]{Lemma}
\newtheorem{fact}[theorem]{Fact}

\newtheorem{conj}[theorem]{Conjecture}

\theoremstyle{definition}
\newtheorem{defn}[theorem]{Definition}

\def\clap#1{\hbox to 0pt{\hss#1\hss}}
\def\mathclap{\mathpalette\mathclapinternal}

\def\mathclapinternal#1#2{%
\clap{$\mathsurround=0pt#1{#2}$}}

\allowdisplaybreaks

\begin{document}
\title{A solution space for a system of null-state partial differential equations IV}

\date{\today}

\author{Steven M. Flores}
\email{steven.flores@helsinki.fi} 
\affiliation{Department of Mathematics \& Statistics, University of New Hampshire, Durham, New Hampshire, 03824,\\
and \\
Department of Mathematics \& Statistics, University of Helsinki, P.O. Box 68, 00014, Finland}

\author{Peter Kleban}
\email{kleban@maine.edu} 
\affiliation{LASST and Department of Physics \& Astronomy, University of Maine, Orono, Maine, 04469-5708, USA}

\begin{abstract}  
This article is the last of four that completely and rigorously characterize a solution space $\mathcal{S}_N$ for a homogeneous system of $2N+3$ linear partial differential equations (PDEs) in $2N$ variables that arises in conformal field theory (CFT) and multiple Schramm-L\"owner evolution (SLE$_\kappa$). The system comprises $2N$ null-state equations and three conformal Ward identities that govern CFT correlation functions of $2N$ one-leg boundary operators.  In the first two articles \cite{florkleb, florkleb2}, we use methods of analysis and linear algebra to prove that $\dim\mathcal{S}_N\leq C_N$, with $C_N$ the $N$th Catalan number.  Using these results in the third article \cite{florkleb3}, we prove that $\dim\mathcal{S}_N=C_N$ and $\mathcal{S}_N$ is spanned by (real-valued) solutions constructed with the Coulomb gas (contour integral) formalism of CFT.

In this article, we use these results to prove some facts concerning the solution space $\mathcal{S}_N$. First, we show that each of its elements equals a sum of at most two distinct Frobenius series in powers of the difference between two adjacent points (unless $8/\kappa$ is odd, in which case a logarithmic term may appear).  This establishes an important element in the operator product expansion (OPE) for one-leg boundary operators, assumed in CFT.  We also identify particular elements of $\mathcal{S}_N$, which we call connectivity weights, and exploit their special properties to conjecture a formula for the probability that the curves of a multiple-SLE$_\kappa$ process join in a particular connectivity.  This leads to new formulas for crossing probabilities of critical lattice models inside polygons with a free/fixed side-alternating boundary condition, which we derive in \cite{fkz}.  Finally, we propose a reason for why  the {\it exceptional speeds} (certain $\kappa$ values that appeared in the analysis of the Coulomb gas solutions in \cite{florkleb3}) and the minimal models of CFT are connected.

\end{abstract}

\maketitle

\section{Introduction}\label{intro}

This article completes the analysis begun in \cite{florkleb,florkleb2,florkleb3}.  In this introduction, we state the problem under consideration and summarize the results from \cite{florkleb,florkleb2,florkleb3}.  The introduction \red{I} and appendix \red{A} of \cite{florkleb} explain the origin of this problem in conformal field theory (CFT) \cite{bpz,fms,henkel}, its relation to multiple Schramm-L\"owner evolution (SLE$_\kappa$) \cite{bbk,dub2,graham,kl,sakai}, and its application \cite{bpz,bbk,bauber,dots,gruz,rgbw,c3,c1} to critical lattice models \cite{grim,bax,wu,fk,stan} and some random walks \cite{zcs,law1,schrsheff,weintru,madraslade}.

The goal of this article and its predecessors \cite{florkleb, florkleb2, florkleb3} is to completely and rigorously determine a certain solution space $\mathcal{S}_N$ of the following system of $2N$ null-state partial differential equations (PDEs) of CFT,
\be\label{nullstate}\Bigg[\frac{\kappa}{4}\partial_j^2+\sum_{k\neq j}^{2N}\left(\frac{\partial_k}{x_k-x_j}-\frac{(6-\kappa)/2\kappa}{(x_k-x_j)^2}\right)\Bigg]F(\boldsymbol{x})=0,\quad j\in\{1,2,\ldots,2N\},\ee
and three conformal Ward identities from CFT,
\be\label{wardid}\sum_{k=1}^{2N}\partial_kF(\boldsymbol{x})=0,\quad \sum_{k=1}^{2N}\left[x_k\partial_k+\frac{(6-\kappa)}{2\kappa}\right]F(\boldsymbol{x})=0,\quad \sum_{k=1}^{2N}\left[x_k^2\partial_k+\frac{(6-\kappa)x_k}{\kappa}\right]F(\boldsymbol{x})=0,\ee
with $\boldsymbol{x}:=(x_1,x_2,\ldots,x_{2N})$ and $\kappa\in(0,8)$.  (Here and in \cite{florkleb3}, but unlike in \cite{florkleb, florkleb2}, we refer to the coordinates of $\boldsymbol{x}$ as ``points.")  The solution space $\mathcal{S}_N$ (over the reals) comprises all (classical) solutions $F:\Omega_0\rightarrow\mathbb{R}$, where
\be\label{components}\Omega_0:=\{\boldsymbol{x}\in\mathbb{R}^{2N}\,|\,x_1<x_2<\ldots< x_{2N-1}< x_{2N}\},\ee
such that for each $F\in\mathcal{S}_N$, there exist positive constants $C$ and $p$ such that
\be\label{powerlaw} |F(\boldsymbol{x})|\leq C\prod_{i<j}^{2N}|x_j-x_i|^{\mu_{ij}(p)},\quad\text{with}\quad\mu_{ij}(p):=\begin{cases}-p, & |x_j-x_i|<1 \\ +p, & |x_j-x_i|\geq1\end{cases}\quad\text{for all $\boldsymbol{x}\in\Omega_0.$}\ee
(We use this bound to prove many of the results in \cite{florkleb, florkleb2}.)  Restricting our attention to $\kappa\in(0,8)$, our goals for these four articles are as follows:
\begin{enumerate}
\item\label{item1} Rigorously prove that $\mathcal{S}_N$ is spanned by real-valued Coulomb gas solutions.
\item\label{item2} Rigorously prove that $\dim\mathcal{S}_N=C_N$, with $C_N=(2N)!/N!(N+1)!$ the $N$th Catalan number.
\item\label{item3} Argue that $\mathcal{S}_N$ has a basis $\mathscr{B}_N$ consisting of $C_N$ \emph{connectivity weights} (physical quantities defined in the introduction \red{I} to \cite{florkleb}) and find formulas for all of the connectivity weights.
\end{enumerate}

In \cite{florkleb,florkleb2,florkleb3}, we use certain elements of the dual space $\mathcal{S}_N^*$ to achieve goals \ref{item1} and \ref{item2}, and in this article, we use these linear functionals again to complete item \ref{item3}, among other things.  To construct these linear functionals $\mathscr{L}:\mathcal{S}_N\rightarrow\mathbb{R}$, we prove in \cite{florkleb} that for all $F\in\mathcal{S}_N$ and all $i\in\{1,2,\ldots,2N-1\}$, the limit
\be\label{lim}\bar{\ell}_1F(x_1,x_2,\ldots,x_i,x_{i+2},\ldots,x_{2N}):=\lim_{x_{i+1}\rightarrow x_i}(x_{i+1}-x_i)^{6/\kappa-1}F(\boldsymbol{x})\ee
exists, is independent of $x_i$, and (after implicitly taking the trivial limit $x_i\rightarrow x_{i-1}$) is an element of $\mathcal{S}_{N-1}$.  Then, we let $\mathscr{L}$ be a composition of $N$ such limits.  These functionals naturally gather into equivalence classes $[\mathscr{L}]$ whose elements differ only by the order in which we take their limits.

For convenience, we represent every equivalence class $[\mathscr{L}]$ by a unique \emph{half-plane diagram} consisting of $N$ non-intersecting curves, called \emph{interior arcs}, in the upper half-plane, with the endpoints of each arc brought together by a limit in every element of $[\mathscr{L}]$.  Alternatively, we represent $[\mathscr{L}]$ by its \emph{polygon diagram}, which is its half-plane diagram continuously mapped onto the interior of a regular polygon $\mathcal{P}$, with arc endpoints sent to vertices.  We call either the \emph{diagram for $[\mathscr{L}]$}.  There are $C_N$ such diagrams, and they correspond one-to-one with the available equivalence classes (figure \ref{Csls}).  We enumerate the equivalence classes $[\mathscr{L}_1]$, $[\mathscr{L}_2],\ldots,[\mathscr{L}_{C_N}]$, let $\mathscr{B}_N^*:=\{[\mathscr{L}_1],[\mathscr{L}_2,]\ldots,[\mathscr{L}_{C_N}]\}\subset\mathcal{S}_N^*$, and define for each $\varsigma\in\{1,2,\ldots,C_N\}$ the \emph{$\varsigma$th connectivity} as the arc connectivity exhibited by the diagram for $[\mathscr{L}_{\varsigma}]$.

\begin{figure}[t]
\centering
\includegraphics[scale=0.3]{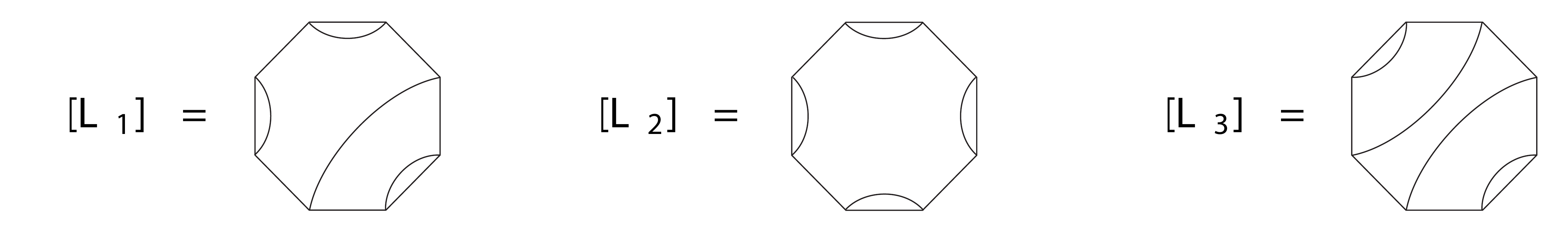}
\caption{Polygon diagrams for three different equivalence classes of allowable sequences of $N=4$ limits.  We find the other $C_4-3=11$ diagrams by rotating one of these three.}
\label{Csls}
\end{figure}

\begin{figure}[b]
\centering
\includegraphics[scale=0.28]{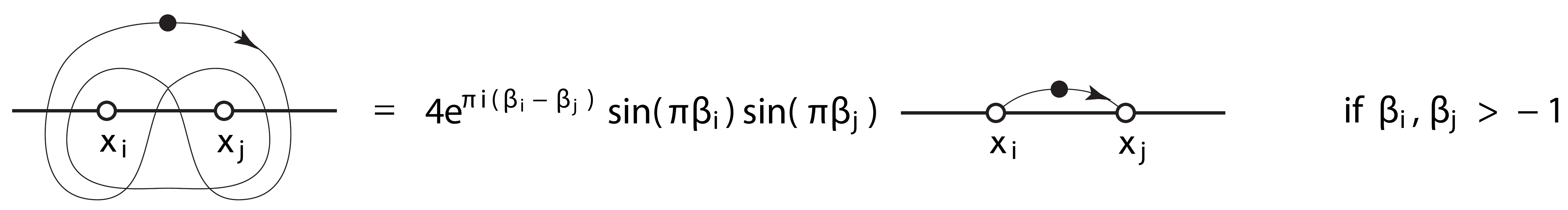}
\caption{The Pochhammer contour $\mathscr{P}(x_i,x_j)$ (left) with ``endpoints" $x_i$ and $x_j$.  If $e^{2\pi i\beta_i}$ and $e^{2\pi i\beta_j}$ are the monodromy factors associated with $x_i$ and $x_j$ respectively, and $\beta_i,\beta_j>-1$, then we may replace $\mathscr{P}(x_i,x_j)$ with the simple contour on the right.}
\label{PochhammerContour}
\end{figure}

\begin{figure}[t]
\centering
\includegraphics[scale=0.3]{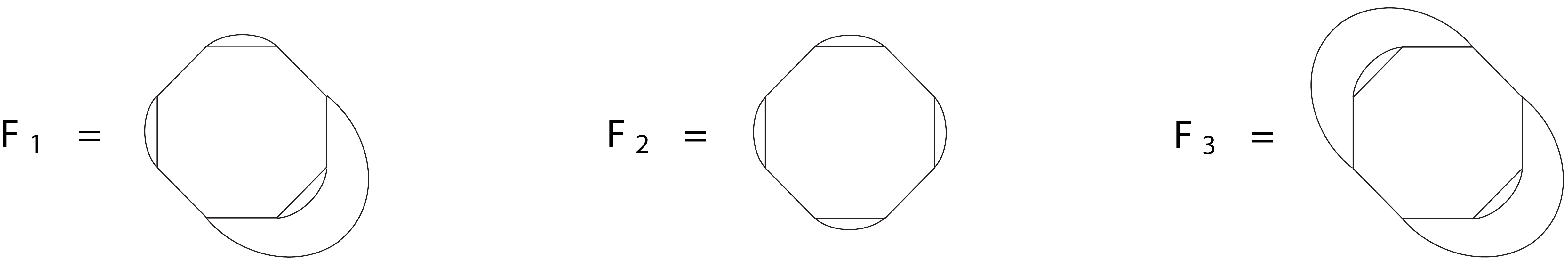}
\caption{Polygon diagrams for three different elements of $\mathcal{B}_4$.  We find the other $C_4-3=11$ diagrams by rotating one of these three.}
\label{Fk}
\end{figure}

We conclude our analysis in \cite{florkleb} by proving that the linear map $v:\mathcal{S}_N\rightarrow\mathbb{R}^{C_N}$ with $v(F)_\varsigma:=[\mathscr{L}_\varsigma]F$ is well-defined and injective, so $\dim\mathcal{S}_N \leq C_N$.
With this bound established, we achieve goals \ref{item1} and \ref{item2} next in \cite{florkleb3}.  For this, we use the CFT Coulomb gas (contour integral) formalism \cite{df1,df2} to construct a set 
\be\label{BN}\mathcal{B}_N:=\{\mathcal{F}_1,\mathcal{F}_2,\ldots,\mathcal{F}_{C_N}\}\subset\mathcal{S}_N\ee
of $C_N$ elements of $\mathcal{S}_N$.  According to corollary \red{9} of \cite{florkleb3}, the function $\mathcal{F}_\vartheta$ may be given explicitly by any one of $2N$ equivalent formulas, labeled $\mathcal{F}_{c,\vartheta}$ with $c\in\{1,2,\ldots,2N\}$.  These formulas are (see definition \red{4} and figure \red{5} of \cite{florkleb3})
\begin{multline}\label{Fexplicit}\mathcal{F}_{c,\vartheta}(\kappa\,|\,\boldsymbol{x})=n(\kappa)\left[\frac{n(\kappa)\Gamma(2-8/\kappa)}{4\sin^2(4\pi/\kappa)\Gamma(1-4/\kappa)^2}\right]^{N-1}\Bigg(\prod_{\substack{j<k \\ j,k\neq c}}^{2N}(x_k-x_j)^{2/\kappa}\Bigg)\Bigg(\prod_{\substack{k=1 \\ k\neq c}}^{2N}|x_c-x_k|^{1-6/\kappa}\Bigg)\overbrace{\oint_{\Gamma_{N-1}}{\rm d}u_{N-1}\dotsm}^{\mathcal{J}(\boldsymbol{x})}\\ 
\underbrace{\dotsm\oint_{\Gamma_2}{\rm d}u_2\,\,\oint_{\Gamma_1}{\rm d}u_1\,\,\mathcal{N}\Bigg[\Bigg(\prod_{\substack{l=1 \\ l\neq c}}^{2N}\prod_{m=1}^{N-1}(x_l-u_m)^{-4/\kappa}\Bigg)\Bigg(\prod_{m=1}^{N-1}(x_c-u_m)^{12/\kappa-2}\Bigg)\Bigg(\prod_{p<q}^{N-1}(u_p-u_q)^{8/\kappa}\Bigg)\Bigg]}_{\mathcal{J}(\boldsymbol{x})},\end{multline}
where $n(\kappa):=-2\cos(4\pi/\kappa)$ is the loop-fugacity function of the O$(n)$ model \cite{gruz, rgbw, smir4,smir}, $\Gamma_m$ is a Pochhammer contour $\mathscr{P}(x_i,x_j)$ (figure \ref{PochhammerContour}) that shares its ``endpoints" $x_i$ and $x_j$ with the $m$th arc in the diagram for $[\mathscr{L}_\vartheta]$, and no contour shares its endpoints with the $N$th arc, which has an endpoint at $x_c$.  Borrowing terminology from the Coulomb gas formalism, we call this exceptional point $x_c$ \emph{the point bearing the conjugate charge.}

We bring attention to some other details concerning this formula (\ref{Fexplicit}).  First, we call the multiple-contour integral $\mathcal{J}(\boldsymbol{x})$ appearing in (\ref{Fexplicit}) a \emph{Coulomb gas} (or \emph{Dotsenko-Fateev}) integral, and the symbol $\mathcal{N}$ selects the branch of the logarithm for each power function in its integrand so $\mathcal{F}_{c,\vartheta}$ is real-valued for $\kappa>0$.  (See appendix \red{B} of \cite{florkleb3}.)  In \cite{florkleb3}, we show that $\mathcal{F}_{c,\vartheta}$ is an analytic function of $(\kappa,\boldsymbol{x})\in[(0,8)\times i\mathbb{R}]\times\Omega_0$ and that if $\kappa>4$, then we may simplify (\ref{Fexplicit}) by replacing each Pochhammer contour $\mathscr{P}(x_i,x_j)$ with a simple contour $[x_i,x_j]$ bent into the upper half-plane and dropping the factors of $4\sin^2(4\pi/\kappa)$ in the prefactor (figure \ref{PochhammerContour}).  Finally, we may generate other elements of $\mathcal{S}_N$ from (\ref{Fexplicit}) by replacing the contours described beneath this formula with any collection of closed nonintersecting contours \cite{dub}.  We call these solutions \emph{Coulomb gas functions} and linear combinations of them \emph{Coulomb gas solutions} \cite{florkleb3}.

\begin{figure}[b]
\centering
\includegraphics[scale=0.3]{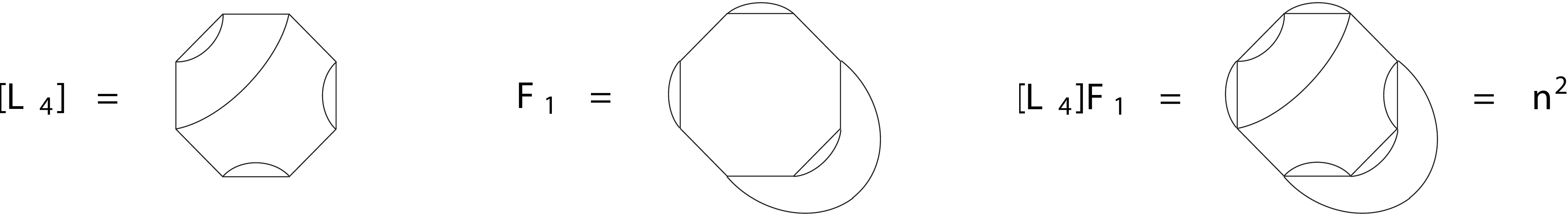}
\caption{The diagram for $[\mathscr{L}_4]\in\mathscr{B}_4^*$, for $\mathcal{F}_1\in\mathcal{B}_4$, and for their product $[\mathscr{L}_4]\mathcal{F}_1\in\mathbb{R}$.  The product diagram contains two loops and therefore evaluates to $n^2$.}
\label{innerproduct}
\end{figure}

If the set $\mathcal{B}_N$ (\ref{BN}) is linearly independent, then it follows from the bound $\dim\mathcal{S}_N\leq C_N$ that the statements of goals \ref{item1} and \ref{item2} above are indeed true.  Hence, we determine the rank of $\mathcal{B}_N$ in \cite{florkleb3}.  To do this, we send each of its elements to a vector $v\in\mathbb{R}^{C_N}$ via the injective linear map $v:\mathcal{S}_N\rightarrow\mathbb{R}^{C_N}$ with $v(F)_\varsigma:=[\mathscr{L}_\varsigma]F$ and show that the square matrix whose columns are the vectors $v(\mathcal{F}_1)$, $v(\mathcal{F}_2),\ldots,v(\mathcal{F}_{C_N})$ has a non-vanishing determinant.  To facilitate this calculation, we invoke the \emph{polygon (resp.\ half-plane) diagram for $\mathcal{F}_\vartheta$} (or more simply, the \emph{diagram for $\mathcal{F}_\vartheta$}), which we define as the diagram for $[\mathscr{L}_{\vartheta}]$, but with all interior arcs replaced by \emph{exterior arcs} drawn outside the $2N$-sided polygon (figure \ref{Fk}) (resp.\ in the lower half-plane).  Then the main result (\red{49}) of \cite{florkleb3} is
\be\label{LkFk}[\mathscr{L}_\varsigma]\mathcal{F}_\vartheta(\kappa)=n(\kappa)^{l_{\varsigma,\vartheta}},\quad \varsigma,\vartheta\in\{1,2,\ldots,C_N\},\quad n(\kappa):=-2\cos(4\pi/\kappa),\ee
with $l_{\varsigma,\vartheta}$ the number of loops in the product diagram for $[\mathscr{L}_\varsigma]\mathcal{F}_\vartheta$ (with the polygon deleted), shown in figure \ref{innerproduct}.  The $C_N\times C_N$ matrix $(M_N\circ n)(\kappa)$ whose $(\varsigma,\vartheta)$th entry is (\ref{LkFk}) is called the \emph{meander matrix} \cite{fgg, fgut, difranc, franc}, and its zeros satisfy
\be\label{thezeros}n(\kappa)=n_{q,q''}\quad\text{for some $n_{q,q''}:=-2\cos(\pi q''/q)$ with $q,q''\in\mathbb{Z}^+$ coprime and $q''<q\leq N+1$.}
\ee
Thus, we conclude that $\mathcal{B}_N$ is linearly independent if and only if $\kappa$ is not a solution of (\ref{thezeros}).  The positive solutions of $n(\kappa)=n_{q,q''}$ are what we call \emph{exceptional speeds}, that is
\be\label{exceptional}\kappa_{q,q'}:=4q/q',\quad\text{with $q>1$ and $\{q,q'\}$  a pair of coprime positive integers.}\ee
We note that the exceptional speeds are really the positive rational speeds, excluding those of the form $4/r$ for some $r\in\mathbb{Z}^+$.  Actually, interesting behavior occurs at all rational speeds $\kappa\in(0,8)$.  Table \ref{tablespeeds} shows the various possibilities.

From these results, we achieve goals \ref{item1} and \ref{item2} for $\kappa$ not an exceptional speed (\ref{exceptional}) with $q\leq N+1$.  Furthermore, if $\kappa$ is such a speed, then we use $\mathcal{B}_N$ to construct a different linearly independent set $\mathcal{B}_N^{\scaleobj{0.75}{\bullet}}$ of $C_N$ elements of $\mathcal{S}_N$ in \cite{florkleb3}, again achieving goals \ref{item1} and \ref{item2}.  We summarize these results in this theorem (previously stated as theorem \red{8} in \cite{florkleb3}).
\begin{theorem}\label{maintheorem} Suppose that $\kappa\in(0,8)$.  Then the following are true.
\begin{enumerate}
\item\label{firstitem} $\mathcal{B}_N$ is a basis for $\mathcal{S}_N$ if and only if $\kappa$ is not an exceptional speed (\ref{exceptional}) with $q\leq N+1$.
\item\label{seconditem}  $\dim\mathcal{S}_N=C_N,$ with $C_N=(2N)!/N!(N+1)!$ the $N$th Catalan number.
\item\label{thirditem} $\mathcal{S}_N$ has a basis consisting entirely of real-valued Coulomb gas solutions.
\item\label{fourthitem} The map $v:\mathcal{S}_N\rightarrow\mathbb{R}^{C_N}$ with $v(F)_\varsigma:=[\mathscr{L}_\varsigma]F$ is a vector-space isomorphism.
\item\label{fifthitem} $\mathscr{B}_N^*:=\{[\mathscr{L}_1],[\mathscr{L}_2],\ldots,[\mathscr{L}_{C_N}]\}$ is a basis for $\mathcal{S}_N^*$.
\end{enumerate}\end{theorem}

\begin{table}[t]
\centering
\begin{tabular}{llllll}
SLE$_\kappa$ speed \hspace{.75cm} & exceptional \hspace{.1cm} & $c(\kappa)$ (\ref{ck}) a central charge \hspace{.2cm} & indicial power of Frobenius  \hspace{.2cm} & Log term in \hspace{.4cm} & all elements of \\
$\kappa\in(0,8)$ & speed & of a CFT minimal model & series differ by an integer & $\psi_1\times\psi_1$ OPE & $\mathcal{B}_N$ algebraic \\
\hline
$8/\kappa\in2\mathbb{Z}^+$ & $\qquad\times$ & $\quad\qquad\qquad\times$ & $\qquad\qquad\qquad$\checkmark & $\qquad\times$ & $\qquad$\checkmark \\
$8/\kappa\in2\mathbb{Z}^++1$ & $\qquad$\checkmark & $\quad\qquad\qquad$\checkmark &  $\qquad\qquad\qquad$\checkmark & $\qquad$\checkmark & $\qquad$\checkmark \\
$\kappa=\kappa_{q,q'}$, $q>2$ & $\qquad$\checkmark & $\quad\qquad\qquad$\checkmark & $\qquad\qquad\qquad\times$ & $\qquad\times$ & $\qquad$? \\
$\kappa\not\in\mathbb{Q}$ & $\qquad\times$ & $\quad\qquad\qquad\times$ & $\qquad\qquad\qquad\times$ & $\qquad\times$ & $\qquad?$
\end{tabular}
\caption{A table of all SLE$_\kappa$ speeds $\kappa\in(0,8)$ collected into disjoint groups with common properties.  Here, (\ref{exceptional}) defines $\kappa_{q,q'}$, and $?$ may be $\checkmark$ or $\times$, depending on the value of $\kappa$.  We prove column three (resp.\ four and five) in section \ref{minmodelsect} (resp.\ section \ref{frobsect}).}
\label{tablespeeds}
\end{table}

In this article, we prove some theorems and corollaries concerning the system (\ref{nullstate}, \ref{wardid}) that follow from these results and that relate to CFT and multiple SLE$_\kappa$.  In section \ref{frobsect}, we prove that any element of $\mathcal{S}_N$ equals a sum of at most two Frobenius series in powers of the distance between two adjacent points (i.e., coordinates of $\boldsymbol{x}\in\Omega_0$). (If $8/\kappa$ is odd, then a logarithmic term may multiply one of these sums.)  This establishes an important element in the operator product expansion (OPE) for one-leg boundary operators, assumed in CFT.   In section \ref{xingprob}, we identify the elements of $\mathcal{S}_N$ that are dual to the linear functionals of $\mathscr{B}_N^*$ (item \ref{fifthitem} of theorem \ref{maintheorem}), and we state some of their properties in theorem \ref{xingasymplem}.  Motivated by these  properties, we posit that these dual functions are in fact the connectivity weights we seek in goal \ref{item3} stated above, and we conjecture a formula (\ref{xing}, \ref{xing2}) for the ``crossing-probability" that the curves of a multiple-SLE$_\kappa$ process join in a specific connectivity.  In section \ref{intervalsect}, we introduce two different definitions  of a ``pure interval."  First, a pure interval in multiple SLE$_\kappa$ is either \emph{contractible} or \emph{propagating} according to the following conditions: If $F\in\mathcal{S}_N$ is the partition function (definition \ref{partitiondefn}) for a multiple-SLE$_\kappa$ process that, with probability one, generates a boundary arc (i.e., a fluctuating multiple-SLE$_\kappa$ curve in the long-time limit) with endpoints at $x_i$ and $x_{i+1}$, then we call $(x_i,x_{i+1})$ a \emph{contractible interval of $F$}.  Alternatively, if this multiple-SLE$_\kappa$ process, with probability one, generates a pair of distinct boundary arcs with endpoints at $x_i$ and $x_{i+1}$ respectively, then we call $(x_i,x_{i+1})$ a \emph{propagating interval of $F$.}  On the other hand, the definition of a pure interval in  CFT is different.  If the one-leg boundary operators at the interval's endpoints have only the identity (resp.\ two-leg) family in their OPE, then we call $(x_i,x_{i+1})$ an \emph{identity (resp.\ a two-leg) interval of $F$.}  Lemma \ref{proptwoleglem} states that propagating intervals and two-leg intervals are identical.  However, we find that contractible intervals and identity intervals are, in general, not identical.  (This may be understood in a statistical mechanics sense by recalling that an identity operator puts no conditions on boundary arc connectivities \cite{js}.)  In order to partially determine the relation between contractible and identity intervals, we ``insert" an identity interval into the domain of a connectivity weight in $\mathscr{B}_{N-1}$, generating an element of $\mathcal{S}_N$.  By decomposing the function that results over the basis $\mathscr{B}_N$ (\ref{finallincmb}), we characterize an identity interval in this situation as a particular linear superposition of a contractible interval and a propagating interval.  Finally, in section \ref{minmodelsect}, we explore the connection between the SLE$_\kappa$ exceptional speeds (\ref{exceptional}) and the CFT minimal models, and we propose conjecture \ref{minmodelconj} as a potential explanation for this connection.

In two future articles, we find explicit formulas for connectivity weights with $N\in\{1,2,3,4\}$ \cite{fsk}, and we combine the crossing-probability formulas (\ref{xing}, \ref{xing2}) with a physical interpretation of the elements of $\mathcal{B}_N$ (\ref{BN}) to predict formulas for cluster-crossing probabilities of critical lattice models (such as percolation, Potts models, and random cluster models) in a polygon with a free/fixed side-alternating boundary condition \cite{fkz}.  We verify our predictions with high-precision computer simulations of the $Q\in\{2,3\}$ critical random cluster model in a hexagon, finding good agreement.

\section{Frobenius series and one-leg boundary OPE}\label{frobsect}

In this section and with $\kappa\in(0,8)$ as usual, we find Frobenius series expansions for elements of $\mathcal{S}_N$ in powers of $x_{i+1}-x_i$ for any $i\in\{1,2,\ldots,2N-1\}$.  Theorem \ref{frobseriescor} summarizes our findings.  After we prove this theorem, we interpret these expansions as OPEs of CFT one-leg boundary operators in this section and again in section \ref{intervalsect}.

To begin, we show that any element of $\mathcal{B}_N$ (\ref{BN}) equals such a Frobenius series.  For every $\mathcal{F}_\vartheta\in\mathcal{B}_N$, (\ref{Fexplicit}) gives $2N$ different choices of formula for it, and these formulas vary only by the location of the point $x_c$ bearing the conjugate charge.  After choosing any $c\not\in\{i,i+1\}$, we note that the integration contours in the selected formula $\mathcal{F}_{c,\vartheta}$ may interact with the points $x_i$ and $x_{i+1}$ in one of these three ways:
\begin{enumerate}
\setcounter{enumi}{1}
\item\label{sc2} Both $x_i$ and $x_{i+1}$ are endpoints of one common contour, call it $\Gamma_1$.
\item\label{sc3} $x_i$ (resp.\ $x_{i+1}$) is an endpoint of one contour, call it $\Gamma_1$, and $x_{i+1}$ (resp.\ $x_i$) is not an endpoint of any contour.
\item\label{sc4} $x_i$ is an endpoint of one contour, call it $\Gamma_1$, and $x_{i+1}$ is an endpoint of a different contour, call it $\Gamma_2$.
\end{enumerate}
(The numbering follows appendix \red{A} of \cite{florkleb3}. We define case \red{1} below.)  Actually, we do not need to consider case \ref{sc4} at all.  Indeed, if the formula $\mathcal{F}_{c,\vartheta}$ falls under case \ref{sc4}, then there is always another $c'\not\in\{i,i+1\}$ such that the alternative formula $\mathcal{F}_{c',\vartheta}$ falls under case \ref{sc3} instead.  Thus, we only consider cases \ref{sc2} and \ref{sc3} here.

If $x_i$ and $x_{i+1}$ are endpoints of a common arc in the half-plane diagram for $\mathcal{F}_\vartheta$, then upon choosing $c\not\in\{i,i+1\}$, the formula $\mathcal{F}_{c,\vartheta}$ (\ref{Fexplicit}) selected for $\mathcal{F}_\vartheta$ falls under case \ref{sc2}.  As we noted between (\red{42}--\red{44}) and beneath (\red{44}) in \cite{florkleb3}, the substitution $u_1(t)=x_i(1-t)+x_{i+1}t$ for the integration along $\Gamma_1$ casts the Coulomb gas integral of (\ref{Fexplicit}) in the form
\be\label{Idecomp'}\mathcal{J}(\boldsymbol{x})=(x_{i+1}-x_i)^{1-8/\kappa}\mathcal{K}(\boldsymbol{x})\ee
for some function $\mathcal{K}$ that is analytic and non-vanishing at $x_{i+1}=x_i$.  (See also section \red{A 2} in \cite{florkleb3}.)  After inserting this factorization (\ref{Idecomp'}) into the selected formula (\ref{Fexplicit}) for $\mathcal{F}_\vartheta$, we find that 
\be\label{case2F}\mathcal{F}_\vartheta(\boldsymbol{x})=(x_{i+1}-x_i)^{1-6/\kappa}\mathcal{G}_\vartheta(\boldsymbol{x})\ee
for some function $\mathcal{G}_\vartheta$ that is analytic and non-vanishing at $x_{i+1}=x_i$.  We conclude that if $x_i$ and $x_{i+1}$ are endpoints of a common arc in the half-plane diagram for $\mathcal{F}_\vartheta$, then this function equals a Frobenius series centered on $x_{i+1}=x_i$ and with indicial power $1-6/\kappa$.  (We previously noted this fact in the paragraph beneath (\red{44}) in \cite{florkleb3}.)

If $x_i$ and $x_{i+1}$ are not endpoints of a common arc in the half-plane diagram for $\mathcal{F}_\vartheta$, then we choose $c\not\in\{i,i+1\}$ such that the formula $\mathcal{F}_{c,\vartheta}$ (\ref{Fexplicit}) for $\mathcal{F}_\vartheta$ falls under case \ref{sc3}.  Assuming $8/\kappa\not\in\mathbb{Z}^+$, we repeat the analysis in section \red{A 3} of \cite{florkleb3} next, deforming the integration contour $\Gamma_1$ of the Coulomb gas integral $\mathcal{J}$ in (\ref{Fexplicit}) into a contour falling under case \ref{sc2} and a collection of contours falling under what we refer to as ``case \red{1}," that is, with no endpoints at $x_i$ or $x_{i+1}$.  After we deform $\Gamma_1$ in this way, we find that the Coulomb gas integral of (\ref{Fexplicit}) decomposes into the sum
\be\label{Idecomp}\mathcal{J}(\boldsymbol{x})=\underbrace{n^{-1}\mathcal{J}_1(\boldsymbol{x})}_{\text{case \ref{sc2}}}+\underbrace{a_2\mathcal{J}_2(\boldsymbol{x})+\dotsm+a_{2N-3}\mathcal{J}_{2N-3}(\boldsymbol{x})+a_{2N-2}\mathcal{J}_{2N-2}(\boldsymbol{x})}_{\text{case \red{1}}},\ee
where $a_k\in\mathbb{R}$, $n$ is defined in (\ref{LkFk}), and $\mathcal{J}_1$ (resp.\ each $\mathcal{J}_k$ with $1<k\leq2N-2$) is a case \ref{sc2} (resp.\ case \red{1}) Coulomb gas integral with the same form and integration contours $\Gamma_2,$ $\Gamma_3,\ldots,\Gamma_{N-1}$ as $\mathcal{J}$ in (\ref{Fexplicit}), but with $\Gamma_1$ now different from that of $\mathcal{J}$.  As we observed earlier, $\mathcal{J}_1$, being a case \ref{sc2} term, factors as in (\ref{Idecomp'}) with $\mathcal{K}$ analytic and non-vanishing at $x_{i+1}=x_i$.  Furthermore, the case \red{1} terms of (\ref{Idecomp}), with neither $x_i$ nor $x_{i+1}$ an endpoint of any integration contour, are also analytic and non-vanishing at $x_{i+1}=x_i$.  Hence, after we insert the factorization (\ref{Idecomp'}) for $\mathcal{J}_1$ into (\ref{Idecomp}) and then insert the decomposition (\ref{Idecomp}) for $\mathcal{J}$ into (\ref{Fexplicit}), we find that
\be\label{expandF}\mathcal{F}_\vartheta(\boldsymbol{x})=(x_{i+1}-x_i)^{1-6/\kappa}\mathcal{G}_\vartheta(\boldsymbol{x})+(x_{i+1}-x_i)^{2/\kappa}\mathcal{H}_\vartheta(\boldsymbol{x})\ee
for some functions $\mathcal{G}_\vartheta$ and $\mathcal{H}_\vartheta$ that are both analytic and non-vanishing at $x_{i+1}=x_i$.  Here, the term with $\mathcal{G}_\vartheta$ (resp.\ $\mathcal{H}_\vartheta$) arises from the case \ref{sc2} term (resp.\ case \red{1} terms) in (\ref{Idecomp}).  We conclude that if $8/\kappa\not\in\mathbb{Z}^+$ and $x_i$ and $x_{i+1}$ are not endpoints of a common arc in the half-plane diagram for $\mathcal{F}_\vartheta$, then this function equals a sum of two Frobenius series in powers of $x_{i+1}-x_i$ and with respective indicial powers $1-6/\kappa$ and $2/\kappa$.  These powers are necessarily the same indicial powers that we derived in the analysis preceding lemma \red{3} in \cite{florkleb} by inserting a Frobenius series expansion for $F\in\mathcal{S}_N$ directly into the null-state PDEs centered on $x_i$ and $x_{i+1}$.

Supposing still that $8/\kappa\not\in\mathbb{Z}^+$, we determine if the other elements of $\mathcal{S}_N$ have the expansions encountered in the previous paragraph.  If in addition, $\kappa$ is not an exceptional speed (\ref{exceptional}) with $q\leq N+1$, then item \ref{firstitem} of theorem \ref{maintheorem} states that $\mathcal{B}_N$ is a basis for $\mathcal{S}_N$.  After decomposing $F\in\mathcal{S}_N$ over this basis and inserting (\ref{expandF}) for each term in the decomposition, we conclude that $F$ has this same form (\ref{expandF}).  Moreover, the indicial powers of these series do not increase due to cancellations of lower-order terms in this decomposition over $\mathcal{B}_N$ because they are fixed by the null-state PDE (\ref{nullstate}) centered on $x_i$ or $x_{i+1}$.  (See the calculation preceding lemma \red{3} of \cite{florkleb}.)  However, if $\kappa$ is an exceptional speed (\ref{exceptional}) with $q\leq N+1$, then whether or not all elements of $\mathcal{S}_N$ exhibit the expansion (\ref{expandF}) is unclear.  Indeed, if $F_1\in\mathcal{S}_N\setminus\text{span}\,\mathcal{B}_N$, then the proof of theorem \red{8} in \cite{florkleb3} shows that there is another function $F_2(\varkappa)\in\text{span}\,\mathcal{B}_N(\varkappa)$ such that for all $\varkappa$ sufficiently close to $\kappa$, 
\be\label{taylorseries}F_2(\varkappa)=(\varkappa-\kappa)F_1+O((\varkappa-\kappa)^2).\ee
Thus we may obtain $F_1$ from $F_2$ by differentiating the latter with respect to $\varkappa$, followed by setting $\varkappa=\kappa$.  This involves differentiating (\ref{Fexplicit}) with respect to $\kappa$, which at least initially introduces factors of $\log(x_{i+1}-x_i)$.

Moreover, if $8/\kappa\in\mathbb{Z}^+$, then the difference of the indicial powers in (\ref{expandF}) is an integer.  We recall the following fact of an ordinary differential equation studied near one of its regular singular points \cite{benors}.  If the zeros of the corresponding indicial polynomial differ by an integer, then typically there are two linearly independent solutions with the following properties.  One equals a Frobenius series in powers of the distance to the regular singular point, with its indicial power the bigger root of the polynomial.  The other equals the sum of another such Frobenius series, with its indicial power the smaller root, and the product of the logarithm of the distance to the regular singular point multiplied by another such Frobenius series, with its indicial power the greater root.  If this fact generalizes to the system (\ref{nullstate}, \ref{wardid}), then we may expect to see logarithmic factors multiplying some of these Frobenius series if $8/\kappa\in\mathbb{Z}^+$.

The following theorem shows that this is not quite the case.  Logarithmic terms  appear, but only if $8/\kappa$ is an {\it odd} integer, i.e., if $8/\kappa\in\mathbb{Z}^+$, and $\kappa$ is an exceptional speed (\ref{exceptional}).

\begin{theorem}\label{frobseriescor} Suppose that $\kappa\in(0,8)$, $F\in\mathcal{S}_N$, and $i\in\{1,2,\ldots,2N-1\}$. 
\begin{enumerate}
\item\label{frobitem1} If $8/\kappa\not\in\mathbb{Z}^+$, then there is an $R>0$ (depending on $x_j$ with $j\neq i+1$) and functions $A_m, B_m$ for each $m\in\mathbb{Z}^+\cup\{0\}$ such that if $0<x_{i+1}-x_i<R$, then ($\pi_i$ is the projection map that removes the $i$th coordinate from $\boldsymbol{x}\in\Omega_0$ (\ref{components}))
\be\label{nolog}F(\boldsymbol{x})=(x_{i+1}-x_i)^{1-6/\kappa}\sum_{m=0}^\infty(A_m\circ\pi_i)(\boldsymbol{x})\,(x_{i+1}-x_i)^m+(x_{i+1}-x_i)^{2/\kappa}\sum_{m=0}^\infty(B_m\circ\pi_i)(\boldsymbol{x})\,(x_{i+1}-x_i)^m.\ee
Also, if $A_0=0$ (resp.\ $B_0=0$), then $A_m=0$ (resp.\ $B_m=0$) for all $m\in\mathbb{Z}^+\cup\{0\}$.
\item\label{frobitem2} If $8/\kappa=r$ with $r$ even, then there is an $R>0$ (depending on $x_j$ with $j\neq i+1$) and functions $A_m$ for each $m\in\{0,1,\ldots,r-2\}$ and $B_m$ for each $m\in\mathbb{Z}^+\cup\{0\}$ such that if $0<x_{i+1}-x_i<R$, then
\be\label{stillnolog}F(\boldsymbol{x})=(x_{i+1}-x_i)^{1-6/\kappa}\sum_{m=0}^{r-2}(A_m\circ\pi_i)(\boldsymbol{x})\,(x_{i+1}-x_i)^m+(x_{i+1}-x_i)^{2/\kappa}\sum_{m=0}^\infty(B_m\circ\pi_i)(\boldsymbol{x})\,(x_{i+1}-x_i)^m.\ee
Also, if $A_0=0$, then $A_m=0$ for all $m\in\{0,1,\ldots,r-2\}$, and if $A_0=B_0=0$, then $F$ is zero.
\item\label{frobitem3} If $8/\kappa=r$ with $r>1$ odd, then there is an $R>0$ (depending on $x_j$ with $j\neq i+1$) and functions $A_m$ for each $m\in\{0,1,\ldots,r-2\}$ and $B_m,C_m$ for each $m\in\mathbb{Z}^+\cup\{0\}$ such that if $0<x_{i+1}-x_i<R$, then
\begin{multline}\label{log}\begin{aligned}F(\boldsymbol{x})=(x_{i+1}-x_i)^{1-6/\kappa}\sum_{m=0}^{r-2} (A_m\circ\pi_i)(\boldsymbol{x})\,(x_{i+1}-x_i)^m\,+\,&(x_{i+1}-x_i)^{2/\kappa}\sum_{m=0}^\infty (B_m\circ\pi_i)(\boldsymbol{x})\,(x_{i+1}-x_i)^m\\
+\,\log(x_{i+1}-x_i)&(x_{i+1}-x_i)^{2/\kappa}\sum_{m=0}^\infty (C_m\circ\pi_i)(\boldsymbol{x})\,(x_{i+1}-x_i)^m.\end{aligned}\end{multline}
Also, if $A_0=0$ or $C_0=0$, then $A_m=0$ for all $m\in\{0,1,\ldots,r-2\}$ and $C_m=0$ for all $m\in\mathbb{Z}^+\cup\{0\}$, and if $A_0=B_0=0$, then $F$ is zero.  Finally, the last series in (\ref{log}) with the logarithm factor dropped is in $\mathcal{S}_N$.
\end{enumerate}
In cases \ref{frobitem1}--\ref{frobitem3}, $\partial_iA_0=0$, $A_0\in\mathcal{S}_{N-1}$, and (if $\kappa\neq4$) $A_1=0$. (We discuss the case $\kappa=4$ after the proof.)
\end{theorem}

\begin{proof}  
Before we prove the theorem, we note that in items \ref{frobitem2} and \ref{frobitem3}, the difference of the indicial powers in (\ref{stillnolog}, \ref{log}) is $2/\kappa-(1-6/\kappa)=r-1\in\mathbb{Z}^+.$  Therefore, we truncate the first series in (\ref{stillnolog}, \ref{log}) at $m=r-2$ and include its tail with the second series.

First, we prove item \ref{frobitem1}.  The discussion preceding this theorem and leading to (\ref{expandF}) proves that if $8/\kappa\not\in\mathbb{Z}^+$, then every element of $\mathcal{B}_N$ admits the expansion (\ref{nolog}).  If $\kappa$ is not an exceptional speed (\ref{exceptional}) with $q\leq N+1$, then $\mathcal{B}_N$ is a basis for $\mathcal{S}_N$ according to item \ref{firstitem} of theorem \ref{maintheorem}, so every element of $\mathcal{S}_N$ admits the expansion (\ref{nolog}) too.  However, if $\kappa$ is such a speed, then $\mathcal{B}_N$ is not a basis for $\mathcal{S}_N$, so this conclusion does not immediately follow.  In this case, the elements in $\mathcal{B}_N$ satisfy exactly $d$ different linear dependencies, and we write each as $s_\varsigma(\kappa)=0$, where
\be\label{lindep}s_\varsigma(\varkappa):=\sum_{\vartheta=1}^{C_N} a_{\varsigma,\vartheta}\mathcal{F}_\vartheta(\varkappa),\quad \varsigma\in\{1,2,\ldots,d\},\ee
and $\{\boldsymbol{a}_1,\boldsymbol{a}_2,\ldots,\boldsymbol{a}_d\}$ with $\boldsymbol{a}_\varsigma:=(a_{\varsigma,1},a_{\varsigma,2},\ldots,a_{\varsigma,C_N})$ is a basis for $\ker\,(M_N\circ n)(\kappa)$.  With $A$ any $C_N\times C_N$ invertible matrix whose first $d$ columns are $\boldsymbol{a}_1$, $\boldsymbol{a}_2,\ldots,\boldsymbol{a}_d$, the proof of theorem \red{8} in \cite{florkleb3} shows that the linearly independent set
\begin{multline}\label{BNprime}\mathcal{B}_N^{\scaleobj{0.75}{\bullet}}(\varkappa):=\left\{(\varkappa-\kappa)^{-1}\sideset{}{_\vartheta}\sum a_{1,\vartheta}\mathcal{F}_\vartheta(\varkappa),\quad(\varkappa-\kappa)^{-1}\sideset{}{_\vartheta}\sum a_{2,\vartheta}\mathcal{F}_\vartheta(\varkappa),\quad\ldots\right.\\
\left.\ldots,\quad(\varkappa-\kappa)^{-1}\sideset{}{_\vartheta}\sum a_{d,\vartheta}\mathcal{F}_\vartheta(\varkappa),\quad\sideset{}{_\vartheta}\sum a_{d+1,\vartheta}\mathcal{F}_\vartheta(\varkappa),\quad\ldots,\quad\sideset{}{_\vartheta}\sum a_{C_N,\vartheta}\mathcal{F}_\vartheta(\varkappa)\right\}\end{multline}
goes to a basis for $\mathcal{S}_N(\kappa)$ as $\varkappa\rightarrow\kappa$.  Because the last $C_N-d$ elements of this set (\ref{BNprime}) are in the span of $\mathcal{B}_N$, each admits the expansion (\ref{nolog}).  To show that the first $d$ elements of (\ref{BNprime}) have this expansion too, we examine the limit of
\be\label{Pidecomp}\mathcal{F}^{\scaleobj{0.75}{\bullet}}_\varrho(\varkappa):=(\varkappa-\kappa)^{-1}\sideset{}{_\vartheta}\sum a_{\varrho,\vartheta}\mathcal{F}_\vartheta(\varkappa),\quad \varrho\leq d\ee
as $\varkappa\rightarrow\kappa$.  According to the proof of theorem \red{8} in \cite{florkleb3}, the sum on the right side of (\ref{Pidecomp}) equals $a_\varrho(\varkappa-\kappa)+O((\varkappa-\kappa)^2)$ for some nonzero constant $a_\varrho$.  Therefore,
\be\label{Laurent}\mathcal{F}^{\scaleobj{0.75}{\bullet}}_\varrho(\varkappa)=
\partial_\varkappa\left[\sideset{}{_\vartheta}\sum a_{\varrho,\vartheta}\mathcal{F}_\vartheta(\varkappa)\right]_{\varkappa=\kappa}+\,\,\,O(\varkappa-\kappa),\quad \varrho\leq d.\ee
Next, we insert the Frobenius series expansion (\ref{nolog}) for each $\mathcal{F}_\vartheta(\varkappa)$ into (\ref{Laurent}), denoting its expansion coefficients as $A_{\vartheta,m}(\varkappa)$ and $B_{\vartheta,m}(\varkappa)$.  Suppressing dependence on the points in $\{x_j\}_{j\neq i,i+1}$, we find
\be\label{Laurentsubs}\begin{aligned}\mathcal{F}^{\scaleobj{0.75}{\bullet}}_\varrho(\varkappa\,|\,x_i,x_{i+1})=\hspace{.1cm}
&(x_{i+1}-x_i)^{1-6/\kappa}\sideset{}{_{\vartheta,m}}\sum a_{\varrho,\vartheta}\partial_\kappa A_{\vartheta,m}(\kappa)(x_{i+1}-x_i)^m\\
+\hspace{.1cm}&(x_{i+1}-x_i)^{2/\kappa\hphantom{-1}}\sideset{}{_{\vartheta,m}}\sum a_{\varrho,\vartheta}\partial_\kappa B_{\vartheta,m}(\kappa)(x_{i+1}-x_i)^m\\
+\hspace{.1cm}&\frac{6}{\kappa^2}\log(x_{i+1}-x_i)(x_{i+1}-x_i)^{1-6/\kappa}\sideset{}{_{\vartheta,m}}\sum a_{\varrho,\vartheta} A_{\vartheta,m}(\kappa)(x_{i+1}-x_i)^m\\
-\hspace{.1cm}&\frac{2}{\kappa^2}\log(x_{i+1}-x_i)(x_{i+1}-x_i)^{2/\kappa\hphantom{-1}}\sideset{}{_{\vartheta,m}}\sum a_{\varrho,\vartheta} B_{\vartheta,m}(\kappa)(x_{i+1}-x_i)^m+O(\varkappa-\kappa).\end{aligned}\ee
Now, the two Frobenius series that appear in the expansion (\ref{nolog}) for $s_\varrho(\varkappa)$ (\ref{lindep}) match the two series that multiply the logarithms in (\ref{Laurentsubs}) at $\varkappa=\kappa$.  With $s_\varrho(\kappa)=0$, we thus find
\be\label{sseries}(x_{i+1}-x_i)^{1-6/\kappa}\sideset{}{_{\vartheta,m}}\sum a_{\varrho,\vartheta} A_{\vartheta,m}(\kappa)(x_{i+1}-x_i)^m+(x_{i+1}-x_i)^{2/\kappa}\sideset{}{_{\vartheta,m}}\sum a_{\varrho,\vartheta} B_{\vartheta,m}(\kappa)(x_{i+1}-x_i)^m=0,\ee
and with $2/\kappa-(1-6/\kappa)\not\in\mathbb{Z}$, (\ref{sseries}) implies that each series that multiplies a logarithm in (\ref{Laurentsubs}) must vanish.  Hence, all elements of the basis $\mathcal{B}_N^{\scaleobj{0.75}{\bullet}}(\kappa)$, and therefore of $\mathcal{S}_N(\kappa)$, admit the expansion (\ref{nolog}).  Furthermore, the analysis that precedes lemma \red{3} in \cite{florkleb} shows that the null-state PDEs (\ref{nullstate}) centered on $x_i$ and $x_{i+1}$ fix the indicial powers.  Because these powers do not differ by an integer, it follows that if $A_0=0$ (resp.\ $B_0=0$), then $A_m$ (resp.\ $B_m=0$) for all $m\in\mathbb{Z}^+\cup\{0\}$.

Next, we prove item \ref{frobitem2}.   From the formula (\red{35}) of \cite{florkleb3} (with $c\in\{1,2,\ldots2N\}$), we see that every element of the basis $\mathcal{B}_N$, and therefore of $\mathcal{S}_N$, admits the expansion (\ref{stillnolog}).  And again, the analysis that preceded lemma \red{3} in \cite{florkleb} shows that the null-state PDEs (\ref{nullstate}) centered on $x_i$ and $x_{i+1}$ fix the indicial powers.  Because these powers differ by $r-1\in\mathbb{Z}^+$, it follows that if $A_0=0$, then $A_m=0$ for all $m\in\{0,1,\ldots,r-2\}$, and if $A_0=B_0=0$, then $B_m=0$ for all $m\in\mathbb{Z}^+$.  In this latter case, $F$ vanishes if $0<x_{i+1}-x_i<R$.  As a result, the coefficients of its decomposition over $\mathcal{B}_N$ vanish too, and we conclude that $F$ is zero.  (The same reasoning also proves that $F$ is zero if $A_0=B_0=0$ in item \ref{frobitem3}.)

Because the proof of item  \ref{frobitem3} is a bit lengthy, we present it in appendix \ref{appendix}.  Finally, to show that $\partial_iA_0=0$ and that $A_1=0$ if $\kappa\neq4$, we insert the expansions (\ref{nolog}--\ref{log}) into the null-state PDEs centered on $x_i$ and $x_{i+1}$ and repeat the calculation preceding lemma \red{3} in \cite{florkleb}.  That $A_0\in\mathcal{S}_{N-1}$ is an immediate consequence of lemma \red{5} in \cite{florkleb}.
\end{proof}

\begin{table}[t]
\centering
\begin{tabular}{p{2cm}p{2.5cm}p{8cm}l}
Interval & Interval type & Frobenius series expansion in powers of $x_{i+1}-x_i$ & OPE content \\
\hline 
\multirow{3}{*}{$(x_i,x_{i+1})$} &
two-leg & (\ref{nolog}) with $A_m=0$ for all $m\in\mathbb{Z}^+$ & $\psi_1(x_i)\times\psi_1(x_{i+1})\sim\psi_2(x_i)$ \\
&identity & (\ref{nolog}) with $A_m\neq0$ and $B_m=0$ for all $m\in\mathbb{Z}^+$ & $\psi_1(x_i)\times\psi_1(x_{i+1})\sim\boldsymbol{1}$\\
&(neither)& (\ref{nolog}) with $A_m\neq0$ and $B_m\neq0$ for all $m\in\mathbb{Z}^+$
& $\psi_1(x_i)\times\psi_1(x_{i+1})\sim\boldsymbol{1}+\psi_2(x_i)$\\
\end{tabular}
\caption{The forms of Frobenius series expansions of $F\in\mathcal{S}_N$ in powers of $x_{i+1}-x_i$ relative to interval types for $(x_i,x_{i+1})$, assuming $\kappa\in(0,8)$ and $8/\kappa\not\in\mathbb{Z}^+$.  The right column shows the corresponding content of the OPE $\psi_1(x_i)\times\psi_1(x_{i+1})$.}
\label{FrobTable}
\end{table}

Theorem \ref{frobseriescor} complements the interpretation of various solutions $F\in\mathcal{S}_N$ of the system (\ref{nullstate}, \ref{wardid}) as CFT correlation functions of $2N$ one-leg boundary operators $\psi_1$: 
\be\label{2Npoint}F(\boldsymbol{x})=\langle\psi_1(x_1)\psi_1(x_2)\dotsm\psi_1(x_{2N})\rangle.\ee
In CFT, one assumes the existence of an OPE between the adjacent primary operators $\psi_1(x_i)$ and $\psi_1(x_{i+1})$ \cite{bpz, fms, henkel}, and except in logarithmic cases, the position of $\psi_1$ in the Kac table limits the content of this OPE to conformal families of just two other primary operators, the identity operator $\mathbf{1}$ and the two-leg boundary operator $\psi_2(x_i)$ \cite{florkleb,bpz, fms, henkel}.  After we insert this OPE into the correlation function (\ref{2Npoint}), we discover that the correlation function admits precisely the Frobenius series expansion described in theorem \ref{frobseriescor}.  In particular, the indicial powers stated in theorem \ref{frobseriescor} follow from the conformal weights of the one-leg boundary operators and the two primary operators in their OPE \cite{bpz, fms, henkel} thus:
\be\label{indicialpowers} 
\psi_1\times\psi_1\sim\begin{cases}\boldsymbol{1}: & \text{indicial power $=-2\theta_1+\theta_0=1-6/\kappa$} \\
\psi_2: & \text{indicial power $=-2\theta_1+\theta_2=2/\kappa$}\end{cases}.
\ee
(See (\red{6}) in \cite{florkleb2} for a formula for the conformal weight $\theta_s$ of the $s$-leg boundary operator $\psi_s$ in terms of $\kappa$.)  Thus, theorem \ref{frobseriescor} rigorously confirms part of the OPE  assumed in CFT.

The claims of theorem \ref{frobseriescor} that $A_0\in\mathcal{S}_{N-1}$, $\partial_iA_0=0$, and (if $\kappa\neq4$) $A_1=0$ have CFT interpretations too.  The first implies that the conformal family of the first Frobenius series belongs to the identity operator.  The second implies that the identity operator is non-local.  And the third implies that the level-one descendant of the identity operator vanishes, which CFT assumes even if $\kappa=4$ ($c=1$).  In particular, if $\kappa=4$ so there is no $A_1$ in (\ref{stillnolog}), then thanks to this vanishing descendant, only the two-leg family contributes to $B_0$ in (\ref{stillnolog}).  Thus, if $B_0\neq0$ (resp.\ $B_0=0$), then the two-leg family is (resp.\ is not) present in the OPE of the one-leg boundary operators at $x_i$ and $x_{i+1}$.

In definition \red{13} of \cite{florkleb} and definition \ref{cftintervaldefn} of section \ref{purecft} below, we define the terms ``two-leg interval" and ``identity interval" $(x_i,x_{i+1})$ of a solution $F\in\mathcal{S}_N$.  In CFT parlance, these terms are taken to indicate respectively that only the two-leg family or the identity family appear in the OPE of $\psi_1(x_i)$ with $\psi(x_{i+1})$ if $F$ is the $2N$-point function (\ref{2Npoint}).   Taken together, theorem \ref{frobseriescor} and corollary \ref{frobintervalcor}, stated below, elucidate the connection between our definitions of these terms and their CFT usage.  In particular and without reference to CFT, corollary \ref{frobintervalcor} identifies these terms with particular forms of the series expansions (\ref{nolog}, \ref{stillnolog}, \ref{log}).  Table \ref{FrobTable} summarizes this connection for $8/\kappa\not\in\mathbb{Z}^+$.

Two conformal families, with respective weights $\theta_0=0$ and $\theta_2=8/\kappa-1$, appear in the OPE of two one-leg boundary operators, and logarithmic CFT anticipates the presence of terms with logarithms in the series expansions of theorem \ref{frobseriescor} if these conformal weights differ by an integer.  However, theorem \ref{frobseriescor} shows that such terms appear only if this integer $\theta_2-\theta_0$ is even (i.e., $8/\kappa\in2\mathbb{Z}^++1$, item \ref{frobitem3} of theorem \ref{frobseriescor}) and do not appear if it is odd (i.e., $8/\kappa\in2\mathbb{Z}^+$, item \ref{frobitem2} of theorem \ref{frobseriescor}).  In the former case, three conformal families contribute to the sums in (\ref{log}): the identity family contributes to the first and second sum; the two-leg family contributes to the third sum; and the logarithmic partner to the two-leg family contributes to the second sum \cite{gurarie}.  Although logarithmic terms do not appear for other $\kappa\in(0,8)$ as we bring together just two points among $x_1,$ $x_2,\ldots,x_{2N}$, they may appear for some exceptional speeds (\ref{exceptional}) as we bring together three or more of these points.  Ref.\ \cite{js, gurarie, gurarie2, matrid, rgw, vjs} and references therein give more information about logarithmic CFT.  In particular, \cite{gurarie} studies the case $\kappa=8$ (not included here because we restrict to $\kappa\in(0,8)$), and \cite{gurarie2} considers the case $\kappa=8/3$.

\section{Multiple-SLE$_\kappa$ arc connectivities}\label{xingprob}

In this section, we define and explore special properties of certain elements of $\mathcal{S}_N$ that we call ``connectivity weights."  Afterwards, we conjecture a formula for the probability that the curves of a multiple-SLE$_\kappa$ process join their endpoints together in a particular arc connectivity.

First, we motivate the notion of a connectivity weight by interpreting the behaviors of elements of $\mathcal{S}_N$ as two adjacent points approach each other in terms of multiple SLE$_\kappa$ \cite{dub2, bbk, graham, kl, sakai}.  In the introduction \red{I} of \cite{florkleb}, we note that the multiple-SLE$_\kappa$ process is defined up to an unspecified function called an ``SLE$_\kappa$ partition function."
\begin{defn}\label{partitiondefn} A function $F:\Omega_0\rightarrow\mathbb{R}$ is an \emph{SLE$_\kappa$ partition function} if it satisfies the system (\ref{nullstate}, \ref{wardid}) and if $F(\boldsymbol{x})\neq0$ for all $\boldsymbol{x}\in\Omega_0$.
\end{defn}
\noindent
Thanks to the intermediate value theorem, it immediately follows that an SLE$_\kappa$ partition function is either positive-valued or negative-valued.  In this article, we consider only SLE$_\kappa$ partition functions in $\mathcal{S}_N$ (i.e., satisfying (\ref{powerlaw})).

\begin{figure}[t]
\centering
\includegraphics[scale=0.3]{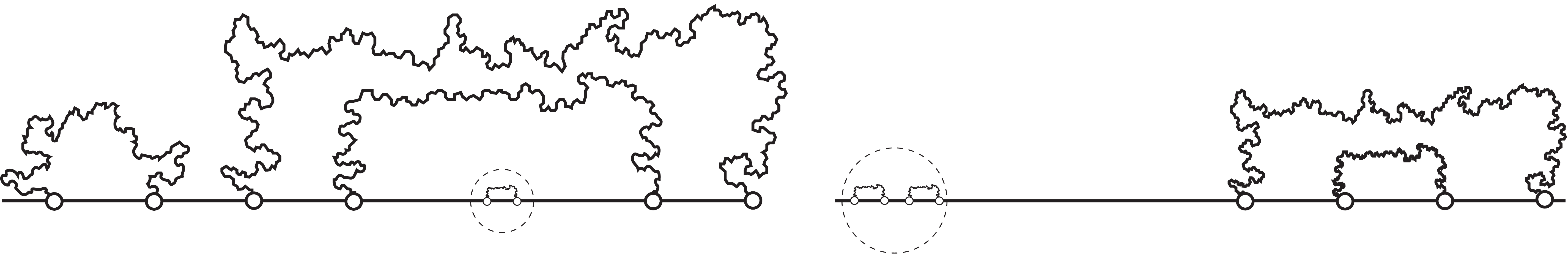}
\caption{With high probability, a boundary arc that shares its endpoints with a very short interval (circled) is microscopic and thus does not influence macroscopic boundary arcs.  The left (resp.\ right) illustration goes with (\ref{Piprefactor}, \ref{Pifactor}) (resp.\ (\ref{Piprevanish}, \ref{Pivanish})).}
\label{Micro}
\end{figure}

In the introduction \red{I} of \cite{florkleb}, we suppose the existence of an SLE$_\kappa$ partition function $\Pi_\vartheta\in\mathcal{S}_N$, called a ``connectivity weight," with the following special property.  Any collection of boundary arcs generated by a multiple-SLE$_\kappa$ process with $\Pi_\vartheta(\boldsymbol{x})$ for its partition function joins the coordinates of $\boldsymbol{x}$ together pairwise in the $\vartheta$th connectivity almost surely.  Supposing that such a partition function exists, we glean information about it from this defining property.  Item \ref{fourthitem} of theorem \ref{maintheorem} implies that in order to determine $\Pi_\vartheta(\boldsymbol{x})$, it suffices to know its asymptotic behavior as we pull the coordinates of $\boldsymbol{x}$ together in various disjoint pairs.  Thus, we take this approach to determining $\Pi_\vartheta$.  For the following discussion, we let $\mathscr{E}_\vartheta$ denote the event that the boundary arcs join the $2N$ coordinates of $\boldsymbol{x}$ in the $\vartheta$th connectivity for some $\vartheta\in\{1,2,\ldots,C_{N-1}\}$.  (Of course, the probability of this event in the multiple-SLE$_\kappa$ process with $\Pi_\vartheta$ for its partition function is one.)  Also in the following discussion, we interpret $\Pi_\vartheta$ as a statistical mechanics partition function (or really, an appropriate ratio of such partition functions) for some critical lattice model (e.g.\ percolation, Potts model, random cluster model), summing exclusively over the event $\mathscr{E}_\vartheta$.  Ref.\ \cite{bbk} explains this interpretation.

To begin, we suppose that the points $x_i$ and $x_{i+1}$ are endpoints of a common arc in the $\vartheta$th connectivity.  If $x_{i+1}-x_i\ll |x_j-x_i|$ for all $j\not\in\{i,i+1\}$, then with high probability, the boundary arc connecting $x_i$ with $x_{i+1}$ in the event $\mathscr{E}_\vartheta$ is very small (figure \ref{Micro}).  As such, it hardly influences the statistics of the other larger arcs in the system, so we expect that $\Pi_\vartheta$, interpreted as a statistical mechanics partition function, factors as
\be\label{Piprefactor}\Pi_\vartheta(\boldsymbol{x})\underset{x_{i+1}\rightarrow x_i}{\sim}\overbrace{\left(\parbox{2.12in}{connectivity weight for microscopic system with only one boundary arc}\right)}^{\in\mathcal{S}_1}\times\overbrace{\left(\parbox{2.12in}{connectivity weight for macroscopic system with $N-1$ boundary arcs}\right)}^{\in\mathcal{S}_{N-1}}.\ee
According to (\red{16}) of \cite{florkleb}, the first factor in (\ref{Piprefactor}), belonging to $\mathcal{S}_1$, equals $(x_{i+1}-x_i)^{1-6/\kappa}$ multiplied by a constant that we  set to one for convenience.  Also, we enumerate the $C_{N-1}$ arc connectivities on the points in $\{x_j\}_{j\neq i,i+1}$ such that the $\vartheta$th connectivity, now with only $N-1$ arcs, follows from dropping the microscopic curve with endpoints at $x_i$ and $x_{i+1}$ from every sample in $\mathscr{E}_\vartheta$.  With $\Xi_\vartheta\in\mathcal{S}_{N-1}$ denoting its corresponding connectivity weight, (\ref{Piprefactor}) becomes
\be\label{Pifactor}\Pi_\vartheta(\boldsymbol{x})\underset{x_{i+1}\rightarrow x_i}{\sim}(x_{i+1}-x_i)^{1-6/\kappa}(\Xi_\vartheta\circ\pi_{i,i+1})(\boldsymbol{x})\quad\text{if a boundary arc joins $x_i$ with $x_{i+1}$ in the $\vartheta$th connectivity.}\ee
In other words, if $\bar{\ell}_1$ (\ref{lim}) sending $x_{i+1}\rightarrow x_i$ is the first limit of some selected element of $[\mathscr{L}_\vartheta]$, then $\bar{\ell}_1\Pi_\vartheta=\Xi_\vartheta$.

Next, we suppose that the points $x_j$ and $x_{j+1}$ are not endpoints of a common arc in the $\vartheta$th connectivity.  For convenience and without loss of generality, we momentarily assume that $j=2$, an arc joins $x_1$ with $x_2$ in this particular connectivity, and a different arc joins $x_3$ with $x_4$ as well.  If $x_4-x_1\ll x_k-x_1$ for all $k\not\in\{1,2,3,4\}$, then with high probability, the two boundary arcs attached to $x_1$, $x_2$, $x_3$, and $x_4$ are very small (figure \ref{Micro}).  As such, they hardly influence the statistics of the other arcs in the system, so we expect that $\Pi_\vartheta$ factors as
\begin{multline}\label{Piprevanish}\Pi_\vartheta(\boldsymbol{x})\underset{\substack{|x_4-x_1|\ll x_5,\\ \quad x_6,\ldots,x_{2N}}}{\sim}\overbrace{\left(\parbox{3.53in}{connectivity weight for microscopic system with one boundary arc joining $x_1$ with $x_2$ and another joining $x_3$ with $x_4$}\right)}^{F\in\mathcal{S}_2}\\
\times\underbrace{\left(\parbox{2.12in}{connectivity weight for macroscopic system with $N-2$ boundary arcs}\right)}_{\in\mathcal{S}_{N-2}}.\end{multline}
According to the previous paragraph, the factor $F\in\mathcal{S}_2$ on the right side of (\ref{Piprevanish}) should asymptotically approach the product of two connectivity weights in $\mathcal{S}_1$, one for an arc joining $x_1$ with $x_2$ and another for an arc joining $x_3$ with $x_4$, as $x_2\rightarrow x_1$.  Moreover, $F$ should not asymptotically approach a different product of two connectivity weights in $\mathcal{S}_1$, one for an arc joining $x_3$ with $x_2$ and another for an arc joining $x_4$ with $x_1$, as $x_3\rightarrow x_2$.  That is,
\be\label{bdycond} F(x_1,x_2,x_3,x_4)\,\,\,\begin{cases}\sim\hphantom{C}(x_4-x_3)^{1-6/\kappa}(x_2-x_1)^{1-6/\kappa}, & x_2\rightarrow x_1 \\ \not\sim C(x_3-x_2)^{1-6/\kappa}(x_4-x_1)^{1-6/\kappa}, & x_3\rightarrow x_2\end{cases},\quad C\in\mathbb{R}.\ee
This boundary condition (\ref{bdycond}) identifies $F$ with a unique element of $\mathcal{S}_2$ whose formula we may find from (\red{17}--\red{19}) in \cite{florkleb}.  It follows from these equations that the first factor $F(x_1,x_2,x_3,x_4)$ in (\ref{Piprevanish}) is $O((x_3-x_2)^{2/\kappa})$ as $x_3\rightarrow x_2$.  Putting everything together in (\ref{Piprevanish}), we expect that (the condition $j=2$ is convenient but not necessary to our present argument, so we do not show it here)
\be\label{Pivanish}\Pi_\vartheta(\boldsymbol{x})\underset{x_{j+1}\rightarrow x_j}{\sim}(x_{j+1}-x_j)^{2/\kappa}(\Lambda_\vartheta\circ\pi_{j+1})(\boldsymbol{x})\quad\text{if no boundary arc joins $x_j$ with $x_{j+1}$ in the $\vartheta$th connectivity}\ee
for some function $\Lambda_\vartheta$.  (We give a partial interpretation of this function beneath (\ref{newxing}).)  In (\ref{Pivanish}), we have ignored the condition $x_4-x_1\ll x_k-x_1$ for all $k>4$ because the power law $(x_3-x_2)^{2/\kappa}$ in (\ref{Pivanish}) (with $j=2$) does not change as $x_4$ moves near the midpoint between $x_3$ and $x_5$.  (Indeed, if it did, then (\ref{nolog}--\ref{log}) implies that $A_0$ would vanish only for $x_4$ close to $x_3$.  But this is impossible because $A_0\in\mathcal{S}_{N-1}$.)  Also, we ignore the condition that the boundary arcs anchored to $x_j$ and $x_{j+1}$ terminate at the points adjacent to them because the behavior (\ref{Pivanish}) of $\Pi_\vartheta(\boldsymbol{x})$ as $x_{j+1}\rightarrow x_j$ should only depend on the fact that no arc joins $x_j$ with $x_{j+1}$ in the $\vartheta$th connectivity.  Hence, we infer from (\ref{Pivanish}) that if $\bar{\ell}_1$ (\ref{lim}) sending $x_{j+1}\rightarrow x_j$ is the first limit of some selected element of $[\mathscr{L}_\varsigma]$, then $\bar{\ell}_1\Pi_\vartheta=0$.  Also, $\varsigma\neq\vartheta$ because $x_j$ and $x_{j+1}$ are not endpoints of a common arc in the $\vartheta$th connectivity but are so in the $\varsigma$th connectivity.

Next, we apply the above analysis again to the connectivity weight $\Xi_\vartheta\in\mathcal{S}_{N-1}$ of (\ref{Pifactor}) and so on.  After $N-1$ repetitions, we eventually discover that $[\mathscr{L}_{\varsigma}]\Pi_\vartheta=\delta_{\varsigma,\vartheta}$.  That is, the $\vartheta$th connectivity weight is dual to $[\mathscr{L}_{\vartheta}]$.  With $\mathscr{B}_N^*:=\{[\mathscr{L}_1],[\mathscr{L}_2],\ldots,[\mathscr{L}_{C_N}]\}$ a basis for $\mathcal{S}_N^*$ thanks to item \ref{fifthitem} of theorem \ref{maintheorem}, this duality relation completely determines the connectivity weights, so we use it as a formal definition for the latter.

\begin{figure}[b]
\centering
\includegraphics[scale=0.3]{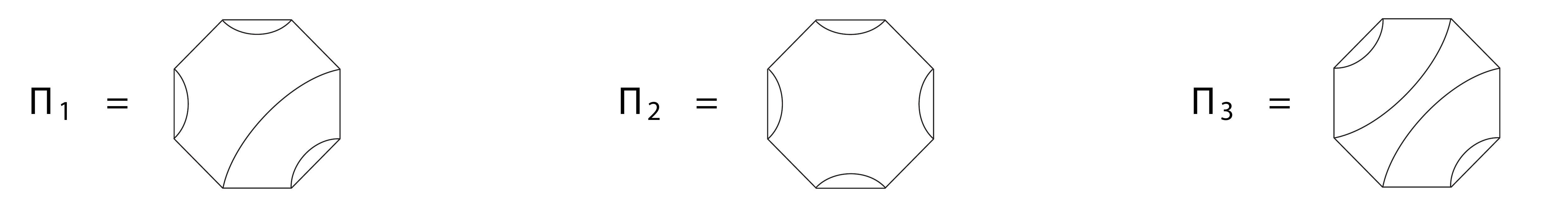}
\caption{Polygon diagrams for three different connectivity weights in $\mathscr{B}_4\subset\mathcal{S}_4$.  We find the other $C_4-3=11$ diagrams by rotating one of these three.  Compare with the polygon diagrams in figure \ref{Csls}.}
\label{PiDiagrams}
\end{figure}

\begin{defn}\label{dualbasis} Supposing that $\kappa\in(0,8)$, we define the \emph{$\varsigma$th connectivity weight} $\Pi_\varsigma$ to be the element of $\mathcal{S}_N$ that is dual to $[\mathscr{L}_{\varsigma}]\in\mathscr{B}_N^*$.  That is
\be\label{duality}\text{$[\mathscr{L}_{\varsigma}]\Pi_\vartheta=\delta_{\varsigma,\vartheta}$ for all $\varsigma,\vartheta\in\{1,2,\ldots,C_N\}$.}\ee
We let $\mathscr{B}_N=\{\Pi_1,\Pi_2,\ldots,\Pi_{C_N}\}$ be the basis for $\mathcal{S}_N$ dual to the basis $\mathscr{B}_N^*=\{[\mathscr{L}_1],[\mathscr{L}_2],\ldots,[\mathscr{L}_{C_N}]\}$ for $\mathcal{S}_N^*$. Finally, we define the \emph{polygon} (resp.\ \emph{half-plane}) diagram for $\Pi_\varsigma\in\mathscr{B}_N$ to be the polygon (resp.\ half-plane) diagram for $[\mathscr{L}_\varsigma]\in\mathscr{B}_N^*$, and we refer to either diagram simply as the \emph{diagram for $\Pi_\varsigma$} (figure \ref{PiDiagrams}).
\end{defn}
\noindent
With this definition, every element $F\in\mathcal{S}_N$ has the following useful decomposition over $\mathscr{B}_N$:
\be\label{decompose}F=a_1\Pi_1+a_2\Pi_2+\dotsm+a_{C_N}\Pi_{C_N},\quad a_{\varsigma}=[\mathscr{L}_\varsigma]F.\ee

Using the formulas (\ref{Fexplicit}) for the elements of $\mathcal{B}_N$, we may calculate explicit formulas for the connectivity weights.  Indeed, according to (\ref{decompose}), the $\varsigma$th coefficient $a_{\vartheta,\varsigma}$ in the decomposition $\mathcal{F}_\vartheta=a_{\vartheta,1}\Pi_1+a_{\vartheta,2}\Pi_2+\dotsm+a_{\vartheta,C_N}\Pi_{C_N}$ is given by (\ref{LkFk}), the $(\varsigma,\vartheta)$th entry of the meander matrix $M_N\circ n$.  Thus, we have
\be\label{F=MPi}\left(\begin{array}{l}\mathcal{F}_1\\\mathcal{F}_2\\\,\,\vdots\\ \mathcal{F}_{C_N}\end{array}\right)=(M_N\circ n)\left(\begin{array}{l}\Pi_1\\ \Pi_2 \\ \,\,\vdots\\ \Pi_{C_N}\end{array}\right).\ee
The formulas for the connectivity weights that follow from (\ref{F=MPi}) are very complicated in general. However, if $N$ is small, then it is often possible to construct simpler formulas by choosing integration contours prudently in (\ref{Fexplicit}).  We use this approach to construct simpler connectivity weight formulas for $N\in\{1,2,3,4\}$ in \cite{fsk}.

In \cite{florkleb3}, we note that $(M_N\circ n)(\kappa)$ is invertible if and only if $\kappa$ is not an exceptional speed (\ref{exceptional}) with $q\leq N+1$.  Hence, if $\kappa$ is such a speed, then we cannot use (\ref{F=MPi}) to calculate the connectivity weights of $\mathscr{B}_N$ explicitly.  However, we may decompose $\Pi_\varsigma$ over the alternative basis $\mathcal{B}_N^{\scaleobj{0.75}{\bullet}}$ used in the proof of theorem \ref{maintheorem} (i.e., theorem \red{8} in \cite{florkleb3}) by replacing  $\mathcal{F}_\vartheta$ with $\mathcal{F}_\vartheta^{\scaleobj{0.75}{\bullet}}$ and $M_N$ with $M_N^{{\scaleobj{0.75}{\bullet}}}$ in (\ref{F=MPi}), where these ``superscript-bullet" quantities are defined in that proof.  Because the elements of $(M_N^{{\scaleobj{0.75}{\bullet}}})^{-1}\circ n$ are continuous functions of $\varkappa\in(\kappa-\epsilon,\kappa+\epsilon)$ for some $\epsilon>0$, we may decompose $\Pi_\varsigma$ over $\mathcal{B}_N^{\scaleobj{0.75}{\bullet}}$ for all $\kappa$ in this interval to show that the limit of $\Pi_\varsigma(\varkappa)$ as $\varkappa\rightarrow\kappa$ exists and equals $\Pi_\varsigma(\kappa)$.  

The previous paragraph shows that we may alternatively invert (\ref{F=MPi}) with $\varkappa=\kappa+\epsilon$ and then send $\epsilon\rightarrow0$ to find a formula for $\Pi_\varsigma(\kappa)$ as a limit of a linear combination of elements of $\mathcal{B}_N$ as $\varkappa\rightarrow\kappa$.  This might seem advantageous because, unlike the elements of $\mathcal{B}_N^{{\scaleobj{0.75}{\bullet}}}$, we already have explicit formulas for those of the former set.  However, there are many quantities in this linear combination that diverge as $\varkappa\rightarrow\kappa$ and therefore must cancel each other in this limit, making this  definition for $\Pi_\varsigma(\kappa)$ too unwieldy for explicit calculations.

We stress that our definition \ref{dualbasis} of a connectivity weight is purely formal.  Whether or not these functions are indeed the desired multiple-SLE$_\kappa$ partition functions (and whether such functions are unique) remains to be proven.  To support the arguments that led to (\ref{duality}) in the first place, we at least prove that conditions (\ref{Pifactor}, \ref{Pivanish}) are satisfied.
\begin{theorem}\label{xingasymplem}Suppose that $\kappa\in(0,8)$ and $\Pi_\vartheta\in\mathscr{B}_N$ is the $\vartheta$th connectivity weight, and let $i\in\{1,2,\ldots,2N-1\}$.
\begin{enumerate}
\item
\begin{enumerate}
\item\label{1a} If $x_i$ and $x_{i+1}$ are endpoints of a common arc in the diagram for $\Pi_\vartheta$, then
\be \lim_{x_{i+1}\rightarrow x_i}(x_{i+1}-x_i)^{6/\kappa-1}\Pi_\vartheta(x_1,x_2,\ldots,x_{2N})=\Xi_\vartheta(x_1,x_2,\ldots,x_{i-1},x_{i+2},\ldots,x_{2N}),\ee
where $\Xi_\vartheta\in\mathscr{B}_{N-1}$ is the connectivity weight whose diagram matches that of $\Pi_\vartheta$ but with the arc deleted.
\item\label{1b} If $x_1$ and $x_{2N}$ are endpoints of a common arc in the diagram for $\Pi_\vartheta$, then
\be \lim_{R\rightarrow\infty}(2R)^{6/\kappa-1}\Pi_\vartheta(-R,x_2,x_3,\ldots,x_{2N-1},R)=\Xi_\vartheta(x_2,x_3,\ldots,x_{2N-1}),\ee
where $\Xi_\vartheta\in\mathscr{B}_{N-1}$ is the connectivity weight whose diagram matches that of $\Pi_\vartheta$ but with the arc deleted.
\end{enumerate}
\item
\begin{enumerate}
\item\label{2a} If $x_i$ and $x_{i+1}$ are not endpoints of a common arc in the diagram for $\Pi_\vartheta$, then
\be\label{twoleglim}\lim_{x_{i+1}\rightarrow x_i}(x_{i+1}-x_i)^{6/\kappa-1}\Pi_\vartheta(x_1,x_2,\ldots,x_{2N})=0.\ee
\item\label{2b} If $x_1$ and $x_{2N}$ are not endpoints of a common arc in the diagram for $\Pi_\vartheta$, then 
\be \lim_{R\rightarrow\infty}(2R)^{6/\kappa-1}\Pi_\vartheta(-R,x_2,x_3,\ldots,x_{2N-1},R)=0.\ee
\end{enumerate}
\end{enumerate}
\end{theorem}
\begin{proof}
If item \ref{1a} (resp.\ \ref{2a}) is true, then we may immediately prove item \ref{1b} (resp.\ \ref{2b}) by using the M\"obius transformation of the proof of lemma \red{5} in \cite{florkleb} and imitating the part of that proof where this transformation is used.  Therefore, it suffices to only prove items \ref{1a} and \ref{2a}.  For this purpose, we choose an $i\in\{1,2,\ldots,2N-1\}$ to use throughout.

First, we let $\mathscr{C}_N^*=\{[\mathscr{L}_1],[\mathscr{L}_2],\ldots,[\mathscr{L}_{C_{N-1}}]\}\subset\mathscr{B}_N^*$ be the subset of all equivalence classes whose diagram has an arc with its endpoints at $x_i$ and $x_{i+1}$ (because this arc fixes the connectivity of its two endpoints, only $2N-2$ points, $N-1$ arcs, and thus $C_{N-1}$ equivalence classes, remain), and we let $\mathscr{C}_N=\{\Pi_1,\Pi_2,\ldots,\Pi_{C_{N-1}}\}\subset\mathscr{B}_N$.  Furthermore, we let the symbol $\mathscr{M}$ stand for an allowable sequence of limits in $\mathcal{S}_{N-1}^*$ involving the points in $\{x_j\}_{j\neq i,i+1}$, and we enumerate the elements of $\mathscr{B}_{N-1}^*=\{[\mathscr{M}_1],[\mathscr{M}_2],\ldots,[\mathscr{M}_{C_{N-1}}]\}$ so the diagram for $[\mathscr{M}_\varsigma]$ is created by removing the arc with endpoints at $x_i$ and $x_{i+1}$ from the diagram for $[\mathscr{L}_\varsigma]\in\mathscr{C}_N^*$.  
Throughout this proof, we choose an element $\mathscr{L}_\varsigma$ of $[\mathscr{L}_\varsigma]\in\mathscr{C}_N^*$ that takes the limit $x_{i+1}\rightarrow x_i$ first.  We formally define this limit $\bar{\ell}_1$ in (\ref{lim}), and we may write $\mathscr{L}_\varsigma=\mathscr{M}_{\varsigma}\bar{\ell}_1$ for some $\mathscr{M}_\varsigma\in[\mathscr{M}_\varsigma]$.  
Because $\bar{\ell}_1F\in\mathcal{S}_{N-1}$ for all $F\in\mathcal{S}_N$ according to lemma \red{5} of \cite{florkleb}, lemma \red{12} of \cite{florkleb} implies that
\be\label{decomp}\text{$[\mathscr{L}_\varsigma]F=[\mathscr{M}_\varsigma]\bar{\ell}_1F$ for all $[\mathscr{L}_\varsigma]\in\mathscr{C}_N^*$ and all $F\in\mathcal{S}_N$.}\ee

We now choose an arbitrary $\Pi_\vartheta\in\mathscr{C}_N$ and prove item \ref{1a} for it.  After inserting $F=\Pi_\vartheta$ in (\ref{decomp}) and invoking the dual relation (\ref{duality}), we find that
\be\label{subdual}\text{$[\mathscr{M}_{\varsigma}]\Xi_{\vartheta}=\delta_{\varsigma,\vartheta}$ for each $\varsigma,\vartheta\in\{1,2,\ldots,C_{N-1}\}$, where $\Xi_{\vartheta}:=\bar{\ell}_1\Pi_\vartheta\in\mathcal{S}_{N-1}$}.\ee  
Because it satisfies the dual relation (\ref{duality}) relative to the elements of $\mathscr{B}_{N-1}^*$, $\Xi_\vartheta:=\bar{\ell}_1\Pi_\vartheta$ is the $\vartheta$th connectivity weight in $\mathscr{B}_{N-1}$.  This proves item \ref{1a}.

Now to finish, we choose an arbitrary $\Pi_\vartheta\in\mathscr{B}_N\setminus\mathscr{C}_N$ and prove item \ref{2a} for it.  If $[\mathscr{L}_\varsigma]\in\mathscr{C}_N^*$, then $\varsigma<\vartheta$, so $[\mathscr{L}_\varsigma]\Pi_\vartheta=0$.  After inserting this fact into (\ref{decomp}) with $F=\Pi_\vartheta$, we find
\be\label{nextdecomp}\text{$0=[\mathscr{L}_\varsigma]\Pi_\vartheta=[\mathscr{M}_{\varsigma}]\Xi_{\vartheta}$ for all $[\mathscr{L}_\varsigma]\in\mathscr{C}_N^*$, where $\Xi_{\vartheta}:=\bar{\ell}_1\Pi_\vartheta\in\mathcal{S}_{N-1}$.}\ee
In other words, $w(\Xi_{\vartheta})=0$, where $w:\mathcal{S}_{N-1}\rightarrow\mathbb{R}^{C_{N-1}}$ is the map whose $\varsigma$th component is $w(F)_\varsigma:=[\mathscr{M}_\varsigma]F$.  According to theorem \ref{maintheorem}, $w$ is a linear bijection.  Therefore, $\bar{\ell}_1\Pi_\vartheta=:\Xi_{\vartheta}=0$.
\end{proof}

It follows from theorem \ref{frobseriescor} that in the case of item \ref{2a}, $\Pi_\vartheta(\boldsymbol{x})=O((x_{i+1}-x_i)^{2/\kappa})$ as $x_{i+1}\rightarrow x_i$, in agreement with what we anticipated in (\ref{Pivanish}).

\begin{figure}[t]
\centering
\includegraphics[scale=0.3]{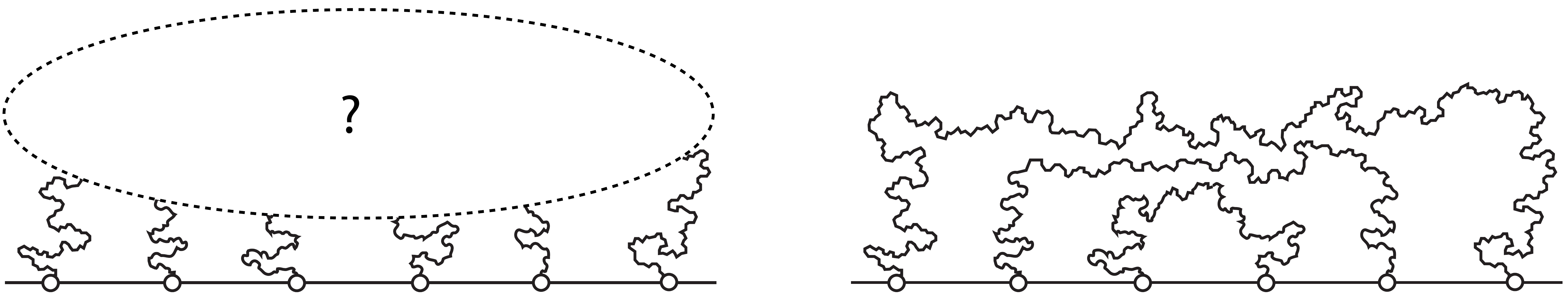}
\caption{Conjecture \ref{connectivityconj} proposes a formula for the probability that the growing curves of a SLE$_\kappa$ process (left) eventually join pairwise in the $\varsigma$th connectivity (right).}
\label{Connect}
\end{figure}

Having verified with theorem \ref{xingasymplem} that the connectivity weights, as defined above, have the desired properties (\ref{Pifactor}, \ref{Pivanish}), we conjecture a formula for the probability that the growing curves of a multiple-SLE$_\kappa$ process with partition function $F\in\mathcal{S}_N$ eventually join pairwise in the $\varsigma$th arc connectivity (figure \ref{Connect}).
\begin{conj}\label{connectivityconj}Suppose that $\kappa\in(0,8)$, and consider a multiple-SLE$_\kappa$ process that grows $2N$ curves in the upper half-plane from the points $x_1<x_2<\ldots<x_{2N}$ with the SLE$_\kappa$ partition function $F\in\mathcal{S}_N$.  Then 
\be\label{xing}
 P_\varsigma(x_1,x_2,\ldots,x_{2N})=[\mathscr{L}_{\varsigma}]F\,\frac{\Pi_\varsigma(x_1,x_2,\ldots,x_{2N})}{F(x_1,x_2,\ldots,x_{2N})},\quad \varsigma\in\{1,2,\ldots,C_N\}
 \ee
gives the ``crossing probability" $P_\varsigma(x_1,x_2,\ldots,x_{2N})$ that these curves eventually join pairwise in the $\varsigma$th connectivity.
\end{conj}

After inserting the decomposition (\ref{decompose}) into the denominator of (\ref{xing}), we find an alternative and more natural form for the conjectured crossing-probability formula
\be\label{xing2}P_\varsigma=\frac{a_\varsigma\Pi_\varsigma}{a_1\Pi_1+a_2\Pi_2+\dotsm+a_{C_N}\Pi_{C_N}},\quad a_\vartheta=[\mathscr{L}_\vartheta]F,\ee
which immediately gives the necessary property $P_1+P_2+\dotsm+P_{C_N}=1$.  In this second form (\ref{xing2}), the formula for $P_\varsigma$ precisely matches the multiple-SLE$_\kappa$ crossing-probability formula originally conjectured in \cite{bauber2}.

If conjecture \ref{connectivityconj} and the following conjecture is true, then the condition $0\leq P_\varsigma\leq1$ implies an important fact concerning the connectivity weights:
\begin{conj}\label{assumeconj} For any $\varsigma\in\{1,2,\ldots,C_N\}$, there is an SLE$_\kappa$ partition function $F\in\mathcal{S}_N$ such that $P_\varsigma(\boldsymbol{x})>0$ for all $\boldsymbol{x}\in\Omega_0$.  (Without loss of generality, we assume that $F$ is positive-valued, so $[\mathscr{L}_\vartheta]F\geq0$ for all $\vartheta\in\{1,2,\ldots,C_N\}$.)
\end{conj}
\noindent
Assuming that these two conjectures are true, the positivity of $P_\varsigma$ implies via (\ref{xing}) that $[\mathscr{L}_{\varsigma}]F>0$ (where $F$ is given by conjecture \ref{assumeconj}), so $\Pi_\varsigma$ is positive-valued too.  Hence, each connectivity weight, as formalized in definition \ref{dualbasis}, is an SLE$_\kappa$ partition function (definition \ref{partitiondefn}).  Also, in the multiple-SLE$_\kappa$ process with $\Pi_\vartheta$ for its partition function, we have 
\be\label{dualprob} \left.\begin{array}{l}\text{formula (\ref{xing}) with $F=\Pi_\vartheta$} \\ \text{dual relation $[\mathscr{L}_\varsigma]\Pi_\vartheta=\delta_{\varsigma,\vartheta}$}\end{array}\right\}\quad\Longrightarrow\quad P_\varsigma=\delta_{\varsigma,\vartheta}.\ee
That is, the growing curves of the multiple-SLE$_\kappa$ process with partition function $\Pi_\vartheta$ join pairwise in $\vartheta$th connectivity (as in the diagram for $\Pi_\vartheta$) almost surely.  This property is what informally defines the connectivity weights before their formal definition \ref{dualbasis}.  A proof of conjectures \ref{connectivityconj} and \ref{assumeconj} would thus confirm that these two definitions are equivalent.

It is interesting to consider the multiple-SLE$_\kappa$ process with partition function $\mathcal{F}_\vartheta\in\mathcal{B}_N$.  (Here, we assume that the connectivity weights are indeed positive-valued, as the previous paragraph suggests, and that $n(\kappa)>0$ so $[\mathscr{L}_\varsigma]\mathcal{F}_\vartheta>0$ too (\ref{LkFk}).  Then it follows from its decomposition (\ref{decompose}) over $\mathscr{B}_N$ that $\mathcal{F}_\vartheta$ is an SLE$_\kappa$ partition function.)  Thanks to (\ref{LkFk}), the crossing probability formula (\ref{xing}, \ref{xing2}) with $F=\mathcal{F}_\vartheta$ becomes
\be\label{xing3}P_\varsigma=\frac{n^{l_{\varsigma,\vartheta}}\Pi_\varsigma}{\mathcal{F}_\vartheta}=\frac{n^{l_{\varsigma,\vartheta}}\Pi_\varsigma}{n^{l_{1,\vartheta}}\Pi_1+n^{l_{2,\vartheta}}\Pi_2+\dotsm+n^{l_{C_N,\vartheta}}\Pi_{C_N}}.\ee
In \cite{fkz}, we interpret the $\varrho$th term $n^{l_{\varrho,\vartheta}}\Pi_\varrho$ appearing in the denominator of (\ref{xing3}) as the partition function for a loop-gas model in the upper half-plane, where ``boundary loops" with fugacity $n$ pass into and out of the system through the points $x_1$, $x_2,\ldots,x_{2N}$ and join these points together in the $\varrho$th (resp.\ $\vartheta$th) connectivity in the upper (resp.\ lower) half-plane.  Such partition functions are related to Potts-model and random-cluster-model partition functions \cite{fkz}.

\section{Pure CFT and pure multiple-SLE$_\kappa$ intervals}\label{intervalsect}

In this article and its predecessors \cite{florkleb,florkleb2,florkleb3}, we imagine the point $\boldsymbol{x}=(x_1,x_2,\ldots,x_{2N})$   in the domain $\Omega_0$ (\ref{components}) of $F\in\mathcal{S}_N$ as a collection of adjacent intervals $(x_1,x_2)$, $(x_2,x_3),\ldots,(x_{2N-1},x_{2N})$, and $(x_{2N},x_1)$ (with the last containing infinity) on the real axis.  Typically, some of these intervals have special properties that distinguish them from others  under the lens of either multiple SLE$_\kappa$ or  CFT.  For example, if $F$ is an SLE$_\kappa$ partition function, then an interval may almost surely share its endpoints with a boundary arc after the multiple-SLE$_\kappa$ process completes, or it may not almost surely.  In the former  (resp.\ latter) case, we call the boundary arc (resp.\ pair of distinct boundary arcs) attached to the interval \emph{contractible} (resp.\ \emph{propagating}), and we call the interval itself by the same name.  For another example, if $F\in\mathcal{S}_N$ is the $2N$-point CFT correlation function (\ref{2Npoint}), then as we collapse an interval so the one-leg boundary operators at its endpoints fuse, the resulting OPE may contain only the conformal family of either the identity operator or the two-leg operator.  In the former (resp.\ latter) case, we call the interval an \emph{identity} (resp.\ \emph{two-leg}) \emph{interval} of $F$.  (Actually, we already defined the terms ``identity interval" and ``two-leg interval" in \cite{florkleb}, and in this section, we motivate these definitions.)  A solution with every other interval pure in this CFT sense is called a ``conformal block."  Actually, in CFT, such functions are not necessarily correlation functions like (\ref{2Npoint}).  Rather, they appear as building blocks of such correlation functions.  In general, not every (and indeed possibly no) interval of $F$ must  be one of these mentioned types, but if it is, then we think of it as ``pure" in the multiple-SLE$_\kappa$ sense or in the CFT sense accordingly.

\begin{figure}[t]
\centering
\includegraphics[scale=0.3]{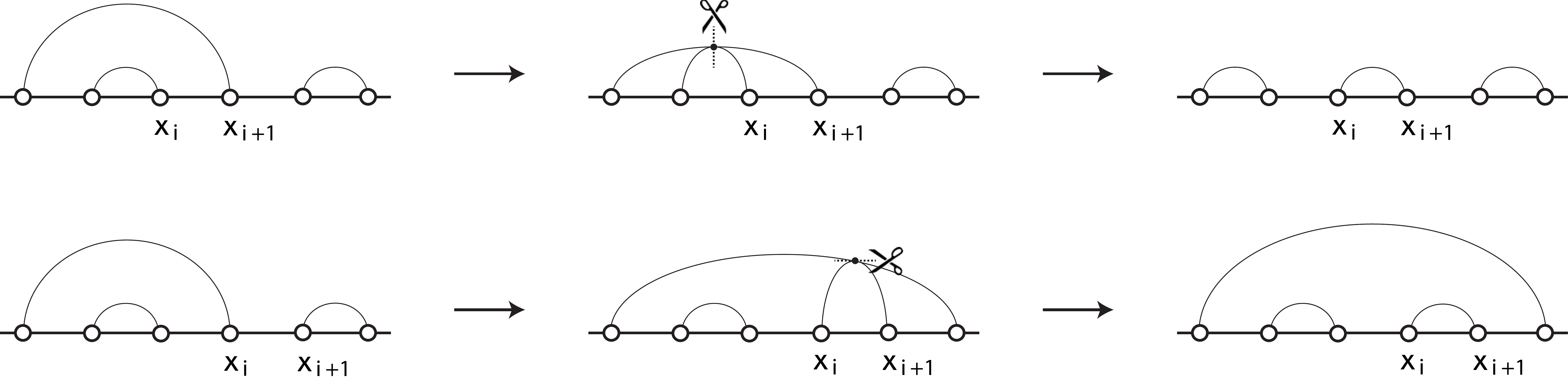}
\caption{Illustrations of the map $\chi$ described in item \ref{cutmap} above definition \ref{sleintervaldefn}.  This map sends the $\varsigma$th connectivity, with $\varsigma\in\{C_{N-1}+1,C_{N-1}+2,\ldots,C_N\}$, to the $\chi(\varsigma)$th connectivity, with $\chi(\varsigma)\in\{1,2,\ldots,C_{N-1}\}$.}
\label{CutMap}
\end{figure}

Throughout this section and appendix \ref{appendix}, we use the following indexing convention and index map $\chi$.
\begin{enumerate}
\item\label{indexorder} We index the $C_N$ available arc connectivities as  in the proof of theorem \ref{xingasymplem} above.  Thus, a unique arc (resp.\ no arc) in the $\vartheta$th arc connectivity has both of its endpoints at $x_i$ and $x_{i+1}$ if $\vartheta\leq C_{N-1}$ (resp.\ $\vartheta>C_{N-1}$).
\item\label{cutmap} We let $\chi:\{C_{N-1}+1,C_{N-1}+2,\ldots,C_N\}\rightarrow\{1,2,\ldots,C_{N-1}\}$ send the $\varsigma$th connectivity to the $\chi(\varsigma)$th connectivity by following this two-step process.  First, we pinch together at a point $p$ the two propagating arcs that share their endpoints with $(x_i,x_{i+1})$ in the $\varsigma$th connectivity diagram.  Then we cut these arcs apart at $p$ and separate them into a contractible arc, sharing its endpoints with $(x_i,x_{i+1})$, and another arc, terminating at the other endpoints of the original two propagating arcs (figures \ref{CutMap}).  The arc connectivity that results is the $\chi(\varsigma)$th connectivity.
\end{enumerate}

\subsection{Pure multiple-SLE$_\kappa$ intervals}\label{puresle}

To begin, we investigate the notion of interval purity in multiple SLE$_\kappa$.  As we previously discussed, we distinguish between two pure interval types: contractible and propagating.

These two interval types are easiest to understand at first in terms of connectivity weights.  Indeed (assuming that conjecture \ref{connectivityconj} is true), formula (\ref{dualprob}) shows that the boundary arcs grown by the multiple-SLE$_\kappa$ process with $\Pi_\varsigma$ for its partition function almost surely join together the endpoints of the intervals in the $\varsigma$th connectivity.  Hence, if a contractible boundary arc (resp.\ an interior arc) shares its endpoints with $(x_i,x_{i+1})$ in the $\varsigma$th connectivity (resp.\ in the half-plane diagram for $\Pi_\varsigma$ (definition \ref{dualbasis})), then this interval must be a contractible interval of $\Pi_\varsigma$.  On the other hand, if a pair of propagating boundary arcs share their (resp.\ no interior arc shares its) endpoints with $(x_i,x_{i+1})$ in the $\varsigma$th connectivity (resp.\ in the half-plane diagram for $\Pi_\varsigma$), then this interval must be a propagating interval of $\Pi_\varsigma$.

More generally, if $F$ is some SLE$_\kappa$ partition function and the probability that a contractible boundary arc shares its endpoints with $(x_i,x_{i+1})$ is either zero or one in the ensuing multiple-SLE$_\kappa$ process, then we may extend these terms to $F$ as well.  We do this in the following definition.
\begin{defn}\label{sleintervaldefn} With $\kappa\in(0,8)$, we select an $F\in\mathcal{S}_N$, and an $i\in\{1,2,\ldots,2N\}$ (below, we identify $i+1=2N+1$ with one).  We say that $(x_i,x_{i+1})$ is a \emph{contractible} (resp.\ \emph{propagating}) \emph{interval of $F\in\mathcal{S}_N\setminus\{0\}$} if, in the decomposition 
\be\label{FdecompPi}F=a_1\Pi_1+a_2\Pi_2+\dotsm+a_{C_N}\Pi_{C_N}\ee
of $F$ over the basis $\mathscr{B}_N$, the half-plane diagram for $\Pi_\varsigma$ has an (resp.\ no) arc with endpoints at $x_i$ and $x_{i+1}$ whenever $a_\varsigma\neq0$.  Also, we say that $(x_i,x_{i+1})$ is a \emph{propagating interval} of the trivial solution $0\in\mathcal{S}_N$.
\end{defn}

Now, we verify that definition \ref{sleintervaldefn} carries the meaning intended by the opening paragraph of section \ref{intervalsect}.  For this purpose, we index the arc connectivities according to item \ref{indexorder} just above this definition, and we use the sets of connectivity weights $\mathscr{C}_N=\{\Pi_1,\Pi_2,\ldots,\Pi_{C_{N-1}}\}\subset\mathscr{B}_N$ and $\mathscr{B}_{N-1}=\{\Xi_1,\Xi_2,\ldots,\Xi_{C_{N-1}}\}$ as defined in the proof of theorem \ref{xingasymplem}.  Thus, the decomposition (\ref{FdecompPi}) of $F$ over $\mathscr{B}_N$ sorts into
\be\label{FdecomposePi}F\,\,=\underbrace{a_1\Pi_1+a_2\Pi_2+\dotsm+a_{C_{N-1}}\Pi_{C_{N-1}}}_{\text{$(x_i,x_{i+1})$ a contractible interval of $\Pi_\varsigma\in\mathscr{C}_N$}}+\,\,\,\underbrace{a_{C_{N-1}+1}\Pi_{C_{N-1}+1}+a_{C_{N-1}+2}\Pi_{C_{N-1}+2}+\dotsm+a_{C_N}\Pi_{C_N}}_{\text{$(x_i,x_{i+1})$ a propagating interval of $\Pi_\varsigma\in\mathscr{B}_N\setminus\mathscr{C}_N$}}\ee
for some real constants $a_1$, $a_2,\ldots,a_{C_N}$.  Then according to definition \ref{sleintervaldefn}, for some SLE$_\kappa$ partition function $F\in\mathcal{S}_N\setminus\{0\}$, $(x_i,x_{i+1})$ is
\begin{enumerate}[I.]
\item\label{intervalit1}  a contractible interval of $F$ if $a_\varsigma=0$ for all $\varsigma>C_{N-1}$ and $a_\varsigma\neq0$ for some $\varsigma\leq C_{N-1}$ in (\ref{FdecomposePi}).  As such, the conjectured formula (\ref{xing2}) shows that the probability of the multiple-SLE$_\kappa$ process with partition function $F(\boldsymbol{x})$ growing a contractible boundary arc that shares its endpoints with $(x_i,x_{i+1})$ is $P_1+P_2+\dotsm+P_{C_{N-1}}=1$.
\item\label{intervalit2} a propagating interval of $F$ if $a_\varsigma=0$ for all $\varsigma\leq C_{N-1}$ and $a_\varsigma\neq0$ for some $\varsigma>C_{N-1}$ in (\ref{FdecomposePi}).   As such, the conjectured formula (\ref{xing2}) shows that the probability of the multiple-SLE$_\kappa$ process with partition function $F(\boldsymbol{x})$ growing a pair of propagating boundary arcs from $x_i$ and $x_{i+1}$ is $P_{C_{N-1}+1}+P_{C_{N-1}+2}+\dotsm+P_{C_N}=1$.
\end{enumerate}
Thus, definition \ref{sleintervaldefn} indeed carries the meaning intended by the opening paragraph of section \ref{intervalsect}.  We note that, while every interval of a connectivity weight is either contractible or propagating, not every interval of each $F\in\mathcal{S}_N$ necessarily falls under one of these two categories.  (Also, although definition \ref{sleintervaldefn} technically extends the meaning of multiple-SLE$_\kappa$ interval purity beyond SLE$_\kappa$ partition functions (definition \ref{partitiondefn}) to all elements of $\mathcal{S}_N$, this extension is really not so useful because only SLE$_\kappa$ partition functions generate the multiple-SLE$_\kappa$ process \cite{bbk}.)

It is interesting to study the effect of shrinking a contractible interval $(x_i,x_{i+1})$ of some SLE$_\kappa$ partition function.  This scenario falls under item \ref{intervalit1} above.  If $x_{i+1}-x_i\ll |x_i-x_j|$ for all $j\not\in\{i,i+1\}$, then with high probability, the contractible boundary arc anchored to $x_i$ and $x_{i+1}$ explores only the area in the upper half-plane near these points, and if $x_{i+1}\rightarrow x_i$, then it contracts to a point.  As such, $P_\varsigma$ with $\varsigma\leq C_{N-1}$ must go to the probability $Q_\varsigma$ of the arc connectivity generated by contracting away this minuscule, isolated boundary arc of the $\varsigma$th connectivity.  Using (\ref{Pifactor}), we see that the conjectured formula for $P_\varsigma$ (\ref{xing2}) has this necessary feature, which further supports conjecture \ref{connectivityconj}:
\be\label{oldxing}P_\varsigma\quad\xrightarrow[x_{i+1}\rightarrow x_i]{}\quad\dfrac{a_\varsigma\Xi_\varsigma}{a_1\Xi_1+a_2\Xi_2+\dotsm+a_{C_{N-1}}\Xi_{C_{N-1}}}=Q_\varsigma,\quad \varsigma\leq C_{N-1}.\ee

Moreover, it is interesting to study the effect of shrinking a propagating interval $(x_i,x_{i+1})$ of some SLE$_\kappa$ partition function.  This scenario falls under item \ref{intervalit2} above.  If $x_{i+1}\rightarrow x_i$, then the pair of propagating boundary arcs anchored to these points join into one arc that touches the real axis at $x_i$, and if we detach that arc from this contact point, then we generate a new connectivity event involving the $2N-2$ remaining points in $\{x_j\}_{j\neq i,i+1}$.  This is in fact the $\chi(\varsigma)$th connectivity event with the contractible arc joining $x_i$ with $x_{i+1}$ removed, where we define the map $\chi$ in item \ref{cutmap} above definition \ref{sleintervaldefn} (figure \ref{CutMap}).  According to (\ref{Pivanish}), the formula for $P_\varsigma$ goes to
\be\label{newxing}P_\varsigma\quad\xrightarrow[x_{i+1}\rightarrow x_i]{}\quad\dfrac{a_\varsigma\Lambda_\varsigma}{a_{C_{N-1}+1}\Lambda_{C_{N-1}+1}+a_{C_{N-1}+2}\Lambda_{C_{N-1}+2}+\dotsm+a_{C_N}\Lambda_{C_N}},\quad \varsigma>C_{N-1},\ee
and we interpret the right side of (\ref{newxing}) as the probability of the $\chi(\varsigma)$th connectivity event involving the points in $\{x_j\}_{j\neq i,i+1}$, conditioned on a particular boundary arc among the $N-1$ available touching the real axis at $x_i$.  Thus, some limits of the crossing-probability formulas (\ref{xing}, \ref{xing2}) may give formulas for crossing probabilities conditioned on boundary visitation events.  Ref.\ \cite{nmk} explicitly constructs formulas for some such probabilities by other means.

\begin{figure}[t]
\centering
\includegraphics[scale=0.35]{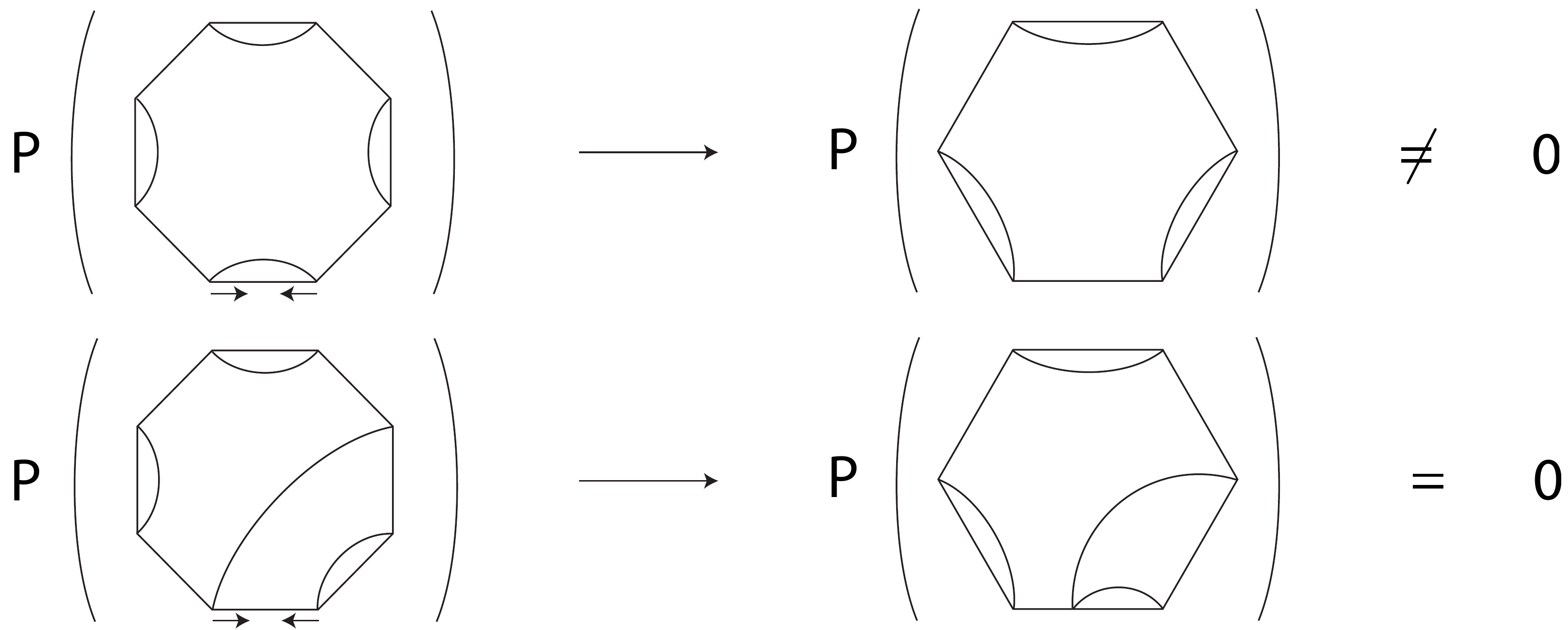}
\caption{If the bottom side of the octagon is not (resp.\ is) a two-leg interval of SLE$_\kappa$ partition function $F\in\mathcal{S}_N$, then the shown octagon $(N=4)$ crossing probability (\ref{xing2}) goes to a hexagon crossing probability (resp.\ zero) as we collapse that side.}
\label{ProbLim}
\end{figure}

Finally, it is interesting to study the effect of shrinking an interval $(x_i,x_{i+1})$ of some SLE$_\kappa$ partition function $F$ that is neither contractible nor propagating.  Such intervals are not pure in the multiple-SLE$_\kappa$ sense described above, but they are combinations of the two possible kinds of pure intervals, contractible and propagating.  As such, they inherit their features from the pure intervals described above, so it is appropriate to discuss them here.  According to definition \ref{sleintervaldefn}, $(x_i,x_{i+1})$ is neither a contractible interval nor a propagating interval of $F$ if $a_\varsigma\neq0$ for some $\varsigma\leq C_{N-1}$ and also for some $\varsigma>C_{N-1}$ in (\ref{FdecomposePi}).  As such, the conjectured formula (\ref{xing2}) shows that the probability of the multiple-SLE$_\kappa$ process with partition function $F(\boldsymbol{x})$ growing a contractible boundary arc that shares its endpoints with (resp.\ growing a pair of propagating boundary arcs from the endpoints of) $(x_i,x_{i+1})$ is $P_1+P_2+\dotsm+P_{C_{N-1}}$ (resp.\ $P_{C_{N-1}+1}+P_{C_{N-1}+2}+\dotsm+P_{C_N}$).  Furthermore, if $x_{i+1}\rightarrow x_i$, then 
\be\label{Plimit}
P_\varsigma\quad\xrightarrow[x_{i+1}\rightarrow x_i]{}\quad\begin{cases}Q_\varsigma=\dfrac{a_\varsigma\Xi_\varsigma}{a_1\Xi_1+a_2\Xi_2+\dotsm+a_{C_{N-1}}\Xi_{C_{N-1}}}, & \varsigma\leq C_{N-1} \\ 0, & \varsigma>C_{N-1}\end{cases}.
\ee
That is, the probability $P_\varsigma$ of the $\varsigma$th connectivity event with $\varsigma\leq C_{N-1}$, where a contractible boundary arc shares its endpoints with $(x_i,x_{i+1})$, approaches the probability $Q_\varsigma<1$ of the connectivity event generated by contracting this boundary arc to a point as $x_{i+1}\rightarrow x_i$.  And also, the probability $P_\varsigma$ of the $\varsigma$th connectivity event with $\varsigma>C_{N-1}$, where a pair of propagating boundary arcs shares endpoints with $(x_i,x_{i+1})$, vanishes as $x_{i+1}\rightarrow x_i$ (figure \ref{ProbLim}).

\begin{figure}[b]
\centering
\includegraphics[scale=0.3]{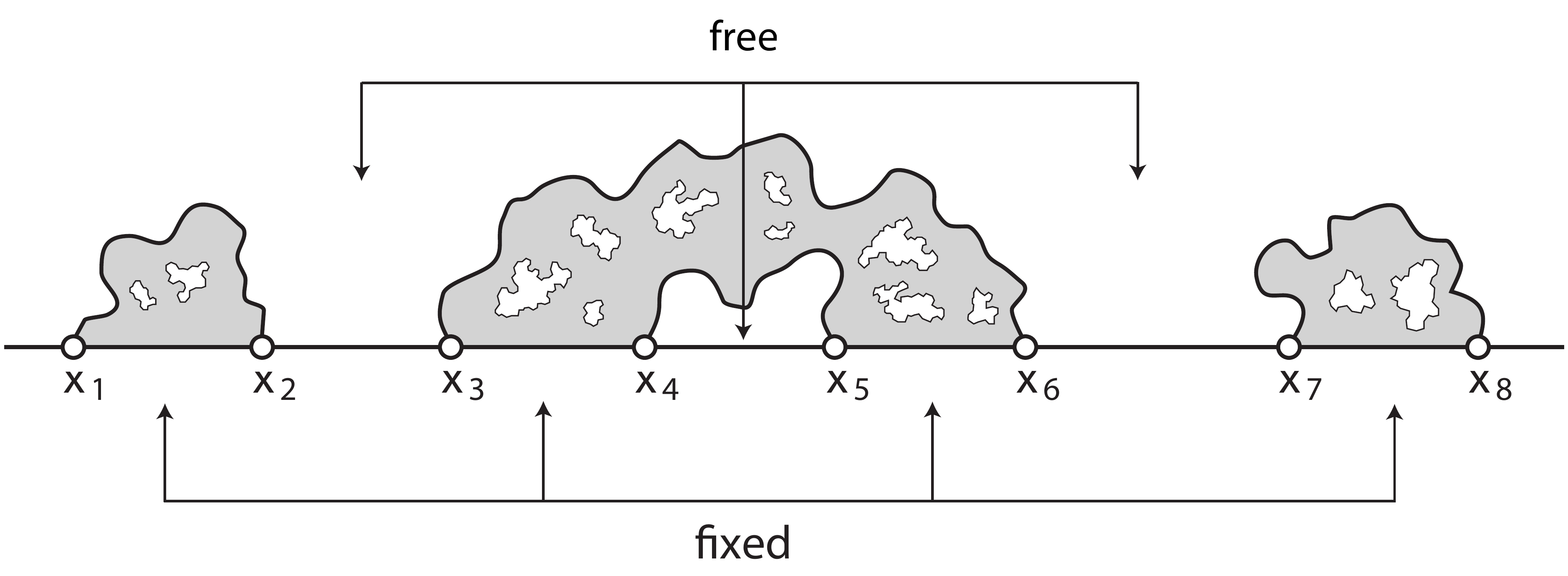}
\caption{Boundary clusters (gray) anchor to fixed intervals, and multiple-SLE$_\kappa$ boundary arcs (thick curves) trace the outer part of these clusters' boundaries.  Boundary cluster connectivities thus correspond one-to-one with boundary arc connectivities.}
\label{FFBCholes}
\end{figure}

The limit (\ref{Plimit}) also has a statistical mechanics interpretation (i.e. for critical percolation, Potts model, random cluster model, etc.).  For this, we choose a partition function $F\in\mathcal{S}_N$ that conditions a critical lattice model in the upper half-plane on a particular side-alternating boundary condition event.  This is an event in which every other interval, say, $\ldots(x_{i-2},x_{i-1}),$ $(x_i,x_{i+1})$, $(x_{i+2},x_{i+3}),\ldots$ exhibits the fixed state while the other intervals exhibit a different state.  In some models, this is the ``free" state (percolation, random cluster model) \cite{car1,car2,salbau}, and in others, it is the fluctuating state (Potts model) \cite{gamcar,fkz}, a distinction that we avoid in our previous article \cite{florkleb}.  (In \cite{fkz}, we identify each element of $\mathcal{B}_N$ with such a partition function.)  Then a (percolation, spin, FK, etc.) ``boundary cluster" anchors to each fixed interval, including $(x_i,x_{i+1})$, and $P_\varsigma$ (\ref{xing}) gives the probability that these clusters join the wired intervals in the $\varsigma$th connectivity (or ``crossing") event \cite{florkleb} (figure \ref{FFBCholes}).  Also, each boundary arc with an endpoint at $x_i$ or $x_{i+1}$ traces part of this cluster's boundary in the upper half-plane.  If $\varsigma\leq C_{N-1}$, then this boundary arc is contractible, and the boundary cluster anchored to $(x_i,x_{i+1})$ is contained within the region between the arc and the interval.  As such, the boundary cluster contracts to a point as $x_{i+1}\rightarrow x_i$, so this limit sends $P_\varsigma$ to the probability $Q_\varsigma$ (\ref{Plimit}) of the crossing event created by removing this lone boundary cluster.  On the other hand, if $\varsigma>C_{N-1}$, then a pair of propagating boundary arcs anchor to $x_i$ and $x_{i+1}$, and the boundary cluster anchored to $(x_i,x_{i+1})$ is contained between them.  As such, the boundary cluster anchored to $(x_i,x_{i+1})$ touches other fixed intervals, and after sending $x_{i+1}\rightarrow x_i$, we find this cluster touching the system boundary exactly at the point $x_i$ within a free boundary segment, an event with probability zero.  Hence, this limit sends $P_\varsigma$ to zero (\ref{Plimit}). 

\subsection{Pure CFT intervals}\label{purecft}

Now we investigate the notion of interval purity in CFT.  As we previously discussed, we distinguish between two different interval types: two-leg versus identity.  (Although CFT inspires these terms, our definitions for them do not use CFT.)  As we will see, there is a straightforward definition for the former that applies to all $\kappa\in(0,8)$.  (In fact, we already gave it as part of definition \red{13} in \cite{florkleb}.)  However, there seems to be a suitable definition for the latter only if $\kappa\in(0,8)$ \emph{and} $8/\kappa\not\in2\mathbb{Z}^++1$.  We approach these definitions in two different ways.

The more straightforward approach is to examine the Frobenius series expansions (\ref{nolog}, \ref{stillnolog}, \ref{log}) of solutions $F\in\mathcal{S}_N\setminus\{0\}$ in powers of $x_{i+1}-x_i$.  To begin, we suppose that $8/\kappa\not\in\mathbb{Z}^+$.  Then as discussed beneath the proof of theorem \ref{frobseriescor} in section \ref{frobsect}, the two sums in (\ref{nolog}) correspond to different conformal families appearing in the OPE of the one-leg boundary operators $\psi_1(x_i)$ and $\psi_1(x_{i+1})$ at the endpoints of $(x_i,x_{i+1})$.  Indeed, the first sum, with its indicial power $1-6/\kappa=-2\theta_1+\theta_0$ (\ref{indicialpowers}), corresponds to the identity family, and the second sum, with its indicial power $2/\kappa=-2\theta_1+\theta_2$ (\ref{indicialpowers}), corresponds to the two-leg family.  As such, if only the first sum vanishes, then the identity family is absent from the OPE of $\psi_1(x_i)$ with $\psi_1(x_{i+1})$, but the two-leg family is present.  Therefore, we call $(x_i,x_{i+1})$ a ``two-leg interval of $F$."  According to item \ref{frobitem1} of theorem \ref{frobseriescor}, this situation arises only if
\be\label{a0vanish}A_0\,\,\,=\lim_{x_{i+1}\rightarrow x_i}(x_{i+1}-x_i)^{6/\kappa-1}F(\boldsymbol{x})=0.\ee
Actually, condition (\ref{a0vanish}) correctly identifies $(x_i,x_{i+1})$ as a two-leg interval of $F$ for $8/\kappa\in\mathbb{Z}^+$ too.  Indeed, in this situation, the indicial powers (\ref{indicialpowers}) differ by an integer, so terms from the identity family mix with terms of the two-leg family to give the second sum in (\ref{stillnolog}, \ref{log}).  But if $A_0=0$, then the first sum in (\ref{stillnolog}, \ref{log}) vanishes, and because only the identity family contributes to it, the identity family must therefore be absent from the OPE of $\psi_1(x_i)$ with $\psi_1(x_{i+1})$.

Now, if $8/\kappa\not\in\mathbb{Z}^+$ and only the second sum in (\ref{nolog}) vanishes, then evidently the two-leg family is absent from the OPE of $\psi_1(x_i)$ with $\psi_1(x_{i+1})$, but the identity family is present.  Therefore, we call $(x_i,x_{i+1})$ an ``identity interval of $F\in\mathcal{S}_N\setminus\{0\}$."  Now, this second sum of (\ref{nolog}) vanishes if the function $(x_{i+1}-x_i)^{6/\kappa-1}F(\boldsymbol{x})$, continuously extended to the part of the boundary of $\Omega_0$ with only the coordinates $x_i$ and $x_{i+1}$ of $\boldsymbol{x}$ equal, is analytic at $x_{i+1}=x_i$.  If it is, then we deem $(x_i,x_{i+1})$ an identity interval of $F$ in item \ref{cftinterval2a} of definition \ref{cftintervaldefn} below.

\begin{figure}[t]
\centering
\includegraphics[scale=0.3]{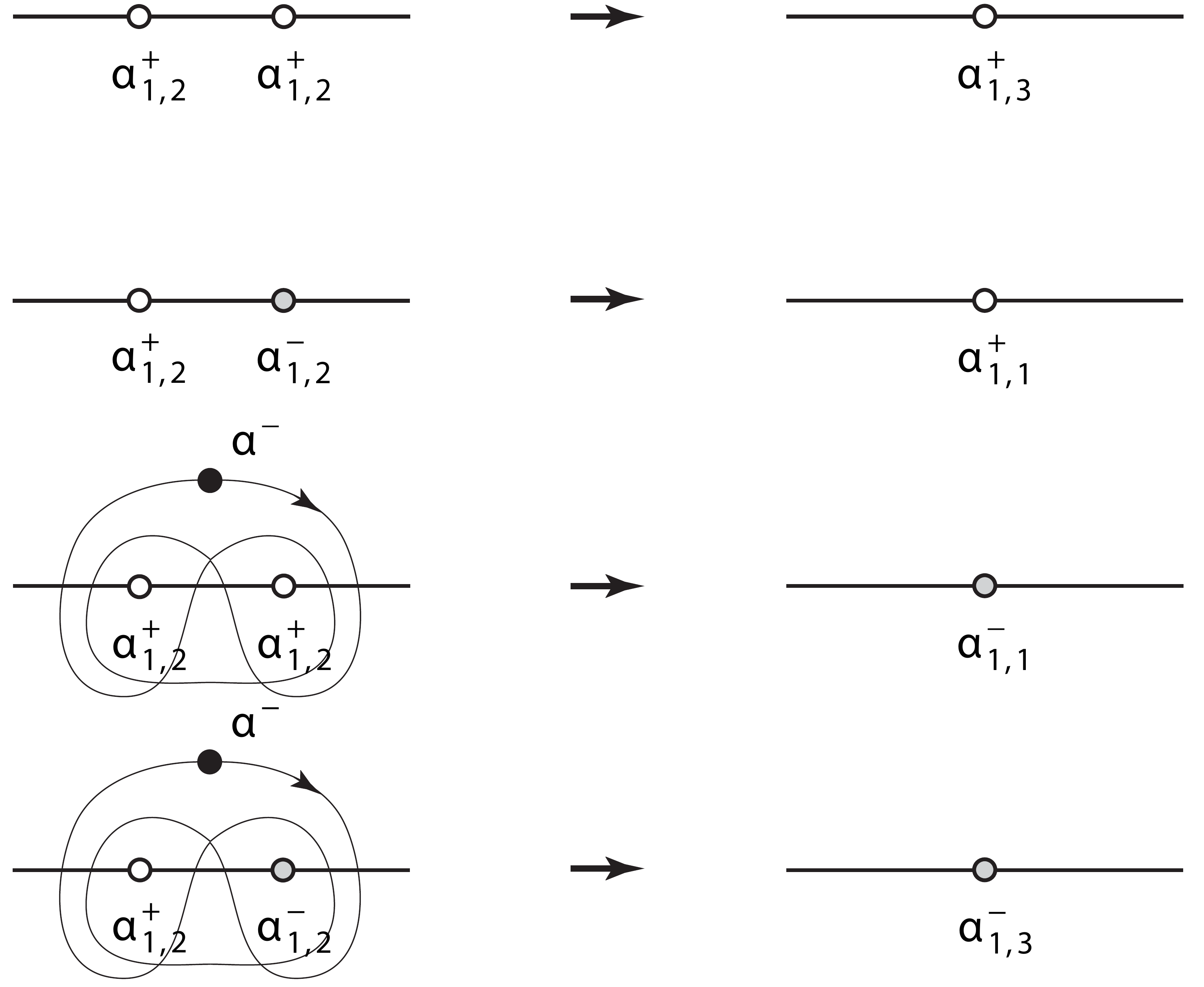}
\caption{Fusion rules for two $V_{1,2}^\pm$ chiral operators, corresponding from top to bottom with (\ref{prodfuse1}--\ref{prodfuse4}).  In the bottom two illustrations, a screening charge entwines these two chiral operators, altering the fusion rules.}
\label{Fusion}
\end{figure}

Unfortunately, this approach to defining an identity interval does not work if $8/\kappa\in\mathbb{Z}^+$.  Indeed, if $8/\kappa\in2\mathbb{Z}^+$, then $(x_{i+1}-x_i)^{6/\kappa-1}F(\boldsymbol{x})$ is always analytic at $x_{i+1}=x_i$ thanks to (\ref{stillnolog}), and if $8/\kappa\in2\mathbb{Z}^++1$ and $A_0\neq0$ (\ref{a0vanish}), then this function is never analytic there thanks to (\ref{nolog}).  To extend the term ``identity interval" at least to $8/\kappa\in2\mathbb{Z}^+$, we re-examine the intended meanings of both this term and the term ``two-leg interval" through the CFT Coulomb gas formalism.  (The next paragraphs assume familiarity with the contents and notations in section \red{II} of \cite{florkleb3}.  Throughout, we use notation for the dense phase $(\kappa>4)$.  In the dilute phase ($\kappa\leq4$), the notation is different, but all formulas and results are the same.  Finally, the  discussion assumes $8/\kappa\not\in2\mathbb{Z}^++1$. We discuss the case $8/\kappa\in2\mathbb{Z}^++1$ in section \ref{8kappaodd} of appendix \ref{appendix}.)  Interpreting a Coulomb gas function $F\in\mathcal{S}_N\setminus\{0\}$ as a correlation function of $2N$ chiral operators and $N-1$ screening operators (as usual, we denote the one point that bears the conjugate charge as $x_c$)
\begin{multline}\label{corrvert}F(x_1,x_2,\ldots,x_{2N})=\langle V_{1,2}^+(x_1)V_{1,2}^+(x_2)\dotsm\\
\dotsm V_{1,2}^+(x_{c-1})V_{1,2}^-(x_c)V_{1,2}^+(x_{c+1})\dotsm V_{1,2}^+(x_{2N-1})V_{1,2}^+(x_{2N})Q_1^-Q_2^-\dotsm Q^-_{N-1}\rangle,\end{multline}
 we investigate how the manner in which a screening operator's integration contour surrounds $x_i$ and $x_{i+1}$ determines what type of pure CFT interval of $F$, if any, the interval $(x_i,x_{i+1})$ is.  Figure \ref{Fusion} shows the four simplest cases, and as we send $x_{i+1}\rightarrow x_i$, their fusion products are as follows (in the same top-to-bottom order as figure \ref{Fusion} shows):
\begin{align}\label{prodfuse1}&i\not\in\{c,c-1\}:&& \alpha_{1,2}^++\alpha_{1,2}^+=\alpha_{1,3}^+&&\Longrightarrow&& V_{1,2}^+(x_i)V_{1,2}^+(x_{i+1})&&\underset{x_{i+1}\rightarrow x_i}{\sim}&&V_{1,3}^+(x_i),\\
\label{prodfuse2}&i\in\{c,c-1\}:&&\alpha_{1,2}^++\alpha_{1,2}^-=\alpha_{1,1}^+&&\Longrightarrow&& V_{1,2}^+(x_i)V_{1,2}^-(x_{i+1})&&\underset{x_{i+1}\rightarrow x_i}{\sim}&&V_{1,1}^+(x_i),\\
\label{prodfuse3}&i\not\in\{c,c-1\}:&&\alpha_{1,2}^++\alpha_{1,2}^++\alpha^-=\alpha_{1,1}^-&&\Longrightarrow&& V_{1,2}^+(x_i)V_{1,2}^+(x_{i+1})Q^-&&\underset{x_{i+1}\rightarrow x_i}{\sim}&&V_{1,1}^-(x_i),\\
\label{prodfuse4}&i\in\{c,c-1\}:&&\alpha_{1,2}^++\alpha_{1,2}^-+\alpha^-=\alpha_{1,3}^-&&\Longrightarrow&& V_{1,2}^+(x_i)V_{1,2}^-(x_{i+1})Q^-&&\underset{x_{i+1}\rightarrow x_i}{\sim}&&V_{1,3}^-(x_i).\end{align}
In first two cases (\ref{prodfuse1}, \ref{prodfuse2}), no integration contour of $F$ surrounds $x_i$ or $x_{i+1}$.  In the last two cases (\ref{prodfuse3}, \ref{prodfuse4}), a single Pochhammer contour entwines $x_i$ with $x_{i+1}$, and sending $x_{i+1}\rightarrow x_i$ contracts this contour to the point $x_i$, draws the screening operator $Q^-$ in with the fusion, and thus adds $\alpha^-$ to the total charge of the fusion product.

Because each of the fusion products of (\ref{prodfuse1}--\ref{prodfuse4}) contains just the conformal family of the chiral operator appearing on the right side of these equations, the interval $(x_i,x_{i+1})$ is pure in the CFT sense. Moreover, the conformal weights of these products are the Kac weights (\ref{hrs}) (again, in the dense phase $\kappa>4$)
\be\label{Kacref} V_{1,3}^\pm\quad\Longleftrightarrow\quad h_{1,3}=8/\kappa-1,\qquad V_{1,1}^\pm\quad\Longleftrightarrow\quad h_{1,1}=0,\ee
which we respectively identify with the boundary two-leg weight $\theta_2$ and identity weight $\theta_0$ \cite{florkleb, florkleb2}.  We infer from (\ref{prodfuse1}) (resp.\ (\ref{prodfuse3})) and (\ref{Kacref}) that $(x_i,x_{i+1})$ is a two-leg (resp.\ an identity) interval of $F$ if no (resp.\ an) integration contour of this function shares its endpoints with this interval and $i\not\in\{c,c-1\}$.  We also infer from (\ref{prodfuse2}) (resp.\ (\ref{prodfuse4})) and (\ref{Kacref}) that $(x_i,x_{i+1})$ is an identity (resp.\ a two-leg) interval of $F$ if no (resp.\ an) integration contour shares both of its endpoints with this interval and $i\in\{c,c-1\}$.  Ref.\ \cite{js} discusses rules (\ref{prodfuse1}--\ref{prodfuse4}) (and more) for percolation $(\kappa=6)$.

Finally, if an integration contour of $F$ entwines points in $\{x_j\}_{j\neq i,i+1}$ in addition to $x_i$ and/or $x_{i+1}$, then whether or not $(x_i,x_{i+1})$ is a pure CFT interval of $F$ is not clear.  This ambiguity arises because such a contour does not contract to a point as we send $x_{i+1}\rightarrow x_i$, so the screening operator that traces it is not completely drawn into the fusion.  For instance, $(x_i,x_{i+1})$ is usually not a pure CFT interval of a case \ref{sc3} or case \ref{sc4} element of $\mathcal{B}_N$ (section \ref{frobsect}).  But if $\kappa$ is an exceptional speed (\ref{exceptional}), then $(x_i,x_{i+1})$ might be such an interval.  The simplest example of this is $\kappa=\kappa_{3,2}=6$, where every element of $\mathcal{B}_N$ equals one (section \ref{n32sect}).  Because all case \ref{sc3} and case \ref{sc4} elements of $\mathcal{B}_N$ equal its case \ref{sc2} elements and $(x_i,x_{i+1})$ is an identity interval of the latter, it follows $(x_i,x_{i+1})$ is an identity interval of the former too.

If an arc in the half-plane diagram for $\mathcal{F}_\vartheta\in\mathcal{B}_N$ shares its endpoints with the interval $(x_i,x_{i+1})$, then either a unique Pochhammer contour shares its endpoints with this interval too, or there is no such contour and $x_c\in\{x_i,x_{i+1}\}$.  (See the discussion beneath the formula (\ref{Fexplicit}) for $\mathcal{F}_\vartheta$ in the introduction \ref{intro}.)  Supposing still that $8/\kappa\not\in2\mathbb{Z}^++1$, we find fusion rules (\ref{prodfuse2}, \ref{prodfuse3}) as we send $x_{i+1}\rightarrow x_i$, so $(x_i,x_{i+1})$ must be an identity interval of $\mathcal{F}_\vartheta$.  (If $8/\kappa\in2\mathbb{Z}^++1$, then we find that $(x_i,x_{i+1})$ is a two-leg interval of $\mathcal{F}_\vartheta$ instead.  See section \ref{s2} of appendix \ref{appendix}.)  Similarly, if $F\in\mathcal{S}_N$ equals a linear combination of elements of $\mathcal{B}_N$ that all have this property, then $(x_i,x_{i+1})$ must be an identity interval of $F$.  Thus, we find a second candidate definition for the term ``identity interval," which we state as item \ref{cftinterval2b} in definition \ref{cftintervaldefn} below.  Unlike the first candidate, the second includes the case $8/\kappa\in2\mathbb{Z}^+$.  However, it does not include the exceptional speeds (\ref{exceptional}) because $\mathcal{B}_N$ is then often not a basis for $\mathcal{S}_N$.

With these considerations, we endow the terms ``two-leg interval" and ``identity interval" with formal definitions following from the above discussion.  Together with the usual restriction $\kappa\in(0,8)$, we state the first candidate definition of an identity interval, for $8/\kappa\not\in\mathbb{Z}^+$, in item \ref{cftinterval2a}, and we state the second, for $8/\kappa\in2\mathbb{Z}^+$ only, in item \ref{cftinterval2b}.  In lemma \ref{identityintervallem} below, we show that these two definitions agree as we extend the second beyond $8/\kappa\in2\mathbb{Z}^+$ to all $\kappa\in(0,8)$ that are not exceptional speeds (\ref{exceptional}).  All of definition \ref{cftintervaldefn} except item \ref{cftinterval2b} restates definition \red{13} in \cite{florkleb}.
\begin{defn}\label{cftintervaldefn}With $\kappa\in(0,8)$, we select an $F\in\mathcal{S}_N$ and an $i\in\{1,2,\ldots,2N-1\}$.  Interpreting $\pi_{i+1}(\Omega_0)$ to be the part of the boundary of $\Omega_0$ (\ref{components}) whose points have only their $i$th and $(i+1)$th coordinates equal, we let
\be\label{secondH}H:\Omega_0\cup\pi_{i+1}(\Omega_0)\rightarrow\mathbb{R},\quad H(\boldsymbol{x}):=(x_{i+1}-x_i)^{6/\kappa-1}F(\boldsymbol{x})\quad\text{for $\boldsymbol{x}\in\Omega_0$},\ee
and with the formula for $H$ on $\Omega_0$ continuously extended to $\pi_{i+1}(\Omega_0)$.  We also let
\be\label{secondF0} F_0:\pi_{i+1}(\Omega_0)\rightarrow\mathbb{R},\quad (F_0\circ\pi_{i+1})(\boldsymbol{x}):=\lim_{x_{i+1}\rightarrow x_i}(x_{i+1}-x_i)^{6/\kappa-1}F(\boldsymbol{x}).\ee
\begin{enumerate}
\item\label{cftinterval1} We define $(x_i,x_{i+1})$ to be a \emph{two-leg interval of $F$} if the limit $F_0$ (\ref{secondF0}) vanishes. 
\item\label{cftinterval2} We define $(x_i,x_{i+1})$ to be an \emph{identity interval of $F$} if the limit $F_0$ (\ref{secondF0}) does not vanish, and
\begin{enumerate}
\item\label{cftinterval2a} if $8/\kappa\not\in\mathbb{Z}^+$ and $H$ (\ref{secondH}) is analytic at every point in $\pi_{i+1}(\Omega_0)$.
\item\label{cftinterval2b} if $8/\kappa\in2\mathbb{Z}^+$ and in the decomposition $F=a_1\mathcal{F}_1+a_2\mathcal{F}_2+\dotsm+a_{C_N}\mathcal{F}_{C_N}$ of $F$ over the basis $\mathcal{B}_N$, the half-plane diagram for $\mathcal{F}_\vartheta$ has an arc with endpoints at $x_i$ and $x_{i+1}$ for all $\vartheta\in\{1,2,\ldots,C_N\}$ with $a_\vartheta\neq0$.
\end{enumerate}
\end{enumerate}
In definition \red{13} of \cite{florkleb}, we naturally extend these terms to the interval $(x_{2N},x_1)$ of $F$ containing infinity.  We describe precisely how to do this in that article, and because this description is somewhat long, we do not restate it here.
\end{defn}
\noindent
(In section \ref{8kappaodd} of appendix \ref{appendix}, we explain the absence of a definition for the term ``identity interval" if $8/\kappa\in2\mathbb{Z}^++1$.)

Consistent with the discussion that began this section, this corollary, following immediately from theorem \ref{frobseriescor}, states that these pure CFT interval types are closely related to the forms of Frobenius series expansions of solutions in $\mathcal{S}_N$.
\begin{cor}\label{frobintervalcor}
Suppose that $\kappa\in(0,8)$, $F\in\mathcal{S}_N$, and $i\in\{1,2,\ldots,2N-1\}$.  Then 
\begin{enumerate}
\item\label{frobintervalcor1} for $8/\kappa$ not odd (resp.\ odd), $(x_i,x_{i+1})$ is a two-leg interval of $F$ if and only if $A_m=0$ (resp.\ $A_m=C_m=0$) for all $m\in\mathbb{Z}^+\cup\{0\}$ in the Frobenius series expansion (\ref{nolog}, \ref{stillnolog}) (resp.\ (\ref{log})) of $F(\boldsymbol{x})$ in powers of $x_{i+1}-x_i$.
\item\label{frobintervalcor2} for $8/\kappa\not\in\mathbb{Z}^+$, $(x_i,x_{i+1})$ is an identity interval of $F$ if and only if $A_0\neq0$ and $B_m=0$ for all $m\in\mathbb{Z}^+\cup\{0\}$ in the Frobenius series expansion (\ref{nolog}) of $F(\boldsymbol{x})$ in powers of $x_{i+1}-x_i$.
\end{enumerate}
\end{cor}

Although restricted to $8/\kappa\in2\mathbb{Z}^+$, the statement of item \ref{cftinterval2b} in definition \ref{cftintervaldefn} is clearly extendible to other $\kappa$ for which $\mathcal{B}_N$ is a basis of $\mathcal{S}_N$ (item \ref{firstitem} of theorem \ref{maintheorem}).  On making this extension, the two definitions (items \ref{cftinterval2a} and \ref{cftinterval2b} in definition \ref{cftintervaldefn})  have overlapping applicability, so we need to verify that they  agree, which is the purpose of the following lemma. (We note that the characterization of an identity interval via the basis $\mathcal{B}_N$ in this lemma is similar to the  description of a contractible interval via the basis $\mathscr{B}_N$ in definition \ref{sleintervaldefn}.)
\begin{lem}\label{identityintervallem}
Suppose that $\kappa\in(0,8)$ is not an exceptional speed (\ref{exceptional}) with $q\leq N+1$ and $F\in\mathcal{S}_N\setminus\{0\}$.  Then $(x_i,x_{i+1})$ is an identity interval of $F$ if and only if in the decomposition 
\be\label{FdecompBN}F=a_1\mathcal{F}_1+a_2\mathcal{F}_2+\dotsm+a_{C_N}\mathcal{F}_{C_N}\ee
of $F$ over the basis $\mathcal{B}_N$, the half-plane diagram for $\mathcal{F}_\vartheta$ has an arc with endpoints at $x_i$ and $x_{i+1}$ whenever $a_\vartheta\neq0$.
\end{lem}
\begin{proof} If $8/\kappa\in2\mathbb{Z}^+$, then the lemma simply restates item \ref{cftinterval2b} in definition \ref{cftintervaldefn}, and if $8/\kappa\in2\mathbb{Z}^++1$, then $\kappa$ is an exceptional speed (\ref{exceptional}) with $q=2\leq N+1$, which is excluded by this lemma.  Thus, we assume $8/\kappa\not\in\mathbb{Z}^+$ here, so an identity interval of $F\in\mathcal{S}_N\setminus\{0\}$ is specified by item \ref{cftinterval2a} in definition \ref{cftintervaldefn}.  We also assume that $i\in\{1,2,\ldots,2N-1\}$ without loss of generality, and we index the arc connectivities according to item \ref{indexorder} above definition \ref{sleintervaldefn}.

To prove the ``if" statement of the lemma, we assume that $a_\vartheta=0$ for all $\vartheta>C_{N-1}$ in (\ref{FdecompBN}) and then show that $(x_i,x_{i+1})$ is an identity interval of $F$.  To this end, we choose a formula $\mathcal{F}_{c,\vartheta}$ (\ref{Fexplicit}) with $c\not\in\{i,i+1\}$ for each $\mathcal{F}_\vartheta\in\mathcal{B}_N$ with $\vartheta\leq C_{N-1}$.  As such, a Pochhammer contour shares its endpoints with $(x_i,x_{i+1})$ in each chosen formula, and by the discussion beneath (\red{44}) in \cite{florkleb3}, it follows that $(x_i,x_{i+1})$ is an identity interval of every $\mathcal{F}_{c,\vartheta}=\mathcal{F}_\vartheta$ with $\vartheta\leq C_{N-1}$.  Then from definition \ref{cftinterval2a} and (\ref{FdecompBN}), we have that $(x_i,x_{i+1})$ is an identity interval of $F$ too, proving the ``if" part.

To prove the ``only if" statement of the lemma, we assume that $(x_i,x_{i+1})$ is an identity interval of $F$ and then show that $a_\vartheta=0$ for all $\vartheta>C_{N-1}$ in (\ref{FdecompBN}).  For each $\vartheta\leq C_{N-1}$, we choose a formula $\mathcal{F}_{c,\vartheta}$ (\ref{Fexplicit}) for $\mathcal{F}_\vartheta$ as in the previous paragraph, and for each $\vartheta>C_{N-1}$, we choose a formula $\mathcal{F}_{c,\vartheta}$ (\ref{Fexplicit}) for $\mathcal{F}_\vartheta$ such that $x_i$ but not $x_{i+1}$ is an endpoint of a Pochhammer contour.  Section \ref{frobsect} calls the former (resp.\ latter) a ``case \ref{sc2}" (resp.\ ``case \ref{sc3}") term, using terminology from the proof of lemma \red{6} and appendix \red{A} of \cite{florkleb3}.  With $\mathcal{F}_{c,\vartheta}=\mathcal{F}_{c',\vartheta}$ for all $c,c'\in\{1,2,\ldots,2N\}$ and $\vartheta\in\{1,2,\ldots,C_N\}$, we may suppress the subscript $c$ from our notation.  As such, the decomposition (\ref{FdecompBN}) sorts into
\be\label{decompF} F=\underbrace{a_1\mathcal{F}_1+a_2\mathcal{F}_2+\dotsm+a_{C_{N-1}}\mathcal{F}_{C_{N-1}}}_{\text{case \ref{sc2} terms}}+\,\,\underbrace{a_{C_{N-1}+1}\mathcal{F}_{C_{N-1}+1}+a_{C_{N-1}+2}\mathcal{F}_{C_{N-1}+2}+\dotsm+a_{C_N}\mathcal{F}_{C_N}}_{\text{case \ref{sc3} terms}}.\ee
Now, by appropriately deforming the Pochhammer contour with an endpoint at $x_i$, we may write each case \ref{sc3} term in (\ref{decompF}) as a sum of a case \ref{sc2} term (multiplied by $n^{-1}$) with so-called ``case \red{1}" terms, that is, terms with neither $x_i$ nor $x_{i+1}$ an endpoint of any contour.  We describe how this is done in item \red{3} of the proof of lemma \red{6} in \cite{florkleb3}, and we explicitly perform the calculation in section \red{A 3} of \cite{florkleb3}.  (See also section \ref{frobsect} of this article.)  We may write the result as
\be\label{decompF1}\mathcal{F}_\vartheta=\underbrace{n^{-1}\mathcal{F}_{\chi(\vartheta)}}_{\text{case \ref{sc2} term}}\,\,+\,\,\underbrace{\text{terms with $(x_i,x_{i+1})$ a two-leg interval}}_{\text{case \red{1} terms}},\quad\vartheta>C_{N-1},\ee
where we define $\chi$ in item \ref{cutmap} above definition \ref{sleintervaldefn} and $n$ in (\ref{LkFk}).  (We note that $n(\kappa)\neq0$ because $8/\kappa\not\in2\mathbb{Z}^++1$.)  Equation  (\ref{decompF1}) states that $(x_i,x_{i+1})$ is a two-leg interval of each case \red{1} term, a fact verified in section \red{A 1} of \cite{florkleb3}.  After isolating the sum of case \red{1} terms in (\ref{decompF1}), we may rewrite (\ref{decompF}) as a linear combination of case \red{1} and case \red{2} terms thus:
\be\label{decompF2}F=\underbrace{\sideset{}{_{\vartheta=1}^{C_{N-1}}}\sum a_\vartheta\mathcal{F}_\vartheta+\sideset{}{_{\vartheta=C_{N-1}+1}^{C_N}}\sum n^{-1}a_\vartheta\mathcal{F}_{\chi(\vartheta)}}_{\text{case \ref{sc2} terms}}+\underbrace{\sideset{}{_{\vartheta=C_{N-1}+1}^{C_N}}\sum a_\vartheta(\mathcal{F}_\vartheta-n^{-1}\mathcal{F}_{\chi(\vartheta)})}_{\text{case \red{1} terms}}.\ee
Because $(x_i,x_{i+1})$ is an identity (resp.\ two-leg) interval of each case \ref{sc2} (resp.\ case \red{1}) term  in (\ref{decompF2}), corollary \ref{frobintervalcor} says that the sum of these terms equals the first (resp.\ second) sum in the Frobenius series expansion (\ref{nolog}) for $F$.  However, $(x_i,x_{i+1})$ is an identity interval of $F$ by assumption, so the third sum in (\ref{decompF2}) must vanish, giving
\be\label{decompF3}F=\sum_{\vartheta=1}^{C_{N-1}}\Bigg(a_\vartheta\,\,\,+\,\,\,\sum_{\mathclap{\substack{\varrho=C_{N-1}+1 \\ \chi(\varrho)=\vartheta}}}^{C_N}\,\,\,n^{-1}a_\varrho\Bigg)\mathcal{F}_\vartheta,\ee
which is a second decomposition of $F$ over $\mathcal{B}_N$.  Because the decomposition of $F$ over the basis $\mathcal{B}_N$ is unique, the two decompositions (\ref{decompF}) and (\ref{decompF3}) must agree.  And because $n(\kappa)^{-1}\neq0$ for any real $\kappa$ (\ref{LkFk}), we therefore have $a_\vartheta=0$ for all $\vartheta>C_{N-1}$, proving the ``only if" part.
\end{proof} 

\subsection{Pure multiple-SLE$_\kappa$ versus pure CFT intervals}

Finally, we investigate the relationships between the types of pure multiple-SLE$_\kappa$ versus pure CFT intervals.  To begin, there is a very simple relationship between between propagating intervals and two-leg intervals of $F\in\mathcal{S}_N$.

\begin{lem}\label{proptwoleglem}
Suppose that $\kappa\in(0,8)$ and $F\in\mathcal{S}_N$.  Then $(x_i,x_{i+1})$ is a propagating interval of $F$ if and only if it is a two-leg interval of $F$.  (Here, we again identify the index $i+1=2N+1$ with one.)
\end{lem}
\begin{proof} Without loss of generality, we assume that $i\in\{1,2,\ldots,2N-1\}$.  If $i=2N$, then the proof is identical to the one shown here, except that we use the version $\underline{\ell}_1$ of the limit (\ref{lim}) that sends $-x_1=x_{2N}=R\rightarrow\infty$ \cite{florkleb,florkleb3}.

To begin, if $(x_i,x_{i+1})$ is a propagating interval of $F$, then definition \ref{sleintervaldefn} and item \ref{2a} of theorem \ref{xingasymplem} together imply that the limit (\ref{secondF0}) vanishes.  As such, item \ref{cftinterval1} of definition \ref{cftintervaldefn} says that $(x_i,x_{i+1})$ is a two-leg interval of $F$.

Next, if $(x_i,x_{i+1})$ is a two-leg interval of $F$, then theorem \ref{xingasymplem} shows that collapsing the interval $(x_i,x_{i+1})$ in the decomposition (\ref{FdecomposePi}) gives (as usual, we have indexed the arc connectivities as per item \ref{indexorder} above definition \ref{sleintervaldefn})
\be\label{col0}0\,\,=\,\,\bar{\ell}_1F\,\,=\sum_{\vartheta=1}^{C_{N-1}}a_\vartheta\bar{\ell}_1\Pi_\vartheta\,\,\,+\,\,\,\sum_{\mathclap{\vartheta=C_{N-1}+1}}^{C_N}\quad a_\vartheta\bar{\ell}_1\Pi_\vartheta\,\,\,=\sum_{\vartheta=1}^{C_{N-1}}a_\vartheta\Xi_\vartheta,\ee
where $\Xi_\vartheta=\bar{\ell}_1\Pi_\vartheta$ is a connectivity weight in $\mathcal{S}_{N-1}$.  Because $\mathscr{B}_{N-1}=\{\Xi_1,\Xi_2,\ldots,\Xi_{C_{N-1}}\}$ is a basis for $\mathcal{S}_{N-1}$, (\ref{col0}) implies that $a_\vartheta=0$ for all $\vartheta\leq C_{N-1}$ in (\ref{FdecomposePi}).  Thus by definition \ref{sleintervaldefn}, $(x_i,x_{i+1})$ is a propagating interval of $F$.
\end{proof}
\noindent
This lemma shows why the ``two-leg interval" is aptly named.  Because a propagating pair of boundary arcs anchors to a two-leg interval's endpoints, we find two arcs, or ``legs," anchored to one point after the endpoints come together.

The relationship between a contractible interval and an identity interval is not as simple. Intuitively, this is because, as mentioned, an identity interval does not condition the connectivity of its attached boundary arcs (figure \ref{TwoLegFuse}).  To understand  in more detail, we compare the probability that a multiple-SLE$_\kappa$ process generates a contractible boundary arc sharing its endpoints with an identity interval to the probability of a pair of propagating boundary arcs doing the same.  Because these probabilities depend on the choice of SLE$_\kappa$ partition function, we make this comparison in only one case.  We promote the connectivity weight $\Xi_\varsigma\in\mathscr{B}_{N-1}$, a function of the points in $\{x_j\}_{j\neq i,i+1}$ for some $i\in\{1,2,\ldots, 2N-1 \}$, to a new SLE$_\kappa$ partition function $\Theta_\varsigma\in\mathcal{S}_N$ of all points in $\{x_j\}_{j=1}^{2N}$ and with two special properties: first, $(x_i,x_{i+1})$ is an identity interval of it, and second, $\Theta_\varsigma\rightarrow n\Xi_\varsigma$ as $x_{i+1}\rightarrow x_i$.  By decomposing $\Theta_\varsigma$ over $\mathscr{B}_N$, we then determine how the lone identity interval $(x_i,x_{i+1})$ decomposes into a combination of contractible and propagating intervals.

Now we construct $\Theta_\varsigma$.  After indexing the arc connectivities according to item \ref{indexorder} above definition \ref{sleintervaldefn}, we define the sets $\mathscr{C}_N=\{\Pi_1,\Pi_2,\ldots,\Pi_{C_{N-1}}\}\subset\mathscr{B}_N$ and $\mathscr{B}_{N-1}=\{\Xi_1,\Xi_2,\ldots,\Xi_{C_{N-1}}\}$ relative to the interval $(x_i,x_{i+1})$ as in the proof of theorem \ref{xingasymplem}, and we write $\mathcal{C}_N=\{\mathcal{F}_1,\mathcal{F}_2,\ldots,\mathcal{F}_{C_{N-1}}\}\subset\mathcal{B}_N$ and $\mathcal{B}_{N-1}=\{\mathcal{G}_1,\mathcal{G}_2,\ldots,\mathcal{G}_{C_{N-1}}\}$, where the function $\mathcal{G}_\vartheta$ is such that its half-plane diagram follows from removing the arc with endpoints at $x_i$ and $x_{i+1}$ from the half-plane diagram of $\mathcal{F}_\vartheta\in\mathcal{C}_N$.  (Thus, the functions in $\mathscr{B}_{N-1}$ and $\mathcal{B}_{N-1}$ depend only on the points in $\{x_j\}_{j\neq i,i+1}$.)  Said differently, we create the function $\mathcal{F}_\vartheta\in\mathcal{C}_N$ from its corresponding function $\mathcal{G}_\vartheta\in\mathcal{B}_{N-1}$ by inserting the missing points $x_i<x_{i+1}$ between $x_{i-1}$ and $x_{i+2}$ (resp.\ to the right of $x_{2N-2}$, resp.\ to the left of $x_3$) if $i\not\in\{1,2N-1\}$ (resp.\ if $i=2N-1$, resp.\ if $i=1$) and entwining them together with a new Pochhammer contour (figure \ref{PochhammerContour}).  
Now, for each connectivity weight $\Xi_\varsigma\in\mathscr{B}_{N-1}$, we invert (\ref{F=MPi}) to find the decomposition
\be\label{Xidecomp}\Xi_\varsigma(\kappa\,|\,\boldsymbol{x})=\sum_{\vartheta=1}^{\mathclap{C_{N-1}}} b_{\varsigma,\vartheta}(\kappa)\mathcal{G}_\vartheta(\kappa\,|\,\boldsymbol{x}),\ee
where the $b_{\varsigma,\vartheta}(\kappa)$ are the entries of the inverse $(M_{N-1}^{-1}\circ n)(\kappa)$ of the $C_{N-1}\times C_{N-1}$ meander matrix \cite{fgg, fgut, difranc, franc}.  Then, for each $\varsigma\in\{1,2,\ldots,C_{N-1}\}$, we define the new function $\Theta_\varsigma\in\mathcal{S}_N$ by
\be\label{Thetadefn}\Theta_\varsigma:(0,8)\times\Omega_0\rightarrow\mathbb{R},\quad\Theta_\varsigma(\kappa\,|\,\boldsymbol{x}):=\sum_{\vartheta=1}^{\mathclap{C_{N-1}}} b_{\varsigma,\vartheta}(\kappa)\mathcal{F}_\vartheta(\kappa\,|\,\boldsymbol{x}).\ee
(If $\kappa$ is an exceptional speed (\ref{exceptional}) with $q\leq N$, then the meander matrix $(M_{N-1}\circ n)(\kappa)$ is not invertible, so the coefficients $b_{\varsigma,\vartheta}(\kappa)$ of (\ref{Thetadefn}) do not exist.  Below, we derive the alternative formula (\ref{finallincmb}) for $\Theta_\varsigma$.  This formula shows that $\Theta_\varsigma(\varkappa)$ has a limit as $\varkappa\rightarrow\kappa$, and we define $\Theta_\varsigma(\kappa)$ to be that limit.)  We think of $\Theta_\varsigma$ as ``almost" a connectivity weight because, according to case \red{2} in the proof of lemma \red{6} in \cite{florkleb3}, if we collapse the inserted interval $(x_i,x_{i+1})$, then we find $n$ (\ref{LkFk}) times the original connectivity weight $\Xi_\varsigma\in\mathscr{B}_{N-1}$ (\ref{Xidecomp}).

\begin{figure}[b]
\centering
\includegraphics[scale=0.38]{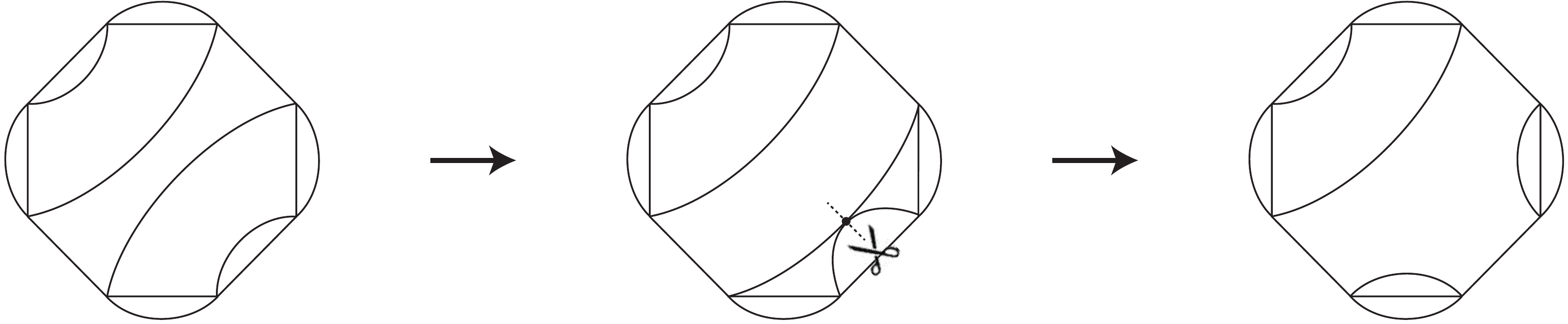}
\caption{The index map $\chi$ of figure \ref{CutMap} applied to the diagram for $[\mathscr{L}_\varrho]\mathcal{F}_\vartheta$ with $\varrho>C_{N-1}$ and $\vartheta\leq C_{N-1}$.  The octagons' bottom sides correspond to $(x_i,x_{i+1})$; the left octagon is the diagram for $[\mathscr{L}_\varrho]\mathcal{F}_\vartheta$; the right octagon is the diagram for $[\mathscr{L}_{\chi(\varrho)}]\mathcal{F}_\vartheta$.}
\label{Cut}
\end{figure}

Anticipating from lemma \ref{identityintervallem} that $(x_i,x_{i+1})$ is an identity interval of $\Theta_\varsigma$, we decompose it over $\mathscr{B}_N$ (\ref{finallincmb}) to determine what boundary arc connectivities the multiple-SLE$_\kappa$ process with $\Theta_\varsigma$ for its partition function generates.  (Later, we use this decomposition to help confirm that $(x_i,x_{i+1})$ is indeed an identity interval of $\Theta_\varsigma(\kappa)$ for all $\kappa\in(0,8)$, except for $8/\kappa$ odd.)  Because $\Theta_\varsigma\rightarrow n\Xi_\varsigma$ as $x_{i+1}\rightarrow x_i$, the only connectivity weight in $\mathscr{C}_N$ to appear in this decomposition is $\Pi_\varsigma$, with coefficient $n$.  (We recall that by theorem \ref{xingasymplem}, the connectivity weights in $\mathscr{C}_N$ do not vanish in this limit.)  Similarly, not every connectivity weight in $\mathscr{B}_N\setminus\mathscr{C}_N$ necessarily appears in this decomposition.  To determine which do appear, we act on both sides of (\ref{Thetadefn}) with $[\mathscr{L}_\varrho]$ for $\varrho>C_{N-1}$.  Thus, we must find $[\mathscr{L}_\varrho]\mathcal{F}_\vartheta$ for $\vartheta\leq C_{N-1}$ and $\varrho>C_{N-1}$.  Now, while $x_i$ and $x_{i+1}$ are not endpoints of the same interior arc in the diagram for each $[\mathscr{L}_\varrho]$ with $\varrho>C_{N-1}$, they are endpoints of the same exterior arc in the diagram for each $\mathcal{F}_\vartheta$ with $\vartheta\leq C_{N-1}$.  The two distinct interior arcs join with this exterior arc to form part of the same loop in the polygon diagram for the product $[\mathscr{L}_\varrho]\mathcal{F}_{\vartheta}$ (figure \ref{innerproduct} with the polygon deleted).  Next, we pinch these two interior arcs together at a point in the polygon and cut them there to form a new loop passing exclusively through the $i$th and $(i+1)$th vertex and separate from what remains of the old loop.  While this does not change the diagram of $\mathcal{F}_\vartheta$, this does change the diagram of $[\mathscr{L}_\varrho]$ to that of some $[\mathscr{L}_\varpi]\in\mathscr{C}_N^*$ (figure \ref{Cut}).  In fact, we have $\chi(\varrho)=\varpi$, where $\varrho>C_{N-1}$, $\varpi\leq C_{N-1}$, and $\chi$ is the map defined in item \ref{cutmap} above definition \ref{sleintervaldefn}.  For all $\vartheta\leq C_{N-1}$, the diagram of $[\mathscr{L}_{\chi(\varrho)}]\mathcal{F}_\vartheta$ has one more loop than that of $[\mathscr{L}_\varrho]\mathcal{F}_\vartheta$, so
\be\label{chimap}\text{$[\mathscr{L}_{\chi(\varrho)}]\mathcal{F}_\vartheta=n[\mathscr{L}_\varrho]\mathcal{F}_\vartheta$ if $\varrho>C_{N-1}$ and $\vartheta\leq C_{N-1}$}\quad\Longrightarrow\quad\text{$[\mathscr{L}_{\chi(\varrho)}]\Theta_\varsigma=n[\mathscr{L}_\varrho]\Theta_\varsigma$ if $\varrho>C_{N-1}$}\ee  
thanks to (\ref{LkFk}, \ref{Thetadefn}).  We previously showed that $[\mathscr{L}_\varpi]\Theta_\varsigma=n\delta_{\varpi,\varsigma}$ if $\varpi\leq C_{N-1}$.  Combining these two equations with $\varpi=\chi(\varrho)\leq C_{N-1}$, we have $[\mathscr{L}_\varrho]\Theta_\varsigma=\delta_{\chi(\varrho),\varsigma}$ for all $\varrho>C_{N-1}$, and hence
\be\label{finallincmb}\Theta_\varsigma\,\,=\,\,n\Pi_\varsigma\,\,+\,\,\sum_{\mathclap{\substack{\varrho=C_{N-1}+1 \\ \chi(\varrho)=\varsigma}}}^{C_N}\,\,\Pi_\varrho.\ee

\begin{figure}[b]
\centering
\includegraphics[scale=0.25]{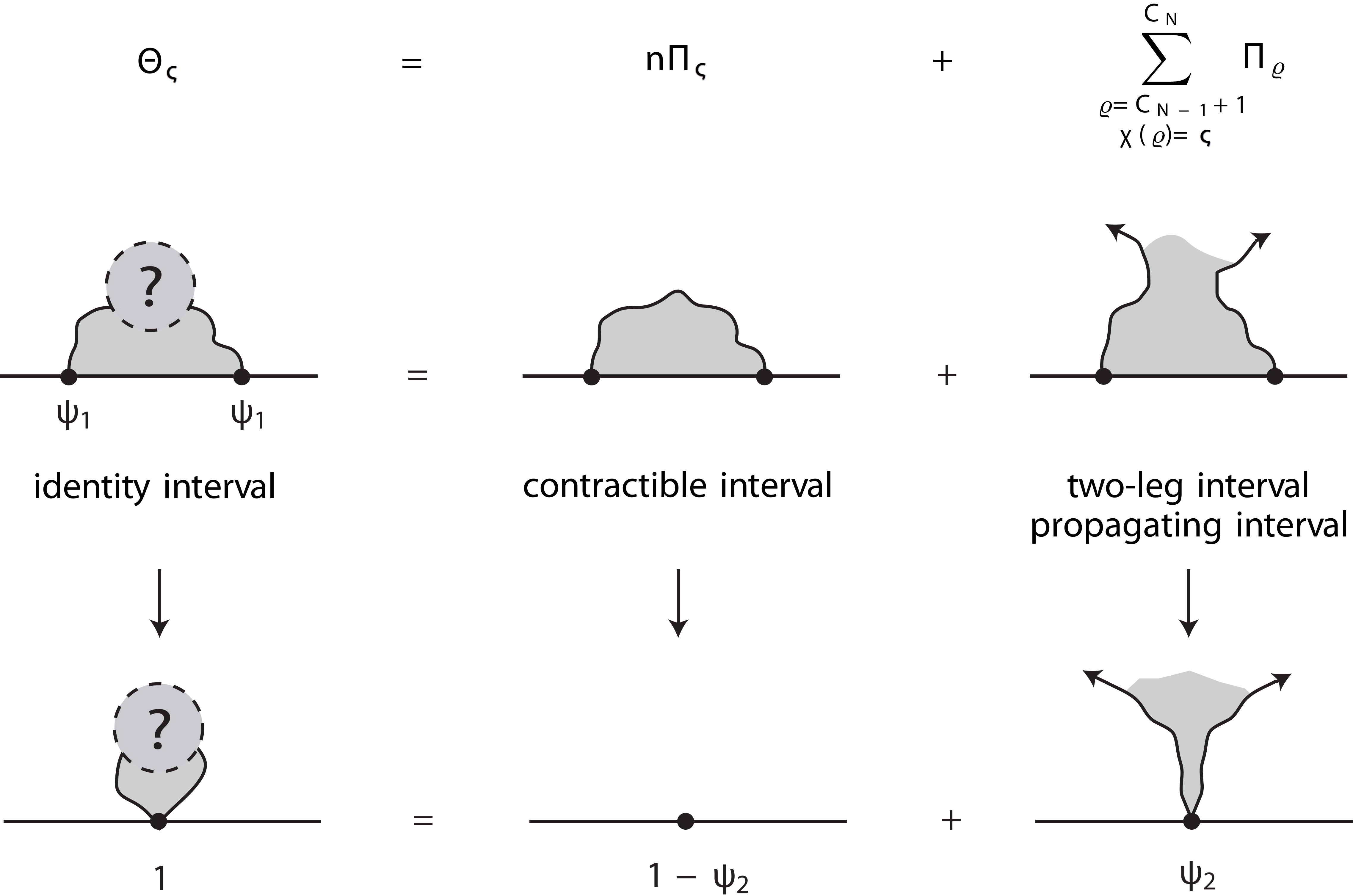}
\caption{Relation between the pure SLE$_\kappa$ (``contractible" and ``propagating") and pure CFT (``identity" and ``two-leg") interval types, the inter-connectivities of their boundary arcs, and the OPE content of their one-leg boundary operators.  (See table \ref{FrobTable}.)}
\label{TwoLegFuse}
\end{figure}

This second formula for $\Theta_\varsigma$ has important uses.  First, (\ref{finallincmb}) gives the continuous extension of (\ref{Thetadefn}) to the exceptional speeds with $q\leq N$.  Second, (\ref{finallincmb}) shows that $\Theta_\varsigma$ is an SLE$_\kappa$ partition function (definition \ref{partitiondefn}) if $n(\kappa)>0$ (\ref{LkFk}), assuming the positivity of the connectivity weights.  (See the discussion beneath (\ref{xing2}) in section \ref{xingprob}.)  Third, (\ref{finallincmb}) helps us to confirm that $(x_i,x_{i+1})$ is an identity interval of $\Theta_\varsigma(\kappa)$ for all $\kappa\in(0,8)$ with $8/\kappa\not\in2\mathbb{Z}^++1$.  Indeed, if $\kappa$ is not an exceptional speed (\ref{exceptional}) with $q\leq N+1$, then this fact follows from the decomposition (\ref{Thetadefn}) and lemma \ref{identityintervallem}.  However, if $\kappa$ is such a speed with $q>2$, then this fact follows in part from (\ref{finallincmb}), and the argument proceeds this way: If $(x_i,x_{i+1})$ is an identity interval of $\Theta_\varsigma(\varkappa)$ for all $\varkappa\in(\kappa-\epsilon,\kappa)\cup(\kappa,\kappa+\epsilon)$ with $\epsilon>0$ small, then according to item \ref{frobintervalcor2} of corollary \ref{frobintervalcor}, the second sum in the expansion (\ref{nolog}) for $\Theta_\varsigma(\varkappa)$ is zero in these two intervals.  But because $\Pi_\varsigma$ is a continuous function of $\kappa\in(0,8)$, $\Theta_\varsigma(\varkappa)$ is continuous on $(\kappa-\epsilon,\kappa+\epsilon)$, so this second sum necessarily vanishes at $\varkappa=\kappa$ too.  Finally, because the right side of (\ref{finallincmb}) is not zero, $\Theta_\varsigma(\kappa)$ is not zero, so the first sum in the expansion (\ref{nolog}) for $\Theta_\varsigma(\kappa)$, being the only one present, is not zero either.  Hence, we conclude from item \ref{frobintervalcor2} of corollary \ref{frobintervalcor} that $(x_i,x_{i+1})$ is an identity interval of $\Theta_\varsigma(\kappa)$ too.  (If, on the other hand, $\kappa$ is an exceptional speed (\ref{exceptional}) with $q=2$, then $n(\kappa)=0$, so the first term on the right side of (\ref{finallincmb}) vanishes.  According to definition \ref{sleintervaldefn} and lemma \ref{proptwoleglem}, $(x_i,x_{i+1})$ is then a two-leg interval of $\Theta_\varsigma$.)  Finally, (\ref{finallincmb}) shows that an identity interval may be thought of as a superposition of a contractible interval with a propagating interval.  The first term on the right side of (\ref{finallincmb}) contributes to the former because $\varsigma\leq C_{N-1}$, and the other terms contribute to the latter.  Figure \ref{TwoLegFuse} expresses this superposition in terms of the OPE of the associated one-leg boundary operators.

But perhaps the most interesting use of (\ref{finallincmb}) is for determining the possible connectivities of the boundary arcs attached to the identity interval $(x_i,x_{i+1})$ in the multiple-SLE$_\kappa$ process with $\Theta_\varsigma$ for its partition function.  (Because of definition \ref{partitiondefn} and the positivity of the connectivity weights assumed beneath (\ref{xing2}), we need $n(\kappa)>0$.  This restriction is satisfied in many intervals of $\kappa$ values, such as $\kappa\in(8/3,8)$.)  Using the conjectured crossing-probability formula (\ref{xing}, \ref{xing2}) with $F=\Theta_\varsigma$ from (\ref{finallincmb}) shows that these boundary arcs almost surely join the points $x_1$, $x_2,\ldots,x_{2N}$ in either the $\varsigma$th connectivity (with $\varsigma\leq C_{N-1}$) or in the $\varrho$th connectivity (with $\varrho>C_{N-1}$ and $\chi(\varrho)=\varsigma$).  (See item \ref{cutmap} above definition \ref{sleintervaldefn} for the definition of $\chi$.)  In the first case, a contractible boundary arc anchors to the points $x_i$ and $x_{i+1}$, and in the second, a pair of propagating boundary arcs anchors to these points.  The restriction $\chi(\varrho)=\varsigma$ determines the possible endpoints of the pair of propagating boundary arcs anchored to $x_i$ and $x_{i+1}$ in the latter case.  In general, these are not just any points among $x_1$, $x_2,\ldots,x_{2N}$ (figure \ref{ThetaFig}).

The decomposition (\ref{finallincmb}) also has a statistical mechanics interpretation (i.e., percolation, Potts models, random cluster model, etc.) that originates in the last paragraph of section \ref{puresle}.  As in the discussion surrounding (\ref{Piprefactor}--\ref{Pivanish}), we interpret each connectivity weight $\Pi_\vartheta\in\mathscr{B}_N$, and also $\Theta_\varsigma$ through (\ref{finallincmb}), as a statistical mechanics partition function for these systems, with the BC of the intervals $(x_1,x_2)$, $(x_3,x_4),\ldots,(x_{2N-1},x_{2N})$ in the ``fixed" state.  As such, a boundary cluster anchors to each fixed interval, and we indicate this cluster by coloring black the region bounded off by the boundary arc (or arcs) that surround it in the polygon diagram for $\Pi_\vartheta$.  Supposing that the identity interval $(x_i,x_{i+1})$ of $\Theta_\varsigma$ is fixed, we modify the diagram for the connectivity weight $\Pi_\varsigma$ into one for $\Theta_\varsigma$ by recoloring the lone black region anchored to this interval gray and extending it to touch every black region to which it has access.  (By definition, the gray region has ``access" to a black region if it can touch that region without touching another black region.)  Figure \ref{ThetaFig} shows an example.  Each connection between the gray region and an accessible black region represents the possibility that \emph{only} the boundary cluster of the latter joins with that anchored to $(x_i,x_{i+1})$, forming one contiguous boundary cluster.  Thanks to (\ref{finallincmb}), all of these possible cluster-crossing events and no others contribute to $\Theta_\varsigma$.  Hence, we interpret $\Theta_\varsigma$ as a (statistical mechanics) partition function summing exclusively over the $\varsigma$th cluster-crossing event and over every cluster-crossing event generated by joining the lone cluster of the $\varsigma$th event anchored to $(x_i,x_{i+1})$ with exactly one of the other boundary clusters accessible to it.

\begin{figure}[t]
\centering
\includegraphics[scale=0.38]{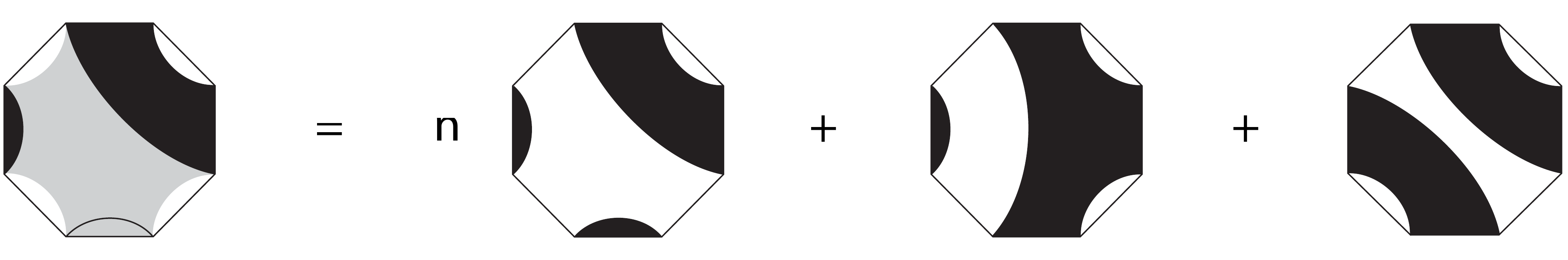}
\caption{The decomposition (\ref{finallincmb}).  The left diagram is $\Theta_2$; the first diagram on the right side is $\Pi_2\in\mathscr{C}_4$; the other two diagrams on the right side are $\Pi_6,\Pi_7\in\mathscr{B}_4\setminus\mathscr{C}_4$.  As such, we have $\chi(6)=\chi(7)=2$, where $\chi$ is the ``cutting map" of figures \ref{CutMap} and \ref{Cut}.}
\label{ThetaFig}
\end{figure}

\section{Exceptional speeds and CFT minimal models}\label{minmodelsect}

In this section, we study the correspondence between the exceptional speeds (\ref{exceptional}) and the CFT minimal models \cite{bpz, fms, henkel}.  Our main purpose is to propose conjecture \ref{minmodelconj}, which suggests a likely explanation for this correspondence.

To begin, we note a relationship between the rational SLE$_\kappa$ speeds and the central charge of the $\mathcal{M}(p,p')$ CFT minimal models.  We index the minimal model central charge \cite{bpz, fms, henkel}, and for convenience, the rational speeds as
\bea\label{mincentral}c_{p,p'}&:=&1-\frac{6(p-p')^2}{pp'},\quad 1<p'<p,\quad \text{$p,p'$ coprime,}\\
\label{rationalspeed}\kappa_{q,q'}&:=&4q/q',\quad q,q'\in\mathbb{Z}^+,\quad \text{$q,q'$ coprime}\eea
respectively.  Now, we insert the rational speed (\ref{rationalspeed}) into the formula $c(\kappa)=(6-\kappa)(3\kappa-8)/2\kappa$ that expresses the CFT central charge in terms of the SLE$_\kappa$ speed parameter \cite{bauber}.  This gives
\be\label{ck} 
c(\kappa)=\frac{(6-\kappa)(3\kappa-8)}{2\kappa}\quad\Longrightarrow\quad c(\kappa_{q,q'})=c_{p,p'},\quad\begin{array}{lll}p&:=&\text{max}\,\{q,q'\}, \\ p'&:=&\text{min}\,\,\{q,q'\}.\end{array}
\ee
Because $p'>1$ for a CFT minimal model, we find that each rational speed $\kappa_{q,q'}$ with $q>1$ and $q'>1$ (i.e., $\kappa\neq4r$ and $\kappa\neq4/r$ for any $r\in\mathbb{Z}^+$) corresponds to a CFT minimal model. In the other direction, we have
\be\label{ck2}
c(\kappa)=c_{p,p'}\quad\Longrightarrow\quad \kappa=\text{$\kappa_{p,p'}$ or $\kappa_{p',p}$},\quad1<p'<p,\quad \text{$p,p'$ coprime.}
\ee
That is, each minimal model central charge $c_{p,p'}$ corresponds to two rational speeds $\kappa_{q,q'}$ with $q>1$ and $q'>1$ (i.e., $\kappa\neq4r$ and $\kappa\neq4/r$ for any $r\in\mathbb{Z}^+$).  To summarize, we have the following two-to-one correspondence.
\begin{fact}\label{fact1}The set of all rational speeds $\kappa_{q,q'}$ (\ref{rationalspeed}) with $q,q'>1$ corresponds two-to-one with the set of all CFT minimal models $\mathcal{M}(p,p')$.  In this correspondence, we have $p=\text{max}\,\{q,q'\}$ and $p'=\text{min}\,\{q,q'\}$.
\end{fact}
Now, the set of all exceptional speeds (\ref{exceptional}) is identical to the subset of rational speeds described in fact \ref{fact1} above, except that the former includes speeds of the form $\kappa=4r$ with $r\in\mathbb{Z}^++1$, while the latter does not.  Because $4r\geq8$, we see that these two sets do agree over the range $\kappa\in(0,8)$ to which the results of this article and its predecessors \cite{florkleb,florkleb2,florkleb3} apply.  Observing also that $\kappa_{p,p'}<8$ if and only if $\kappa_{p',p}>2$ in (\ref{ck2}), we surmise the following from fact \ref{fact1}:
\begin{fact}\label{fact2} The set of all exceptional speeds $\kappa_{q,q'}\in(2,8)$ (\ref{exceptional}) corresponds two-to-one with a subset of CFT minimal models $\mathcal{M}(p,p')$.  In this correspondence, we have $p=\text{max}\,\{q,q'\}$ and $p'=\text{min}\,\{q,q'\}$.
\end{fact}
\noindent
On the other hand, if $\kappa_{p',p}<2$ in (\ref{ck2}), then we have $\kappa_{p,p'}>8$.  Because this second speed is outside the range $\kappa\in(0,8)$ to which our results apply, we state the correspondence of fact \ref{fact1} differently for $\kappa_{q,q'}\in(0,2]$:
\begin{fact}\label{fact3} The set of all exceptional speeds $\kappa_{q,q'}\in(0,2]$ (\ref{exceptional}) corresponds one-to-one with the subset of CFT minimal models $\mathcal{M}(p,p')$ excluded from fact \ref{fact2}.  In this correspondence, we have $p=q'$ and $p'=q$.
\end{fact}
\noindent
Although our results do not involve the range $\kappa\in[8,\infty)$, we note from fact \ref{fact1} that only in this range, not every exceptional speed corresponds to a CFT minimal model:
\begin{fact} The set of all exceptional speeds $\kappa_{q,q'}\in[8,\infty)$ (\ref{exceptional}) with $q'>1$ corresponds one-to-one with the subset of CFT minimal models $\mathcal{M}(p,p')$ excluded from fact \ref{fact2}.  In this correspondence, we have $p=q$ and $p'=q'$.
\end{fact}

Now we consider the reason for the connection between exceptional speeds (\ref{exceptional}) and CFT minimal models conveyed in facts \ref{fact2} and \ref{fact3} above.  The fact that $\kappa\in(0,8)$ is an exceptional speed $\kappa_{q,q'}$ (\ref{exceptional}) with $q\leq N+1$ if and only if $\mathcal{B}_N$ is linearly dependent (lemma \red{6} of \cite{florkleb3}) is the former's distinguishing property.  And what distinguishes minimal models from other CFTs is that they may be constructed from  a finite collection of primary operators \cite{bpz, fms, henkel}.  Thus, (\ref{ck}) suggests that these properties are related.

Before we propose a reason for this relation, we recall some relevant facts about the structure of CFT.  First, in a CFT with central charge $c\leq1$, the Verma module $V(c,h)$ is reducible if its conformal weight $h$ is among
\be\label{hrs}h_{r,s}(\kappa)=\frac{1-c(\kappa)}{96}\Bigg[\Bigg(r+s+(r-s)\sqrt{\frac{25-c(\kappa)}{1-c(\kappa)}}\,\Bigg)^2-4\Bigg]=\frac{1}{16\kappa}\begin{cases}(\kappa r-4s)^2-(\kappa-4)^2,&\kappa>4\\(\kappa s-4r)^2-(\kappa-4)^2,&\kappa\leq4\end{cases},\quad r,s\in\mathbb{Z}^+,\ee
(called a \emph{Kac weight}) because the module then harbors a level $rs$ null-state vector.  CFT translates this reducibility into a ``null-state" PDE that governs any correlation function containing the primary operator $\phi_{r,s}$, called a \emph{Kac operator}, with weight $h_{r,s}$ \cite{bpz, fms, henkel}.  For example, we identify the one-leg boundary operator $\psi_1(x_j)$ with the Kac operator $\phi_{1,2}(x_j)$ (resp.\,$\phi_{2,1}(x_j)$) if $\kappa>4$ (resp.\,$\kappa\leq4$) in appendix \red{A} of \cite{florkleb}.  Thus, the correlation function
\be\label{corrfunc} 
F(\boldsymbol{x})=\langle\psi_1(x_1)\psi_1(x_2)\dotsm\psi_1(x_{2N})\rangle,\quad\psi_1(x):=\begin{cases}\phi_{1,2}(x),&\kappa>4 \\ \phi_{2,1}(x),&\kappa\leq 4\end{cases}\ee
satisfies a distinct PDE associated with $\psi_1(x_j)$ for each $j\in\{1,2,\ldots,2N\}$ \cite{florkleb}, the null-state PDE centered on $x_j$ (\ref{nullstate}), in addition to the conformal Ward identities (\ref{wardid}) that any CFT correlation function of primary operators must satisfy.

Next, we recall two properties of $\mathcal{M}(p,p')$.  First, its operator content entirely comprises the conformal families of the Kac operators $\phi_{r,s}$ with $1\leq r<p'$ and $1\leq s< p$ \cite{bpz, fms, henkel}.  Second, the Verma module $V_{r,s}:=V(h_{r,s},c_{p,p'})$ has not just one null-state vector at level $l=rs$, but an infinite tower of null-state vectors at levels \cite{bpz, fms, henkel}
\be\label{nulllevels} l\in\{rs,\quad(p'-r)(p-s),\quad rs+(p'-r)(p+s),\quad rs+(p'+r)(p-s),\quad\ldots\}.\ee  
Therefore, a correlation function containing the $\phi_{r,s}$ Kac operator (such as (\ref{corrfunc}), where $(r,s)=(1,2)$ or $(2,1)$)  satisfies not just a single level-$rs$ null-state PDE, but an infinite system of null-state PDEs at different levels.  In most cases, the precise form of these PDEs is not explicitly known, although \cite{bauerfranc} gives a (non-trivial) recipe for finding them.

According to theorem \ref{maintheorem}, $\mathcal{B}_N$ is linearly dependent if and only if $\kappa\in(0,8)$ is an exceptional speed $\kappa_{q,q'}$ (\ref{exceptional}) with $q\leq N+1$.  Now, we posit that if $\kappa=\kappa_{q,q'}$, then the span of $\mathcal{B}_N$ equals the span of all $2N$-point correlation functions of one-leg boundary operators (\ref{corrfunc}).  This means that every element of $\mathcal{B}_N$ must satisfy not just the original system of null-state PDEs (\ref{nullstate}), but the entire infinite system described in the previous paragraph.  This requirement may explain the linear dependence of $\mathcal{B}_N$.  In fact, we propose something a little stronger.
\begin{conj}\label{minmodelconj}Let $\mathcal{R}_N\subset\mathcal{S}_N$ be the solution space for the infinite collection of null-state PDEs governing the correlation function (\ref{corrfunc}) with $\kappa\in(0,8)$  an exceptional speed $\kappa_{q,q'}$ (\ref{exceptional}) with $q\leq N+1$. Then $\text{span}\,\mathcal{B}_N=\mathcal{R}_N$. 
\end{conj}
\noindent
This conjecture is difficult to state more precisely because the infinite system of PDEs is very difficult to determine entirely.  However, if it holds, then (see the discussion surrounding (\ref{LkFk}--\ref{exceptional}) in the introduction \ref{intro} and corollary \red{7} of \cite{florkleb3})
\be\dim\mathcal{R}_N=\text{rank}\,\mathcal{B}_N=\text{rank}\,M_N\circ n\ee
where, again, $M_N$ is the $C_N\times C_N$ meander matrix \cite{fgg,fgut,difranc,franc}.  According to \cite{franc}, the rank of $M_N$ equals the multiplicity $d_N(q,q'')$ of the zero $n_{q,q''}$ (\ref{thezeros}) of the meander determinant for any $q''\in\{1,2,\ldots,N\}$.  Explicitly, this is \cite{fgg}
\be\label{rankzero}\text{rank}\,M_N=d_N(q,q'')=\frac{1}{2q}\sum_{p=1}^{q-1}\left(2\sin\frac{\pi p}{q}\right)^2\left(2\cos\frac{\pi p}{q}\right)^{2N}.\ee
We note that this formula in fact does not depend on $q''$, and only the first index $q$ of the exceptional speed $\kappa_{q,q'}$ enters.  Ref.\ \cite{fgg} shows that (\ref{rankzero}) is a positive integer.

In what follows, we explore conjecture \ref{minmodelconj} further for the zeros (\ref{thezeros}) $n_{2,1}$, $n_{3,1}$, $n_{3,2}$, and $n_{N+1,q''}$ with $q''\in\{1,2\ldots,N\}$ of the meander determinant.  We focus our attention mainly on solutions of (\ref{thezeros}) with $\kappa\geq2$. In this range, the scaling limit of loops in the loop-gas representation of the O$(n)$-model are conjectured to (locally) have the law of SLE$_\kappa$, with the loop fugacity $n$ and $\kappa$ related through (\ref{LkFk}) \cite{gruz, rgbw, smir4,smir}.  Because many critical lattice models may be mapped onto this loop-gas representation, their SLE$_\kappa$ descriptions therefore also have $\kappa\geq2$.

Now if $n_{q,q''}\in[0,2)$, then exactly two speeds in the range $\kappa\in[2,8]$ satisfy (\ref{thezeros}).  They are
\be\label{minimalmodelkappa}\kappa_{q,q''}=4q/q''\in(4,8],\quad\kappa_{q,2q-q''}=4q/(2q-q'')\in[8/3,4),\ee
and they respectively belong to the dense ($\kappa>4$) and dilute ($\kappa<4$) phase of both SLE$_\kappa$ and the O$(n)$ model.  One well-studied example is $n_{2,1}=0$, corresponding to $\kappa_{2,1}=8$ for the uniform spanning tree \cite{lsw} and $\kappa_{2,3}=8/3$ for the self-avoiding walk \cite{lsw2}.  Another is $n_{3,2}=1$, corresponding to $\kappa_{3,2}=6$ (percolation cluster boundaries) \cite{smir2, lsw3} and $\kappa_{3,4}=3$ (Ising spin cluster boundaries) \cite{smir3}.  On the other hand, if $n_{q,q''}\in[-2,0)$, then exactly one exceptional speed in the range $[2,8]$ satisfies (\ref{thezeros}).  It is $\kappa_{q,2q-q''}\in[2,8/3)$, and it is in the dilute phase of SLE$_\kappa$ and the O$(n)$ model.

\subsection{The case $n_{2,1}=0$}

From (\ref{thezeros}), we see that if $\kappa\in(0,8)$ and $n(\kappa)=n_{2,1}=0$, then $\kappa=8/r$ for some odd $r>1$.  This case has some distinctive features that we explore.  Here, $\mathcal{B}_N$ exhibits $d_N(2,1)=C_N$ \cite{fgg} distinct linear dependencies. Because its cardinality is also $C_N$, each of its elements equals zero.  In \cite{florkleb3}, the proof of theorem \red{8} shows that the set
\be\label{BN'}\mathcal{B}_N^{\scaleobj{0.75}{\bullet}}=\left\{\mathcal{F}_1^{\scaleobj{0.75}{\bullet}}:=\lim_{\varkappa\rightarrow\kappa}n(\varkappa)^{-1}\mathcal{F}_1(\varkappa),\quad\mathcal{F}_2^{\scaleobj{0.75}{\bullet}}:=\lim_{\varkappa\rightarrow\kappa}n(\varkappa)^{-1}\mathcal{F}_2(\varkappa),\quad\ldots,\quad\mathcal{F}_{C_N}^{\scaleobj{0.75}{\bullet}}:=\lim_{\varkappa\rightarrow\kappa}n(\varkappa)^{-1}\mathcal{F}_{C_N}(\varkappa)\right\}\ee
(defined slightly differently from what is shown in that proof for simplicity) has full rank and is therefore a basis for $\mathcal{S}_N$.  We note that the change from $\mathcal{B}_N$ to $\mathcal{B}_N^{{\scaleobj{0.75}{\bullet}}}$ amounts to dropping a single vanishing factor of $n(\kappa)$ from the formula (\ref{Fexplicit}) of every element of $\mathcal{B}_N$, leaving a function that is not zero.

The correspondence between this case and its associated minimal models is trivial.  Indeed, the solutions $\kappa=8/r$ with $r>1$ odd of the equation $n(\kappa)=n_{2,1}=0$ are also the exceptional speeds $\kappa_{2,4m\pm1}$ with $m\in\mathbb{Z}^+$.  These correspond with the $\mathcal{M}(p,2)$ minimal models with $p=4m\pm1>2$.  (We note that $\{4m\pm1\,|\,m\in\mathbb{Z}^+\}$ is the set of all odd positive integers greater than one.)  Because $\kappa_{2,4m\pm1}\leq4$, the one-leg boundary operator $\psi_1$ is now the Kac operator $\phi_{2,1}$, which does not appear in this model.  Thus, the correlation functions (\ref{corrfunc}) do not exist at all, and the set $\mathcal{B}_N$ that would contain them instead contains only zero.

\subsection{The cases $n_{3,2}=1$ and $n_{3,1}=-1$.}\label{n32sect}

From (\ref{thezeros}), we see that if $\kappa\in(0,8)$ and $n(\kappa)=n_{3,2}=1$, then $\kappa=\kappa_{3,2}=6$ or $\kappa_{3,6m\pm2}<4$ for some $m\in\mathbb{Z}^+$, corresponding with the $\mathcal{M}(3,2)$ and $\mathcal{M}(6m\pm2,3)$ minimal models respectively.  Among these, the $\mathcal{M}(3,2)$ (resp.\  $\mathcal{M}(4,3)$) model, after being appropriately extended \cite{matrid}, may be used to calculate correlation functions of critical percolation clusters (resp.\ Ising model spin clusters).  In all of these cases, $\mathcal{B}_N$ exhibits $d_N(3,2)=C_N-1$ \cite{fgg} distinct linear dependencies, and because its cardinality is $C_N$, we infer that all of its elements are multiples of each other.  Furthermore, the image of each element under the map $v$ of item \ref{fifthitem} in theorem \ref{maintheorem} is a vector with all components one, thanks to (\ref{LkFk}).  Because $v$ is a bijection, we infer that all elements of $\mathcal{B}_N$ are the same nonzero function $\mathcal{F}_1$.

Using these facts and conjecture \ref{minmodelconj}, we study the $\mathcal{M}(3,2)$ minimal model.  Because $\kappa_{3,2}=6>4$, the one-leg boundary operator $\psi_1$ (\ref{corrfunc}) is now the Kac operator $\phi_{1,2}$, and with $h_{1,2}(6)=h_{1,1}(6)$ (\ref{hrs}), we identify it with $\phi_{1,1}$ in this model.  This identification implies that the correlation functions (\ref{corrfunc}) satisfy the null-state PDEs associated with $\phi_{1,1}$, which are
\be\label{subsystem}\partial_j F(\boldsymbol{x})=0,\quad j\in\{1,2,\ldots,2N\},\ee
in addition to the original system (\ref{nullstate}, \ref{wardid}).  They also satisfy the infinite collection of other null-state PDEs described above \cite{fms} (all of which are linear, homogeneous, and lack a constant term if $\kappa=\kappa_{3,2}$), but the sub-system (\ref{subsystem}) alone is enough to see that $\mathbb{R}$ is the solution space $\mathcal{R}_N$ for this infinite system.  Now to verify our conjecture that $\text{span}\,\mathcal{B}_N=\mathcal{R}_N$ for this case, we check that $\text{span}\,\mathcal{B}_N=\mathbb{R}$ too.  Because each entry of the meander matrix $(M_N\circ n)(\kappa_{3,2})$ is one, the map $v$ of item \ref{fourthitem} in theorem \ref{maintheorem} sends the single unique element $\mathcal{F}_1\in\mathcal{B}_N$ to the image $v(1)$ of the constant solution $1\in\mathcal{S}_N$.  (Indeed, it is easy to see that constants satisfy (\ref{nullstate}, \ref{wardid}) if $\kappa=\kappa_{3,2}$.)
Furthermore, because $v$ is a bijection, we must have $\mathcal{F}_1=1$, so $\text{span}\,\mathcal{B}_N=\mathbb{R}=\mathcal{R}_N$.  This is  the only case for which we have completely verified conjecture \ref{minmodelconj}.

Incidentally, we may use the result that $\mathcal{F}_1=1$ to indirectly evaluate the Coulomb gas integral $\mathcal{J}$ that appears in each element of $\mathcal{B}_N$ if $\kappa=\kappa_{3,2}=6$.  Equation (\ref{Fexplicit}) gives an explicit formula for $\mathcal{F}_1$ for any $c\in\{1,2,\ldots,2N\}$.  We find that for any collection $\{\Gamma_1,\Gamma_2,\ldots,\Gamma_N\}$ of simple nonintersecting contours in the upper half-plane, each with both of its endpoints distinct and among $x_1<x_2<\ldots<x_{2N}$, 
\begin{multline}\label{defint}\sideset{}{_{\Gamma_{N-1}}}\int\ldots\sideset{}{_{\Gamma_2}}\int\sideset{}{_{\Gamma_1}}\int\mathcal{N}\Bigg[\Bigg(\prod_{l\neq c}^{2N}\prod_{m=1}^{N-1}(x_l-u_m)^{-2/3}\Bigg)\Bigg(\prod_{p<q}^{N-1}(u_p-u_q)^{4/3}\Bigg)\Bigg]{\rm d}u_1\,{\rm d}u_2\dotsm {\rm d}u_{N-1}\\
=\frac{\Gamma(1/3)^{2N-2}}{\Gamma(2/3)^{N-1}}\prod_{\substack{i<j \\ i,j\neq c}}^{2N}(x_i-x_j)^{-1/3},\end{multline}
where $x_c$ is an endpoint of $\Gamma_N$ (this being the only contour along which we do not integrate in (\ref{defint})).  (We recall that the symbol $\mathcal{N}$ determines the branch of the logarithm for each individual power function in the integrand so the Coulomb gas integral (\ref{defint}) is real-valued.  See both item \red{3} of definition \red{4} and figure \red{6} of \cite{florkleb3}.)

Next, we study the $\mathcal{M}(4,3)$ minimal model.  Because $\kappa_{3,4}=3<4$, the one-leg boundary operator $\psi_1$ is now the Kac operator $\phi_{2,1}$ (\ref{corrfunc}), and with $h_{2,1}(3)=h_{1,3}(3)$ (\ref{hrs}), we identify it with $\phi_{1,3}$ in this model.  This identification implies that the correlation functions (\ref{corrfunc}) satisfy the null-state PDEs associated with $\phi_{1,3}$, which are
\begin{multline}\label{13nullstate}\Bigg[\frac{2}{\kappa}\partial_j^3+2\sum_{k\neq j}^{2N}\left(\frac{\partial_k}{x_k-x_j}-\frac{\kappa/2-1}{(x_k-x_j)^2}\right)\partial_j\\
-\left(\frac{\kappa}{2}-1\right)\sum_{k\neq j}^{2N}\left(\frac{\partial_k}{(x_k-x_j)^2}-\frac{\kappa-2}{(x_k-x_j)^3}\right)\Bigg]F(\boldsymbol{x})=0,\quad j\in\{1,2,\ldots,2N\},\end{multline}
in addition to those of the original system (\ref{nullstate}, \ref{wardid}).  They also satisfy the infinite collection of other null-state PDEs described above, but as we previously mentioned, these PDEs are not so easy to determine.  Because of this complexity, we have not completely verified our conjecture that $\text{span}\,\mathcal{B}_N=\mathcal{R}_N$ for this case.  Indeed, we would need to show that the unique element of $\mathcal{B}_N$ satisfies the sub-system (\ref{13nullstate}) and the infinite collection of other null-state PDEs first.  Then we would need to show that each element of a basis for the quotient space $\mathcal{S}_N/\mathcal{B}_N$ satisfies neither (\ref{13nullstate}) nor one of these extra PDEs next.  If $N=2$, then it is easy to show that the unique element $\mathcal{F}_1\in\mathcal{B}_2$ satisfies (\ref{13nullstate}) but no element of $\mathcal{S}_2$ outside the span of $\mathcal{B}_2$ satisfies (\ref{13nullstate}).  Thus, $\dim\mathcal{R}_2\leq1$.  Showing that $\mathcal{F}_1$ satisfies the entire infinite system would  give $\text{span}\,\mathcal{B}_2\subset\mathcal{R}_2$, but to do this explicitly is difficult.  However, if we could do it, then because $\text{rank}\,\mathcal{B}_2=1$, it would immediately follow that $\text{span}\,\mathcal{B}_2=\mathcal{R}_2$, confirming conjecture \ref{minmodelconj} for this case.

More generally, we see from (\ref{thezeros}) that if $\kappa\in(0,4]$ and $n(\kappa)=n_{3,q''}=\pm1$ with $q''\in\{1,2\}$, then $\kappa=\kappa_{3,6m\pm q''}$ for some $m\in\mathbb{Z}^+$, corresponding with the $\mathcal{M}(6m\pm q'',3)$ minimal model.  In all of these cases, $\mathcal{B}_N$ exhibits $d_N(3,q'')=C_N-1$ \cite{fgg} linear dependencies, so again, all of its elements are multiples of each other.  Because $\kappa_{3,6m\pm q''}<4$, the one-leg boundary operator $\psi_1$ is now the Kac operator $\phi_{2,1}$ (\ref{corrfunc}), and with $h_{2,1}(\kappa_{3,6m\pm q''})=h_{1,6m\pm q''-1}(\kappa_{3,6m\pm q''})$ (\ref{exceptional}, \ref{hrs}), we identify it with the first-column Kac operator $\phi_{1,6m\pm q''-1}$ in this model.  This identification implies that the correlation functions (\ref{corrfunc}) satisfy the null-state PDEs associated with $\phi_{1,6m\pm q''-1}$ in addition to those (\ref{nullstate}) of the original system (\ref{nullstate}, \ref{wardid}) that are associated with $\phi_{2,1}$.  Thanks to the Benoit-Saint-Aubin formula \cite{fms,benaub}, these PDEs are explicitly known.  Hence, we may repeat the analysis of the previous paragraph (corresponding to our present situation if $m=1$ and $q''=2$ in $\phi_{1,6m-q''-1}$) to again support conjecture \ref{minmodelconj}.

As we previously observed, all elements of $\mathcal{B}_N$ are equal if $n(\kappa)=n_{3,2}=1$.  However, in the case $n(\kappa)=n_{3,1}=-1$, this is not true, which implies an interesting identity concerning the number $l_{\varsigma,\vartheta}$ (\ref{LkFk}) of loops in the diagram for $[\mathscr{L}_\varsigma]\mathcal{F}_\vartheta$ (figure \ref{innerproduct}).  Indeed, with the rank of $\mathcal{B}_N$ still $C_N-d_N(3,1)=1$, it follows that all elements of $\mathcal{B}_N$ are multiples of each other.  However, unlike the case $n(\kappa)=n_{3,2}=1$, they are not necessarily equal.  Rather,
\be\label{Fvarth}\mathcal{F}_\vartheta=(-1)^{p_{\vartheta,\varrho}}\mathcal{F}_\varrho,\quad p_{\vartheta,\varrho}\in\{1,2\},\quad\vartheta,\varrho\in\{1,2,\ldots,C_N\},\quad n(\kappa)=n_{3,1}=-1.\ee
Indeed, this relation follows from the fact that $v$ is a bijection (item \ref{fourthitem} of theorem \ref{maintheorem}) and that either $v(\mathcal{F}_\vartheta)=v(\mathcal{F}_\varrho)$ or $v(\mathcal{F}_\vartheta)=-v(\mathcal{F}_\varrho)$ thanks to (\ref{LkFk}).  Now after acting on both sides of (\ref{Fvarth}) with $[\mathscr{L}_\varsigma]$ and using (\ref{LkFk}), we find
\be\label{loopdiff}p_{\vartheta,\varrho}=l_{\varsigma,\vartheta}-l_{\varsigma,\varrho}\mod2,\ee
where $l_{\varsigma,\vartheta}-l_{\varsigma,\varrho}$ is the difference of the number of loops appearing in the diagram for $[\mathscr{L}_\varsigma]\mathcal{F}_\vartheta$ versus $[\mathscr{L}_\varsigma]\mathcal{F}_\varrho$ (with the polygons deleted).  Interestingly, (\ref{loopdiff}) shows that this difference in modulo two is the same for all $\varsigma\in\{1,2,\ldots,C_N\}$.

\begin{figure}[b]
\centering
\includegraphics[scale=0.3]{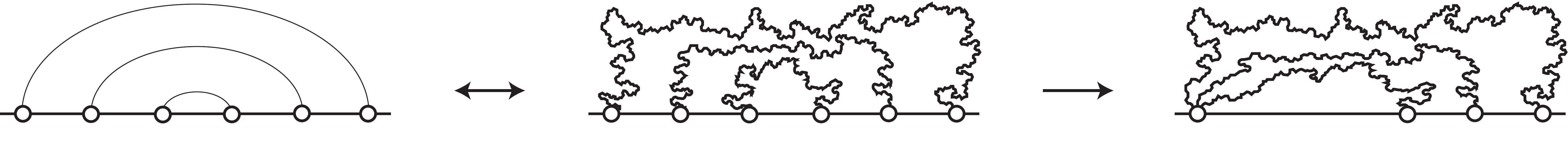}
\caption{The ``rainbow diagram" for $\Pi_5\in\mathscr{B}_3$ (left), the corresponding fifth boundary arc connectivity (middle), and the configuration found by sending $x_3,x_2\rightarrow x_1$ in (\ref{multilim}) to form the three-leg boundary operator $\psi_3(x_1)$ in (\ref{Nptop}) with $N=3$ (right).}
\label{Rainbow}
\end{figure}

\subsection{The case $n=n_{N+1,q''}$}

The number $n_{q,q''}$ (\ref{thezeros}) is a zero of the meander determinant $\det M_N\circ n$ only if $q\leq N+1$, and at $q=N+1$, its multiplicity is $d_N(N+1,q'')=1$ \cite{fgg}.  From (\ref{thezeros}), we see that if $\kappa\in(0,8)$ and $n(\kappa)=n_{N+1,q''}$, then $\kappa=\kappa_{N+1,q'}$ with $q'=q''$ or $q'=2m(N+1)\pm q''$ for some $m\in\mathbb{Z}^+$, corresponding with the $\mathcal{M}(N+1,q'')$ and $\mathcal{M}(2m(N+1)\pm q'',N+1)$ minimal models respectively.  In all of these cases, $\mathcal{B}_N$ exhibits $d_N(N+1,q'')=1$ distinct linear dependence \cite{fgg}.

Supposing that the half-plane diagram for the $C_N$th connectivity is the \emph{rainbow diagram} (figure \ref{Rainbow}), with its $j$th arc having endpoints at $x_j$ and $x_{2N-j+1}$ \cite{fgg}, we may generate a basis for $\mathcal{S}_N$ from the full-rank set $\{\mathcal{F}_1,\mathcal{F}_2,\ldots,\mathcal{F}_{C_N-1}\}\subset\mathcal{B}_N$  by adding one element not in the span of $\mathcal{B}_N$.  We choose this to be the connectivity weight $\Pi_{C_N}$.  Indeed, to prove that $\Pi_{C_N}\not\in\text{span}\,\mathcal{B}_N$, we assume otherwise and allow every $[\mathscr{L}_\varsigma]$ with $\varsigma\in\{1,2,\ldots,C_N-1\}$ to act on 
\be \Pi_{C_N}=a_1\mathcal{F}_1+a_2\mathcal{F}_2+\dotsm+a_{C_N-1}\mathcal{F}_{C_N-1},\quad a_1,a_2,\ldots,a_{C_N-1}\in\mathbb{R}.\ee
With $\boldsymbol{a}=(a_1,a_2,\ldots,a_{C_N-1})\in\mathbb{R}^{C_N-1}\setminus\{0\}$, we find the matrix equation $M\cdot\boldsymbol{a}=0$, where $M$ is the upper-left $(C_N-1)\times(C_N-1)$ sub-matrix of the meander matrix $M_N\circ n$.  According to \cite{franc}, the determinant of this matrix is $(n-n_{N+1,q''})^{-1}\det (M_N\circ n)$, which does not approach zero as $n(\kappa)\rightarrow n_{N+1,q''}$.  Hence $\boldsymbol{a}=0$, a contradiction.  From this, we conclude that $\Pi_{C_N}\not\in\text{span}\,\mathcal{B}_N$, so the following set is a basis for $\mathcal{S}_N$:
\be\label{basis}\{\mathcal{F}_1,\mathcal{F}_2,\ldots,\mathcal{F}_{C_N-1},\Pi_{C_N}\}.\ee

\begin{figure}[t]
\centering
\includegraphics[scale=0.3]{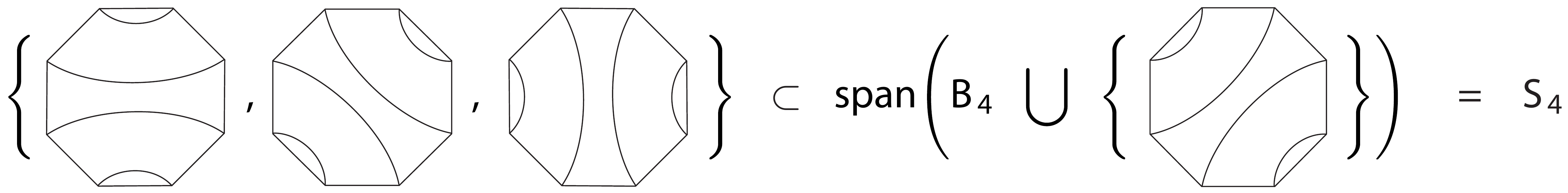}
\caption{If $\kappa=\kappa_{5,q'}$ for some $q'\in\{1,2,3,4\}$, then $\text{span}\,(\mathcal{B}_4\cup\{\Pi_{C_4}\})=\mathcal{S}_4$, and the other $4-1=3$ connectivity weights from which we may generate a four-leg boundary operator (polygon diagrams shown on the left side) are therefore in this span.}
\label{Contain}
\end{figure}

This new basis (\ref{basis}) has an interesting interpretation in terms of CFT minimal models.  According to definition \ref{sleintervaldefn} and lemma \ref{proptwoleglem}, if all of $(x_1,x_2)$, $(x_2,x_3),\ldots,(x_{N-2},x_{N-1})$, and $(x_{N-1},x_N)$ are simultaneously two-leg intervals of $F\in\mathcal{S}_N$, then no diagram of any connectivity weight with a nonzero coefficient in the decomposition (\ref{FdecompPi}) of $F$ over $\mathscr{B}_N$ may have an arc that shares both of its endpoints with one of these intervals.  Because only the rainbow diagram (figure \ref{Rainbow}) has this property, we infer that $F\propto\Pi_{C_N}$.  With this observation, it follows from \cite{florkleb2} that the limit
\be\label{multilim}\lim_{x_N\rightarrow x_{N-1}}\dotsm\lim_{x_3\rightarrow x_2}\lim_{x_2\rightarrow x_1}\prod_{j=1}^{N-1}(x_{j+1}-x_j)^{-\Delta^+(\theta_{N-j+1})}F(\boldsymbol{x}),\quad\Delta^+(\theta_s)=2s/\kappa\ee
(with $\theta_s$ and $\Delta^+$ defined in (\red{6}) and (\red{29}) of \cite{florkleb2} respectively) exists if and only if $F\propto\Pi_{C_N}$, and the discussion in section \red{I A} of \cite{florkleb2} suggests that we identify such a solution with the $(N+1)$-point correlation function (figure \ref{Rainbow})
\be\label{Nptop}\langle\psi_N(x_1)\psi_1(x_{N+1})\psi_1(x_{N+2})\dotsm\psi_1(x_{2N})\rangle
\ee
containing the $N$-leg boundary operator $\psi_N$.  In particular, the limit (\ref{multilim}) does not exist for any element of $\mathcal{B}_N$.  Hence, if the span of $\mathcal{B}_N$ does contain all $2N$-point correlation functions (\ref{corrfunc}) of one-leg boundary operators, then the $N$-leg boundary operator $\psi_N$ must be absent from the associated minimal model.  Indeed, this is true.  If $\kappa=\kappa_{N+1,q'}$ with $q'=q''$ so $\kappa>4$ and $\psi_1=\phi_{1,2}$ (resp.\ with $q'=2m(N+1)\pm q''$ so $\kappa<4$ and $\psi_1=\phi_{2,1}$), then $\psi_N=\phi_{1,N+1}$ (resp.\ $\phi_{N+1,1}$) is the first first-column (resp.\ first-row) Kac operator that appears outside the Kac table for the corresponding $\mathcal{M}(N+1,q'')$ (resp.\ $\mathcal{M}(2m(N+1)\pm q'',N+1)$) minimal model.  Because of the absence of $\psi_N$ from the model, this minimal model does not include the $\psi_N$ conformal family in either of the fusion products 
\be\label{fusions}
\psi_1(x_1)\times\psi_1(x_2)\times\dotsm\psi_1(x_N),\quad\psi_1(x_{N+1})\times\psi_1(x_{N+2})\times\dotsm\psi_1(x_{2N}).
\ee
(Its absence from the second product follows from symmetry.)  Inserting $\Pi_{C_N}$ into the set (\ref{basis}) effectively re-introduces the $\psi_N$ conformal family into the theory, from outside the Kac table.

By rotating the polygon diagram for $\Pi_{C_N}$, we generate polygon diagrams for $N-1$ more connectivity weights, and with slight adaptations, the same arguments that we presented above show that none of these connectivity weights are in the span of $\mathcal{B}_N$.  So in addition to (\ref{fusions}), the $N$-leg boundary conformal family is initially absent from any fusion product of $N$ adjacent one-leg boundary operators.  (Here, we consider $\psi_1(x_1)$ and $\psi_1(x_{2N})$ to be adjacent.)  Although it might seem that one must insert all $N$ of these weights into $\mathcal{B}_N$ to restore this missing conformal family to all of these products, this is not true.  Indeed, $\text{span}\,(\mathcal{B}_N\cup\{\Pi_{C_N}\})=\mathcal{S}_N$ already includes all of the $N-1$ other weights (figure \ref{Contain}).  Hence, inserting only $\Pi_{C_N}$ into $\mathcal{B}_N$ restores the $\psi_N$ conformal family to not just the fusion products in (\ref{fusions}) but also to the fusion product of all $2N$ available collections of $N$ adjacent one-leg boundary operators in (\ref{corrfunc}).

\section{Summary}

Using our previous results from \cite{florkleb,florkleb2,florkleb3} in this article, we state and prove some additional facts concerning elements of the solution space $\mathcal{S}_N$ for the system of $2N+3$ PDEs (\ref{nullstate}, \ref{wardid}) in $2N$ variables $x_1,$ $x_2,\ldots,x_{2N}$ that govern a conformal field theory (CFT) correlation function (\ref{corrfunc}) of $2N$ distinct one-leg boundary operators and that also appear in multiple Schramm-L\"owner evolution (SLE$_{\kappa}$) with parameter $\kappa\in(0,8)$.  In section \ref{frobsect}, we prove theorem \ref{frobseriescor}, which states that if $8/\kappa\not\in2\mathbb{Z}^++1$, then any element of $\mathcal{S}_N$ equals a sum of at most two Frobenius series in powers of $x_{i+1}-x_i$ for any $i\in\{1,2,\ldots,2N-1\}$.  And in appendix \ref{appendix}, we prove that if $8/\kappa\in2\mathbb{Z}^++1$, then a Frobenius series multiplied by $\log(x_{i+1}-x_i)$ may appear.  This establishes part of the operator product expansion (OPE) assumed in CFT.  In section \ref{xingprob}, we identify special elements of $\mathcal{S}_N$ called \emph{connectivity weights} (definition \ref{dualbasis}), which we hypothesized to exist in \cite{florkleb}.   Theorem \ref{xingasymplem} establishes their essential properties.  Then, we use these special functions in conjecture \ref{connectivityconj} to propose a formula (\ref{xing}, \ref{xing2}) for the \emph{crossing probability} that the $2N$ curves of a multiple-SLE$_\kappa$ process join pairwise in a particular connectivity. (In a forthcoming article \cite{fkz}, we use this formula to calculate probabilities of cluster-crossing events of critical lattice models in polygons.)  In section \ref{intervalsect}, we classify intervals $(x_i, x_{i+1})$ for any $F\in\mathcal{S}_N$ in two ways.  First, if a single boundary arc generated by a multiple-SLE$_\kappa$ process with partition function $F$ almost surely shares (resp.\ doesn't share) its endpoints with the interval, then we call $(x_i,x_{i+1})$ a \emph{contractible interval} (resp.\ \emph{propagating interval}) of $F$.  Second, if the one-leg boundary operators of the CFT correlation function $F$ (\ref{corrfunc}) at the interval's endpoints have only the identity family (resp.\ only the two-leg family) in their operator product expansion (OPE), then we call $(x_i,x_{i+1})$ an \emph{identity interval} (resp.\ \emph{two-leg interval}) of $F$.  Lemma \ref{proptwoleglem} states that propagating intervals and two-leg intervals are identical.  Meanwhile, we find that contractible intervals and identity intervals are different.  Indeed, (\ref{finallincmb}) and figure \ref{TwoLegFuse} show how an identity interval, after we ``insert" it into the domain of a connectivity weight, decomposes into a linear superposition of a contractible interval and propagating intervals.  Finally, in section \ref{minmodelsect}, we explore the connection between the SLE$_\kappa$ exceptional speeds (\ref{exceptional}) and the CFT minimal models.  In particular, we propose conjecture \ref{minmodelconj} as an explanation for this connection, and examine its application to several particular cases.

During the writing of this article, we learned that K.\ Kyt\"ol\"a and E.\ Peltola recently obtained results very similar to ours by using a completely different approach based on quantum group methods \cite{kype,kype2}.

\section{Acknowledgements}

We thank J.\ J.\ H.\ Simmons and K.\ Kyt\"ol\"a for insightful conversations, and we thank C.\ Townley Flores for carefully proofreading the manuscript.

This work was supported by National Science Foundation Grants Nos.\ PHY-0855335 (SMF) and DMR-0536927 (PK and SMF).

\appendix{}

\section{Proof of item \ref{frobitem3} in theorem \ref{frobseriescor}}\label{appendix}

In this appendix, we prove item \ref{frobitem3} of theorem \ref{frobseriescor} first for the elements of the basis $\mathcal{B}_N^{\scaleobj{0.75}{\bullet}}$ (\ref{BN'}), and then for all elements of $\mathcal{S}_N$.  In this situation, $8/\kappa\in2\mathbb{Z}^++1$, and the fugacity function (\ref{LkFk}) vanishes: $n(\kappa)=0$.  This appendix presumes familiarity with the notations and results of sections \red{A 1}--\red{A 3} in \cite{florkleb3}.

\begin{figure}[b]
\centering
\includegraphics[scale=0.3]{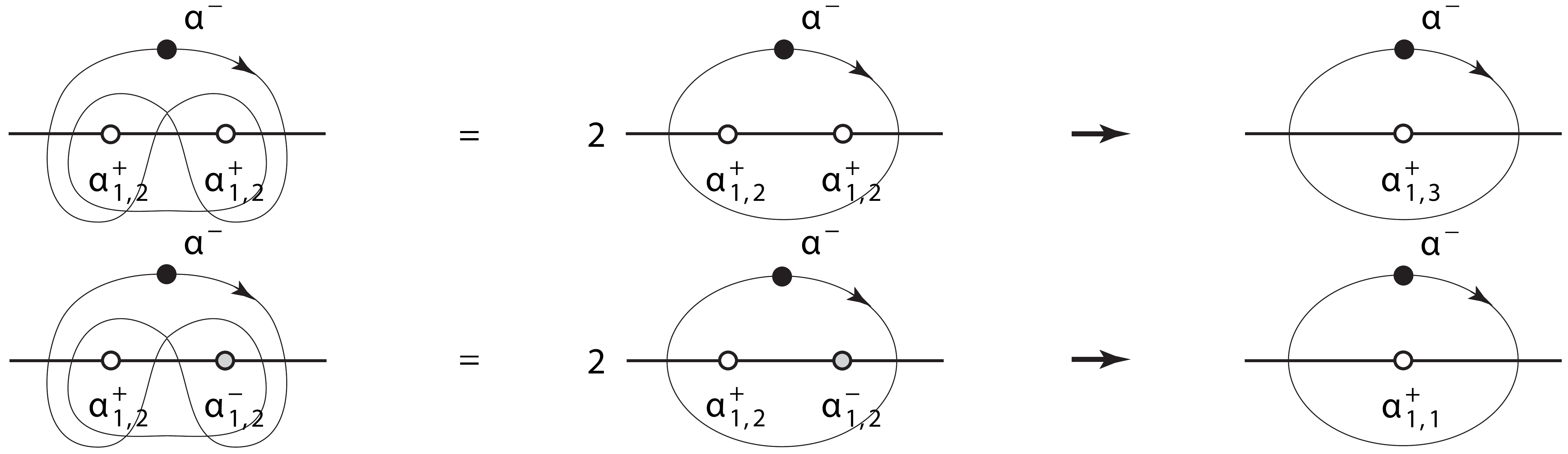}
\caption{If $8/\kappa$ is odd, then the Pochhammer contour entwining the two chiral operators (section \ref{purecft}) decouples into two clockwise simple loops, and the screening operator tracing its path is not drawn in with the fusion of those two chiral operators.}
\label{LoopDec}
\end{figure}

In the formula (\ref{Fexplicit}, \ref{BN'}) for $\mathcal{F}_\vartheta^{\scaleobj{0.75}{\bullet}}\in\mathcal{B}_N^{\scaleobj{0.75}{\bullet}}$, the integration contours interact with the points $x_i$ and $x_{i+1}$ in one of the three ways listed as cases \ref{sc2}, \ref{sc3}, and \ref{sc4} in section \ref{frobsect} and in appendix \red{A} of \cite{florkleb3}, where $x_i$ and $x_{i+1}$ appear in the Frobenius series expansions of theorem \ref{frobseriescor}.  In the work below and without loss of generality, we always choose $c\in\{1,2,\ldots,2N\}$ in this formula such that three things are true.  First, $c\not\in\{i,i+1\}$.    Second, the formula for $\mathcal{F}^{\scaleobj{0.75}{\bullet}}_\vartheta$ does not fall under case \ref{sc4}.  And third, if the formula for $\mathcal{F}^{\scaleobj{0.75}{\bullet}}_\vartheta$ falls under case \ref{sc3}, then $x_i$ but not $x_{i+1}$ is an endpoint of an integration contour.  These choices simplify our calculations and exposition.

Throughout this appendix, we index the arc connectivities according to item \ref{indexorder} above definition \ref{sleintervaldefn}.  As a result, the decomposition of $F\in\mathcal{S}_N$ over $\mathcal{B}_N^{\scaleobj{0.75}{\bullet}}$ if $8/\kappa\in2\mathbb{Z}^++1$ assumes a form similar to (\ref{decompF}),
\be\label{decompF'}F=\underbrace{a_1\mathcal{F}_1^{\scaleobj{0.75}{\bullet}}+a_2\mathcal{F}_2^{\scaleobj{0.75}{\bullet}}+\dotsm+a_{C_{N-1}}\mathcal{F}_{C_{N-1}}^{\scaleobj{0.75}{\bullet}}}_{\text{case \ref{sc2} terms}}+\underbrace{a_{C_{N-1}+1}\mathcal{F}_{C_{N-1}+1}^{\scaleobj{0.75}{\bullet}}+a_{C_{N-1}+2}\mathcal{F}_{C_{N-1}+2}^{\scaleobj{0.75}{\bullet}}+\dotsm+a_{C_N}\mathcal{F}_{C_N}^{\scaleobj{0.75}{\bullet}}}_{\text{case \ref{sc3} terms}}\ee
for some real constants $a_1$, $a_2,\ldots,a_{C_N}$.  Also throughout this appendix, we let $\bar{\ell}_1$ be the limit (\ref{lim}) that acts on $F\in\mathcal{S}_N$ by collapsing the interval $(x_i,x_{i+1})$.

\subsection{Proof of item \ref{frobitem3} in theorem \ref{frobseriescor} for $\mathcal{B}_N$, case \ref{sc2}}\label{s2}

Supposing that $\mathcal{F}_\vartheta^{\scaleobj{0.75}{\bullet}}$ falls under case \ref{sc2} (so $\vartheta\leq C_{N-1}$), we prove that it has a Frobenius series expansion of the form (\ref{log}).  In this case, a Pochhammer contour entwines $x_i$ with $x_{i+1}$, and as this contour circles counterclockwise once around either one of these points, the integrand of (\ref{Fexplicit}) acquires a phase factor of $e^{2\pi i(-4/\kappa)}=-1$ because $8/\kappa$ is odd.  As a result, we see from (\red{27}) of \cite{florkleb3} with $\beta_i=\beta_j=-4/\kappa$ that the Pochhammer contour decomposes into two clockwise simple loops $\Gamma_0$, both winding once around $x_i$ and $x_{i+1}$ (figure \ref{LoopDec}).  After expanding all factors in the formula (\ref{Fexplicit}, \ref{BN'}) for $\mathcal{F}_\vartheta^{\scaleobj{0.75}{\bullet}}$ that are analytic at $x_{i+1}=x_i$ in a Taylor series centered on $x_i$, we find
\be\label{leftout3} \mathcal{F}_\vartheta^{\scaleobj{0.75}{\bullet}}(\kappa\,|\,\boldsymbol{x})\,\,=\,\,(x_{i+1}-x_i)^{2/\kappa}\,\,\dotsm\,\,\times 2\oint_{\Gamma_0}\mathcal{N}\Big[\,(u_1-x_i)^{-8/\kappa}\,\dotsm\,\Big]\,{\rm d}u_1\,\,+\,\,O((x_{i+1}-x_i)^{2/\kappa+1}),\ee
where the ellipses represent the expanded factors.  Those factors in the integrand are analytic at $u_1=x_i$, and because $-8/\kappa\in\mathbb{Z}^-$, we may use the Cauchy integral formula to evaluate the leading term of (\ref{leftout3}), finding that it is not zero.  Thus, we see that $\mathcal{F}_\vartheta^{\scaleobj{0.75}{\bullet}}$ equals a Frobenius series of the form (\ref{log}) with $A_m=C_m=0$ for all $m\in\mathbb{Z}^+\cup\{0\}$ and $B_0\neq0$. 

We also conclude from this result that if $8/\kappa\in2\mathbb{Z}^++1$ and $\vartheta\leq C_{N-1}$, then $(x_i,x_{i+1})$ is a two-leg interval of $\mathcal{F}_\vartheta^{\scaleobj{0.75}{\bullet}}$  (item \ref{cftinterval1} of definition \ref{cftintervaldefn}).  This contrasts with our findings in section \ref{purecft} for the case $8/\kappa\not\in2\mathbb{Z}^++1$.  There, we concluded from fusion rules (\ref{prodfuse2}, \ref{prodfuse3}) that if $8/\kappa\not\in2\mathbb{Z}^++1$ and $\vartheta\leq C_{N-1}$, then $(x_i,x_{i+1})$ is an identity interval of $\mathcal{F}_\vartheta$.

\subsection{Proof of item \ref{frobitem3} in theorem \ref{frobseriescor} for $\mathcal{B}_N$, case \ref{sc3}}\label{s3}

Supposing that $\mathcal{F}_\vartheta^{\scaleobj{0.75}{\bullet}}$ falls under case \ref{sc3} (so $\vartheta>C_{N-1}$), we prove that it has a Frobenius series expansion of the form (\ref{log}).  To do this, we decompose $\Gamma_1$ in formula (\ref{Fexplicit}, \ref{BN'}) into two contours
\be \Gamma_1=\mathscr{P}(x_{i-1},x_i)+\Gamma_1'\ee
as per item \red{3} of the proof of lemma \red{6} in \cite{florkleb3}, where $\Gamma_1'$ is an integration contour with neither endpoint being $x_i$ or $x_{i+1}$.  In section \ref{frobsect} and above theorem \ref{frobseriescor}, we note that the term generated from integrating along $\Gamma_1'$, called a ``case \red{1} term," contributes to the second sum in (\ref{log}).  Hence, all that remains to proving that $\mathcal{F}_\vartheta^{\scaleobj{0.75}{\bullet}}$ has the form (\ref{log}) is to show that the term generated from integrating along $\mathscr{P}(x_{i-1},x_i)$, called a ``case \ref{sc3} term" in section \ref{frobsect}, has this form too. (If $i=1$, then we identify $i-1=0$ with $2N$.)

To prove this, we repeat the analysis in section \red{A 3} of \cite{florkleb3} up to (\red{A21}).  Summarizing the main steps, we press the integration contours $\Gamma_2,$ $\Gamma_3,\ldots,\Gamma_{N-1}$ onto the real axis and write 
\be\label{Fform}\mathcal{F}_\vartheta^{\scaleobj{0.75}{\bullet}}\,=\,\dotsm\,\times\oint_{\Gamma_{N-1}}\dotsm\oint_{\Gamma_3}\oint_{\Gamma_2}\,\dotsm\,\bigg[I_{i-1}+\oint_{\Gamma_1'}\,\dotsm\,{\rm d}u_1\bigg]\,{\rm d}u_2\,{\rm d}u_3\,\dotsm\,{\rm d}u_{N-1},\ee
where the ellipses not appearing between integrals or integration measures stand for the factors explicitly shown in (\ref{Fexplicit}) (with a factor of $n(\kappa)$ dropped, as per (\ref{BN'})), and where for appropriate powers $\beta_j$,
\be\label{Ikintegrals}
I_k(x_1,x_2,\ldots,x_K):=\frac{1}{4\sin\pi\beta_k\sin\pi\beta_{k+1}}\sideset{}{_{\mathscr{P}(x_k,x_{k+1})}}\oint\mathcal{N}\Bigg[\prod_{j=1}^K(u_1-x_j)^{\beta_j}\Bigg]\,{\rm d}u_1,\quad x_{K+1}:=x_1,\quad K:=3N-2,\ee
with $x_1<x_2<\ldots<x_K$ simply the points $x_1$, $x_2,\ldots,x_{2N}$, $u_2$, $u_3,\ldots,u_{N-1}$ re-indexed in increasing order.  After identifying $I_{i-1}$ with the contour integral along $\Gamma_1$ in (\ref{Fexplicit}), we find that
\be\begin{gathered}\label{applicationsc3} 
s:=\sideset{}{_{j=1}^K}\sum\beta_j=-2,\qquad\text{$\beta_j\in\{-4/\kappa,8/\kappa,12/\kappa-2\}$ for all $j\in\{1,2,\ldots,K\},$}\\
\quad\beta_i=\beta_{i+1}=-4/\kappa,\qquad\beta_{i-1},\beta_{i+2}\in\{-4/\kappa,12/\kappa-2\}\quad (\beta_0:=\beta_K, \beta_{K+1}:=\beta_1,\beta_{K+2}:=\beta_2).
\end{gathered}\ee
To find a Frobenius series expansion with the form (\ref{log}) for $\mathcal{F}_\vartheta$, we express $I_{i-1}$ as a linear combination of $I_k$ with $k\neq i\pm1$.  We do this in section \red{A 3} of \cite{florkleb3}, finding (for $i\not\in\{1,2,K-1,K\}$, but see section \red{A 5 c} of \cite{florkleb3} if otherwise)
\be\label{result}I_{i-1}=\frac{1}{\sin\pi(\beta_i+\beta_{i+1})}\Bigg[-\sum_{k=1}^{i-2}\left(\sin\pi\sum_{l=k+1}^{i+1}\beta_l\right)I_{k}+\sum_{k=i+2}^K\left(\sin\pi\sum_{l=i+2}^k\beta_l\right)I_{k}-\sin\pi\beta_{i+1}I_{i}\Bigg].\ee
Now, the right side of (\ref{result}) may seem to have a simple pole at $\kappa=8/r$ with $r>1$ odd because $\beta_i+\beta_{i+1}=-r$, thanks to (\ref{applicationsc3}).  However, because neither $\beta_i=-r/2$ nor $\beta_{i+1}=-r/2$ is a pole of $I_{i-1}$, both sides of (\ref{result}) must be analytic at this $\kappa$.  We conclude that $\kappa$ is a zero of the bracketed factor in (\ref{result}).  Hence, to find the value of $I_{i-1}(\kappa)$, we expand the bracketed factor and denominator on the right side of (\ref{result}) to first order in $\varkappa-\kappa$ and send $\varkappa\rightarrow\kappa$, finding
\be\label{Ii}\begin{aligned}I_{i-1}(\kappa)=&-\left(\frac{\kappa^2}{8\pi}\right)\sum_{\,\,\,k=1\,\,\,}^{i-2}\left[\sin\left(\pi\sum_{l=k+1}^{i+1}\beta_l(\kappa)\right)\partial_\varkappa I_k(\kappa)+\cos\left(\pi\sum_{l=k+1}^{i+1}\beta_l(\kappa)\right)\left(\pi\sum_{l=k+1}^{i+1}\partial_\varkappa\beta_l(\kappa)\right)I_k(\kappa)\right]\\
&-\left(\frac{\kappa^2}{8\pi}\right)\sum_{k=i+2}^K\left[\sin\left(\pi\sum_{l=i+2}^k\beta_l(\kappa)\right)\partial_\varkappa I_k(\kappa)+\cos\left(\pi\sum_{l=i+2}^k\beta_l(\kappa)\right)\left(\pi\sum_{l=i+2}^k\partial_\varkappa\beta_l(\kappa)\right)I_k(\kappa)\right]\\
&-\left(\frac{\kappa^2}{8\pi}\right)\sin\left(\frac{4\pi}{\kappa}\right)\partial_\varkappa I_i(\kappa),\quad8/\kappa\in2\mathbb{Z}^++1.\end{aligned}
\ee
Because the terms with summations do not involve a contour with an endpoint at $x_i$ or $x_{i+1}$, these terms are analytic at $x_{i+1}=x_i$.  Multiplied by the factor $(x_{i+1}-x_i)^{2/\kappa}$ in (\ref{Fexplicit}), they contribute to the second sum in (\ref{log}).
 
The behavior of the last term in (\ref{Ii}) as $x_{i+1}\rightarrow x_i$ is more complicated and interesting.  After inserting $u_1(t)=(1-t)x_i+tx_{i+1}$ into $I_i$, differentiating it with respect to $\varkappa$, and setting $\varkappa=\kappa$ in the result, we find
\begin{multline}\label{partialI}-\left(\frac{\kappa^2}{8\pi}\right)\sin\left(\frac{4\pi}{\kappa}\right)\partial_\varkappa I_i(\kappa\,|\,x_1,x_2,\ldots,x_K)=\\
\begin{aligned}&-\frac{1}{\pi}\sin\left(\frac{4\pi}{\kappa}\right)\log(x_{i+1}-x_i)\,I_i(\kappa\,|\,x_1,x_2,\ldots,x_K)-\left(\frac{\kappa^2}{8\pi}\right)\sin\left(\frac{4\pi}{\kappa}\right)(x_{i+1}-x_i)^{1-8/\kappa}\\
&\times\,\,\partial_\varkappa\Bigg(\frac{1}{4\sin^2(4\pi/\varkappa)}\oint_{\mathscr{P}(0,1)}t^{-4/\varkappa}(1-t)^{-4/\varkappa}\mathcal{N}\Bigg[\prod_{j\neq i,i+1}^K(x_j-x_i-(x_{i+1}-x_i) t)^{\beta_j(\varkappa)}\Bigg]\,{\rm d}t\Bigg)_{\varkappa=\kappa}.\end{aligned}\end{multline}
Now, the definite integral $I_i$ in the first term on the right side falls under case \ref{sc2}.  Multiplied by the factor of $\log(x_{i+1}-x_i)$ and then by the factor of $(x_{i+1}-x_i)^{2/\kappa}$ in (\ref{Fexplicit}), this term contributes to the last sum in (\ref{log}).  Next, the derivative in the second term is analytic at $x_{i+1}=x_i$.  Multiplied by the factor of $(x_{i+1}-x_i)^{1-8/\kappa}$ in (\ref{partialI}) and then by the factor of $(x_{i+1}-x_i)^{2/\kappa}$ in (\ref{Fexplicit}), this second term equals a Frobenius series in powers of $x_{i+1}-x_i$ and with indicial power $1-6/\kappa$.  Because $8/\kappa=r>1$ is odd, the difference $r-1$ between this power and the previous $2/\kappa$ is a positive integer.  Hence, this second term contributes to both the first and second sum in (\ref{log}).  We conclude that $\mathcal{F}_\vartheta^{\scaleobj{0.75}{\bullet}}$ equals a sum of Frobenius series of the form (\ref{log}) if $\vartheta>C_{N-1}$.

For use in section \ref{lastsect} below, we employ (\ref{partialI}) to determine the asymptotic behavior of the last term in (\ref{Ii}).  Indeed, it is asymptotically dominant over the other terms as $x_{i+1}\rightarrow x_i$ and thus the only term on the right side of (\ref{Ii}) that contributes to the limit $\bar{\ell}_1\mathcal{F}_\vartheta^{\scaleobj{0.75}{\bullet}}$, where $\bar{\ell}_1$ (\ref{lim}) collapses the interval $(x_i,x_{i+1})$.  We find that
\begin{multline}\label{result5}-\left(\frac{\kappa^2}{8\pi}\right)\sin\left(\frac{4\pi}{\kappa}\right)\partial_\varkappa I_i(\kappa\,|\,x_1,x_2,\ldots,x_K)\underset{x_{i+1}\rightarrow x_i}{\sim}\\
-\left(\frac{\kappa^2}{8\pi}\right)\sin\left(\frac{4\pi}{\kappa}\right)(x_{i+1}-x_i)^{1-8/\kappa}\,\,\partial_\varkappa\Bigg(\frac{\Gamma(1-4/\varkappa)^2}{\Gamma(2-8/\varkappa)}\,\,\mathcal{N}\Bigg[\prod_{j\neq i,i+1}^K(x_j-x_i)^{\beta_j(\varkappa)}\Bigg]\Bigg)_{\varkappa=\kappa}\end{multline}
after recognizing the beta function in (\ref{partialI}) that follows from setting $x_{i+1}=x_i$.  (See (\red{43}) of \cite{florkleb3}.)  After noting that $\Gamma(1-4/\kappa)^2/\Gamma(2-8/\varkappa)$ vanishes as $\varkappa\rightarrow\kappa=8/r$ with $r>1$ odd, inserting the identity
\be\label{anid}-\left(\frac{\kappa^2}{8\pi}\right)\sin\left(\frac{4\pi}{\kappa}\right)\lim_{\varkappa\rightarrow\kappa}\partial_\varkappa\left(\frac{\Gamma(1-4/\varkappa)^2}{\Gamma(2-8/\varkappa)}\right)
=\lim_{\varkappa\rightarrow\kappa}\left(\frac{\Gamma(1-4/\varkappa)^2}{n(\varkappa)\Gamma(2-8/\varkappa)}\right),\quad 8/\kappa\in2\mathbb{Z}^++1,\ee
(with $n(\varkappa)$ given in (\ref{LkFk})) and recalling once again that $\partial_\varkappa I_i(\kappa)$ is asymptotically dominant over all of the other terms on the right side of (\ref{Ii}), we find that for $8/\kappa\in2\mathbb{Z}^++1$,
\be\label{result6}I_{i-1}(x_1,x_2,\ldots,x_K)\underset{x_{i+1}\rightarrow x_i}{\sim}\lim_{\varkappa\rightarrow\kappa}\left(\frac{\Gamma(1-4/\varkappa)^2}{n(\varkappa)\Gamma(2-8/\varkappa)}\right)(x_{i+1}-x_i)^{1-8/\kappa}\,\,\mathcal{N}\Bigg[\prod_{j\neq i,i+1}^K(x_j-x_i)^{\beta_j}\Bigg].\ee
In section \red{A 3} of \cite{florkleb3}, we found the asymptotic behavior of $I_{i-1}(x_1,x_2,\ldots,x_K)$ as $x_{i+1}\rightarrow x_i$ for $\kappa\in(0,8)$ with $8/\kappa\not\in\mathbb{Z}^+$.  Not surprisingly, we see that (\ref{result6}) matches the previous result (\red{A26}) of \cite{florkleb3} after we insert $\beta_i=\beta_{i+1}=-4/\kappa$ (\ref{applicationsc3}) into the latter and divide both of its sides by $4e^{\pi i(\beta_{i-1}-\beta_i)}\sin\pi\beta_{i-1}\sin\pi\beta_i$ (thus generating the prefactor in (\ref{Ii})).  Thus, the same main result of section \red{A 3} in \cite{florkleb3} holds for $\bar{\ell}_1\mathcal{F}_\vartheta^{\scaleobj{0.75}{\bullet}}$.  That is to say, $\bar{\ell}_1\mathcal{F}_\vartheta^{\scaleobj{0.75}{\bullet}}$ equals the element of $\mathcal{B}_{N-1}^{\scaleobj{0.75}{\bullet}}$ generated from the formula for $\mathcal{F}_\vartheta^{\scaleobj{0.75}{\bullet}}$ by dropping all factors involving $x_i$, $x_{i+1}$, and $u_1$, dropping the integration along $\Gamma_1$, and reducing the power $N-1$ of the prefactor in (\ref{Fexplicit}) by one.  In the notation of section \ref{intervalsect}, we have
\be\label{G'}\bar{\ell}_1\mathcal{F}_\vartheta^{\scaleobj{0.75}{\bullet}}=\mathcal{G}^{\scaleobj{0.75}{\bullet}}_{\chi(\vartheta)},\quad\mathcal{G}^{\scaleobj{0.75}{\bullet}}_{\chi(\vartheta)}(\kappa):=\lim_{\varkappa\rightarrow\kappa}n(\varkappa)^{-1}\mathcal{G}_{\chi(\vartheta)}(\varkappa)\in\mathcal{B}_{N-1}^{\scaleobj{0.75}{\bullet}},\quad\vartheta>C_{N-1,}\ee
where we define the function $\mathcal{G}_\vartheta$ above (\ref{Xidecomp}) and $\chi$ is the index map defined in item \ref{cutmap} above definition \ref{sleintervaldefn}.

\subsection{Proof of item \ref{frobitem3} in theorem \ref{frobseriescor} for $\mathcal{S}_N$}\label{lastsect}

In this section, we finish proving item \ref{frobitem3} of theorem \ref{frobseriescor}.  In sections \ref{s2} and \ref{s3} above, we prove that each $\mathcal{F}^{\scaleobj{0.75}{\bullet}}_\vartheta\in\mathcal{B}_N^{\scaleobj{0.75}{\bullet}}$ (\ref{BN'}) admits a Frobenius series expansion of the form (\ref{log}).  Because $\mathcal{B}_N^{\scaleobj{0.75}{\bullet}}$ is a basis for $\mathcal{S}_N$, it follows from the decomposition (\ref{decompF'}) that any element $F\in\mathcal{S}_N$ admits the Frobenius series expansion (\ref{log}) too.  Moreover, the analysis that precedes lemma \red{3} in \cite{florkleb} shows that the null-state PDE (\ref{nullstate}) centered on either $x_i$ or $x_{i+1}$ fixes the indicial powers in (\ref{log}).  Thus, if $A_0=0$ (resp.\ $C_0=0$), then $A_m=0$ for all $m\in\{0,1,\ldots,r-2\}$ (resp.\ $C_m=0$ for all $m\in\mathbb{Z}^+\cup\{0\}$), and if $A_0=B_0=0$, then $B_m=0$ for all $m\in\mathbb{Z}^+$.  Hence, to finish the proof of item \ref{frobitem3} in theorem \ref{frobseriescor}, we must show that $A_0=0$ if and only if $C_0=0$ and that the last series in (\ref{log}) with the logarithm factor dropped is in $\mathcal{S}_N$.

To prove that $A_0=0$ if and only if $C_0=0$ in (\ref{log}), we find an expression for $A_0$ by acting on both sides of the decomposition (\ref{decompF'}) with $\bar{\ell}_1$ (\ref{lim}).  At the end of section \ref{s2}, we note that $(x_i,x_{i+1})$ is a two-leg interval of $\mathcal{F}_\vartheta$, and therefore $\bar{\ell}_1\mathcal{F}_\vartheta^{\scaleobj{0.75}{\bullet}}=0$, if $\vartheta\leq C_{N-1}$ and $\kappa\in(0,8)$.  Using (\ref{G'}) for the remaining terms, we find
\be\label{A0} A_0=\bar{\ell}_1F\,\,\,=\,\,\,\sum_{\mathclap{\varrho=C_{N-1}+1}}^{C_N}\,\,\,a_\varrho\mathcal{G}^{\scaleobj{0.75}{\bullet}}_{\chi(\varrho)}=\sum_{\vartheta=1}^{C_{N-1}}\Bigg(\sum_{\substack{\varrho=C_{N-1}+1 \\ \chi(\varrho)=\vartheta}}^{C_N}a_\varrho\Bigg)\mathcal{G}^{\scaleobj{0.75}{\bullet}}_\vartheta.\ee
Next, we discard from (\ref{log}) all terms without a logarithm.  We do this for the elements of $\mathcal{B}_N^{\scaleobj{0.75}{\bullet}}$ (\ref{BN'}) first.  If $\varrho\leq C_{N-1}$, then according to section \ref{s2}, this expansion (\ref{log}) for $\mathcal{F}_\varrho^{\scaleobj{0.75}{\bullet}}$ has no logarithmic term.  Hence, this map sends
\be\label{replace0}\mathcal{F}_\varrho^{\scaleobj{0.75}{\bullet}}\,\,\xrightarrow[\hspace{1cm}]{}\,\,0,\quad\varrho\leq C_{N-1}.\ee
The integration along $\Gamma_1'$ in the expression (\ref{Fform}) for $\mathcal{F}_\varrho^{\scaleobj{0.75}{\bullet}}$ with $\varrho>C_{N-1}$ has no logarithmic term, but the contour integral $I_{i-1}$ (\ref{Ikintegrals}) does.  The logarithm appears only in the last term (\ref{partialI}) on the right side of (\ref{Ii}).  Isolating it sends 
\be\label{prereplace}I_{i-1}\,\,\xrightarrow[\hspace{1cm}]{}\,\,-\frac{1}{\pi}\sin\left(\frac{4\pi}{\kappa}\right)\log(x_{i+1}-x_i)I_i.\ee
Therefore, to isolate the logarithmic term in $\mathcal{F}_\varrho^{\scaleobj{0.75}{\bullet}}$ with $\varrho>C_{N-1}$, we drop the integration along $\Gamma_1'$ from (\ref{Fform}) and insert (\ref{prereplace}) into (\ref{Fform}).  With this, the case \ref{sc2} contour $\mathscr{P}(x_i,x_{i+1})$ of $I_i$ replaces the case \ref{sc3} contour of $I_{i-1}$ in (\ref{Fexplicit}), so
\be\label{replace}\mathcal{F}_\varrho^{\scaleobj{0.75}{\bullet}}\,\,\xrightarrow[\hspace{1cm}]{}\,\,-\frac{1}{\pi}\sin\left(\frac{4\pi}{\kappa}\right)\log(x_{i+1}-x_i)\mathcal{F}^{\scaleobj{0.75}{\bullet}}_{\chi(\varrho)},\quad\varrho>C_{N-1}\ee
(figure \ref{CutMap}).  Finally, we isolate the terms with a logarithm in the expansion (\ref{log}) for $F\in\mathcal{S}_N$ and discard all others.  To do this, we apply the map (\ref{replace0}, \ref{replace}) to each term in the decomposition (\ref{decompF'}) of $F$ over $\mathcal{B}_N^{\scaleobj{0.75}{\bullet}}$.  This sends
\be\label{replace2}F\,\,\xrightarrow[\hspace{1cm}]{}\,\,-\frac{1}{\pi}\sin\left(\frac{4\pi}{\kappa}\right)\log(x_{i+1}-x_i)\sum_{\vartheta=1}^{C_{N-1}}\Bigg(\sum_{\substack{\varrho=C_{N-1}+1 \\ \chi(\varrho)=\vartheta}}^{C_N}a_\varrho\Bigg)\mathcal{F}^{\scaleobj{0.75}{\bullet}}_\vartheta,\ee
where the right side equals the last sum in (\ref{log}).  Now, $C_m=0$ in (\ref{log}) for all $m\in\mathbb{Z}^+\cup\{0\}$ if and only if the right side of (\ref{replace2}) vanishes, and with $\mathcal{B}_N^{\scaleobj{0.75}{\bullet}}$ (\ref{BN'}) linearly independent, the latter happens if and only if the coefficient of $\mathcal{F}^{\scaleobj{0.75}{\bullet}}_\vartheta$ on the right side of (\ref{replace2}) vanishes for all $\vartheta\leq C_{N-1}$.  And with $\mathcal{B}_{N-1}$ (\ref{BN'}) linearly independent too, $A_0=0$ in (\ref{A0}) if and only if the coefficient of $\mathcal{G}^{\scaleobj{0.75}{\bullet}}_\vartheta$ on the right side of (\ref{A0}) vanishes for all $\vartheta\leq C_{N-1}$.  In the first paragraph of this section, we prove that $C_m=0$ for all $m\in\mathbb{Z}^+\cup\{0\}$ if and only if $C_0=0$.  Therefore,
\be\label{infer}A_0=0\quad\Longleftrightarrow\quad\sum_{\substack{\varrho=C_{N-1}+1 \\ \chi(\varrho)=\vartheta}}^{C_N}a_\varrho=0\,\,\,\text{for all $\vartheta\in\{1,2,\ldots,C_{N-1}\}$}\quad\Longleftrightarrow\quad C_m=0\quad\Longleftrightarrow\quad C_0=0.\ee

To conclude the proof of item \ref{frobitem3} in theorem \ref{frobseriescor}, we show that the last series in (\ref{log}) with the logarithm factor dropped is in $\mathcal{S}_N$.  But according to the previous paragraph, this series equals the right side of (\ref{replace2}) with $\log(x_{i+1}-x_i)$ dropped, which is evidently in $\mathcal{S}_N$.

\subsection{Identity intervals for $8/\kappa$ odd}\label{8kappaodd}

We now discuss the definition of an identity interval in the case $8/\kappa\in2\mathbb{Z}^++1$, which is not included in section \ref{purecft}.  Because of the logarithmic term in (\ref{log}), the analyticity condition on $H$ in item \ref{cftinterval2a} of definition \ref{cftintervaldefn} never holds, so this definition cannot be used.  In addition, we cannot adopt item \ref{cftinterval2b} of definition \ref{cftintervaldefn} either.  Indeed, its condition on the decomposition of $F\in\mathcal{S}_N$ (\ref{decompF'}), now over $\mathcal{B}_N^{\scaleobj{0.75}{\bullet}}$ (\ref{BN'}), implies that $a_\vartheta=0$ for all $\vartheta>C_{N-1}$.  But according to section \ref{s2}, $(x_i,x_{i+1})$ is a two-leg interval of all terms in (\ref{decompF'}) with $\vartheta\leq C_{N-1}$, and thus of $F$, rather than an identity interval of these functions.  Thus this condition is not useful in defining an identity interval of $F$.

The ineffectiveness of item \ref{cftinterval2b} if $8/\kappa\in2\mathbb{Z}^++1$ follows from the decomposition of the Pochhammer contour entwining $x_i$ with $x_{i+1}$ into two loops (figure \ref{LoopDec}).  Because they are not entangled with $x_i$ and $x_{i+1}$, these loops do not contract as we send $x_{i+1}\rightarrow x_i$, so the screening operator that traces them is not drawn in with the fusion. Hence, the rules (\ref{prodfuse3}, \ref{prodfuse4}) reset to their original versions (\ref{prodfuse1}, \ref{prodfuse2}) respectively, and the first of these (\ref{prodfuse3}) now fails to give an identity interval.  In addition, the fusion rule (\ref{prodfuse2}) does not apply if $8/\kappa$ is odd either.  Indeed, if we choose a formula for $\mathcal{F}_\vartheta^{\scaleobj{0.75}{\bullet}}$ with $\vartheta\leq C_{N-1}$ that has $c\in\{i,i+1\}$, then because $(x_i,x_{i+1})$ is a two-leg interval of $\mathcal{F}_\vartheta^{\scaleobj{0.75}{\bullet}}$ according to section \ref{s2}, we must find fusion rule (\ref{prodfuse1}) in place of (\ref{prodfuse2}).  Hence, we fail to construct an identity interval if $8/\kappa$ is odd.

Concepts from logarithmic CFT help to clarify this situation.  Indeed, the fact that $A_0=0$ if and only if $C_0=0$ in the expansion (\ref{log}), proved in item \ref{frobitem3} of theorem \ref{frobseriescor}, implies that the two-leg family, multiplied by a logarithm, always appears with the identity family in the OPE for $\psi_1(x_i)$ with $\psi_1(x_{i+1})$ if $8/\kappa$ is odd.  (Indeed, $(x_i,x_{i+1})$ is a two-leg interval of the sum on the right side of (\ref{replace2}).  Therefore, after identifying the logarithmic term in (\ref{log}) with the right side of (\ref{replace2}) in section \ref{lastsect}, we realize that only the two-leg family contributes to the series multiplying the logarithm in (\ref{log}).  This agrees with a similar observation in \cite{gurarie} regarding a certain CFT four-point function solving the system (\ref{nullstate}, \ref{wardid}) for $N=2$ and $\kappa=8$.)  In addition, the family of the logarithmic partner \cite{gurarie,gurarie2} to the two-leg boundary operator contributes to the second sum in (\ref{log}).  Because the identity family never appears alone in the OPE, $(x_i,x_{i+1})$ must not be an ``identity interval" of any $F\in\mathcal{S}_N$.  Hence, this term has no meaning for $8/\kappa$ odd.

In spite of this outcome, one might wonder if it is reasonable to extend the meaning of a ``pure interval" (other than a two-leg interval, which is still defined) to $8/\kappa$ odd by taking some particular combination of the three families appearing in the OPE of $\psi_1(x_i)$ with $\psi_1(x_{i+1})$.  Such a combination could give another kind of pure interval that, although not an identity interval, would play a similar role.  However, this does not seem plausible because logarithmic partners are defined only up to the addition of the CFT operator to which they couple \cite{gurarie2}.  To clarify, we suppose that $(x_i,x_{i+1})$ is this putative pure interval of some $F\in\mathcal{S}_N$ and is also a two-leg interval of some other $F_0\in\mathcal{S}_N$.  Adding $F_0$ to $F$,  this interval of $F$ is no longer  pure.  But then again, doing this only adds the boundary two-leg operator, generated by sending $x_{i+1}\rightarrow x_i$ in $F_0(\boldsymbol{x})$, to the logarithmic partner arising from $F(\boldsymbol{x})$ in the same limit, sending the latter operator to another logarithmic partner.  Thus, which particular combination of the sums in (\ref{log}) with $A_0,C_0\neq0$ should make $(x_i,x_{i+1})$ ``pure" seems arbitrary and therefore not useful from a CFT point of view.

\end{document}